\DocumentMetadata{
	lang=en-US,
	pdfstandard=a-4f,
}

\documentclass[reqno,centertags,11pt,a4paper,noamsfonts]{amsart}
\usepackage[american]{babel}
\usepackage[babel=true,expansion=alltext,protrusion=alltext-nott,final]{microtype}
\usepackage[margin={3cm},marginparwidth={2.5cm}]{geometry}
\usepackage[svgnames]{xcolor}
\usepackage[utf8]{inputenc}
\usepackage[T1]{fontenc}
\usepackage{textcomp}
\usepackage{libertine}
\usepackage[libertine]{newtxmath}
\useosf
\usepackage[varqu,varl]{inconsolata}
\usepackage{esint}
\usepackage{amsthm}
\usepackage{mathtools}

\usepackage{pgfplots}
\pgfplotsset{compat=1.15}
\usetikzlibrary{arrows}

\usepackage[inline]{enumitem}
\numberwithin{equation}{section}
\providecommand{\mr}[1]{\href{https://www.ams.org/mathscinet-getitem?mr=#1}{MR~#1}}

\providecommand{\arxiv}[1]{\href{https://arxiv.org/abs/#1}{arXiv:#1}}
\providecommand{\doi}[2]{\href{https://doi.org/#1}{#2}}
\newcommand{\RR}{\mathbb{R}}

\renewcommand{\P}{\mathbb{P}}

\DeclareMathOperator{\E}{\mathbb{E}}
\DeclareMathOperator{\R}{\mathbb{R}}
\DeclareMathOperator{\T}{\mathbb{T}}
\DeclareMathOperator{\N}{\mathbb{N}}
\DeclareMathOperator{\Z}{\mathbb{Z}}

\theoremstyle{plain}
\newtheorem{theorem}{Theorem}[section]
\newtheorem{proposition}[theorem]{Proposition}

\newtheorem{lemma}[theorem]{Lemma}

\theoremstyle{definition}
\newtheorem{definition}[theorem]{Definition}
\newtheorem{remark}[theorem]{Remark}

\usepackage{hyperref}
\hypersetup{linktoc=all,colorlinks=true,linkcolor=Green,citecolor=Orange,urlcolor=DarkBlue,
	pdftitle={Large scale regularity and correlation length for almost length-minimizing random curves in the plane},
	pdfauthor={T. Ried, C. Wagner}, 
	pdfsubject={},
	pdfkeywords={random curves, regularity theory, Random Field Ising Model, chaining method}}

\begin{document}
\title[Large scale regularity of random curves in the plane]{Large scale regularity and correlation length for almost length-minimizing random curves in the plane}

\author[T.~Ried]{Tobias Ried}
\address[T.~Ried]{School of Mathematics, Georgia Institute of Technology, Atlanta, GA 30332-0160, United States of America} 
\email{tobias.ried@gatech.edu}

\author[C.~Wagner]{Christian Wagner}
\address[C.~Wagner]{Max-Planck-Institut für Mathematik in den Naturwissenschaften, Inselstraße 22, 04103 Leipzig, Germany}
\curraddr{Institute for Science and Technology Austria (ISTA), Am Campus 1, 3400 Klosterneuburg, Austria}
\email{christian.wagner@ist.ac.at}

\date{Jun 24, 2026.}
\thanks{\textcopyright 2025 by the authors. Faithful reproduction of this article, in its entirety, by any means is permitted for noncommercial purposes.}
\begin{abstract}
We consider a model of random curves in the plane related to the large-scale behavior of the Random Field Ising Model (RFIM) at temperature zero in two space dimensions. 
Our work is motivated by attempts to quantify the Imry--Ma phenomenon concerning the rounding of the phase transition by quenched disorder, and connects to recent advances regarding the decay of correlations in the RFIM. 

We study a continuum model of minimal surfaces in two space dimensions subject to an external, quenched random field, and restrict ourselves to isotropic surface integrands. The random fields we consider behave like white noise on large scales with an ultra-violet regularization reminiscent of the lattice structure of the RFIM.

We give a finer description of the minimizer below the length scale $ L_* $ starting from which the influence of boundary conditions is suppressed with a given probability, which has recently been shown to satisfy $ \log L_* \sim \varepsilon^{-\frac{4}{3}} $ in the amplitude $\varepsilon>0$ of the noise. 
More precisely, we prove flatness of the phase boundaries on scales $ L $ up to $ \log L \lesssim \varepsilon^{-\frac{4}{13}} $.
\end{abstract}

\maketitle

\section{Introduction}

It is a classical prediction by Imry and Ma \cite{IM75} that the Random Field Ising Model (RFIM) in the critical case of two space dimensions, where the boundary effects and fluctuations of the random field are of the same order, does not exhibit long-range order. This effect is due to the quenched disorder and it persists even without thermal fluctuations. The absence of a phase transition was rigorously established in \cite{AW90} for all temperatures, including the case of zero temperature. We will focus on the temperature zero case from now on and take a variational point of view: The (random) Gibbs measures are concentrated on local minimizers of the (random) Hamiltonian, and the results of \cite{AW90} show that there is an almost surely unique infinite volume minimizer.

The quantification of the Imry-Ma phenomenon has attracted the attention of the probability and statistical mechanics communities in recent years: it needs the collective influence from disorder on a large set to shield off the influence of the boundary data. This raises the question of how fast the boundary influence decays, which was first addressed in \cite{Chatterjee18}, and culminated in the works \cite{DX19,APH}. In the latter they establish exponential decay of correlations for the RFIM with independent Gaussian disorder. Following up on this, \cite{DW} addressed the correlation length for small amplitude noise. More precisely, they establish that the length scale $ L_* $ starting from which the influence of boundary conditions on the the spin at the origin is suppressed with a given probability is of order\footnote{Here and in the following we use the notation $ A \lesssim B $ whenever there exists a universal constant $ 0 < c < \infty $ such that $ A \leq c B $. If $ c = c ( \theta ) $ depends on some parameter $ \theta $, we write $ A \lesssim_{ \theta } B $. The symbol $ \gtrsim $ is accordingly used for lower bounds. We write $ A \sim B $ provided $ A \lesssim B $ and $ B \lesssim A $ hold.}
 $ \log L_* \sim \varepsilon^{ - \frac{4}{3} } $ in the amplitude $ \varepsilon > 0 $ of the noise, see \eqref{eqn:in-dw-bound} for the precise statement.\footnote{For a more detailed account of the literature we refer to Section \ref{sec:literature} below.}

In this article we aim to give a finer description of the minimizer below this length scale based on a continuum description of the model, giving us access to the regularity theory for minimal surfaces. For the moment we study this problem for an \emph{isotropic} energy rather than the RFIM itself.\footnote{See the discussion around equation \eqref{eqn:anisotropic} for the connection with the RFIM and potential generalizations of our main result to (slightly) anisotropic integrands.} More precisely, we are interested in the large-scale behavior of local, i.e.~subject to their own boundary conditions, minimizers of the energy functional
\begin{align}\label{eqn:in-energy}
E(m,U) = \int_{ \overline{U} } | \nabla m | - \varepsilon \int_{U} \xi m,
\end{align}
where $ m : \R^2 \rightarrow \{ -1, +1 \} $ is a function of (locally) bounded variation and $ U \subset \R^2 $. We normalize the noise such that $ \E \xi(x)^2 = 1 $, so that $ \varepsilon $ acts as the amplitude or variance of the noise. Furthermore, we assume that the noise $ \xi $ is a centered Gaussian field subject to a suitable ultra-violet (UV) cut-off. In view of the underlying unit size lattice in the RFIM, we choose to impose this regularization on scales of order one, c.f.~Section~\ref{sec:assumptions-on-noise} below for the precise assumptions. In particular, the energy \eqref{eqn:in-energy} is almost surely well-defined and admits well-defined local minimizers, see Definition~\ref{defn:in-energy-minimizer} for details. Note that due to the UV cut-off, we do not expect a universal behavior of minimizers on  scales smaller than the cut-off scale. Nevertheless, it is convenient to have access to small-scale regularity, which in particular implies that the jump set, i.e.\ the curve along which $m$ changes from $-1$ to $+1$, is regular, see Remark~\ref{rem:gmt} below.

Let us stress again that the scaling of the first term (the boundary integral) and the fluctuations of the second term (the random field) are of the same order exactly in dimension two, which is the focus of this work. From a statistical physics point of view, the addition of the random field is a substantial perturbation of the model with $ \varepsilon = 0 $: while not exactly in the framework of \cite{AW90}, a similar argument\footnote{We decided to not include the detailed proof in this article because our main result here is independent of almost sure uniqueness of the infinite-volume minimizer.} shows the almost sure uniqueness of minimizers for arbitrary $ \varepsilon > 0 $. Nonetheless, in the regime we are looking at, the functional \eqref{eqn:in-energy} may also be viewed as a (random) perturbation of the perimeter problem. In fact, for such length scales minimizers of \eqref{eqn:in-energy} are \emph{almost minimizers} of the perimeter problem, see Definition \ref{def:eta-minimizer} and Proposition \ref{proposition:eta-minimality}. 
From an analytic point of view, this suggests that minimizers inherit regularity, i.e. flatness, from the perimeter problem. 
A natural way of measuring this flatness is by looking at the variation of the normal $\nu(x)$ as $x$ varies over the jump set of a minimizer $m$. Due to the UV regularization, the normal is well-defined at every point in the jump set of $m$. However, let us emphasize that this does not capture the effective large-scale behavior of the model \eqref{eqn:in-energy}, since it relies on the non-universal (i.e., depending on the UV regularization) small-scale behavior. 

In order to obtain a universal large-scale regularity result, we introduce the notion of an average (on scales of order one) normal $ \bar \nu_{1} $ to the jump set of the minimizer: to this end, let $x$ be in the jump set of $m$ and define for $\int_{ B_1(x) } \nabla m \neq 0$, 
\begin{align*}
	\bar \nu_1 ( x ) \coloneq \left| \int_{ B_1(x) } \nabla m \, \right|^{-1} \, \int_{ B_1(x) } \nabla m .
\end{align*}
The integral has the geometric interpretation as integral average of the interior normal (pointing from $\{m=-1\}$ to $\{m=1\}$) along the jump set.
More importantly, the average normal is the best approximation of the classical normal on scales of order one, that is, a minimizer of 
	\begin{align}\label{eqn:average-normal}
	\min_{ \bar \nu \in \mathbb{S}^1 } \frac{1}{2} \int_{ B_1(x) } | \nu - \bar \nu |^2 | \nabla m | = \min_{ \bar \nu \in \mathbb{S}^1 } \bigg( \int_{ B_1(x) } | \nabla m | - \bar \nu \cdot  \int_{ B_1(x) } \nabla m \bigg),
	\end{align}
which is a natural quantity in regularity theory, called the \emph{excess}. 
While minimizers of this variational problem may not be unique (a typical example is provided by bubbles on scales less than one), it turns out that in the situation considered here, this is always the case. 
We prefer to work with the latter characterization of the average normal, because it provides a robust definition that is independent of the small-scale regularization.

In this language our main theorem is:

\begin{theorem}\label{thm:mr}
There exists a (universal) constant $ 0 < \delta < 1 $ such that for any $ 0 < p_0 < 1 $ there exist constants $ 0 < c_{  \rm exp } < 1 $ and $ C_{  \rm cont } < \infty $ such that with probability at least $ p_0 $ the following holds:

Let $ m $ be a local minimizer of the energy \eqref{eqn:in-energy}. If $ R \leq \exp ( c_{  \rm exp } \varepsilon^{ - \frac{4}{13} } ) $, and $ x, y \in B_{ \delta R } $ are contained in the jump set of $ m $, then the average normals $\bar\nu_1(x)$ and $\bar\nu_1(y)$ are well-defined and we have
\begin{align}\label{eqn:main}
| \bar \nu_{ 1 }( x ) - \bar \nu_{ 1 } ( y )  |
\leq C_{ \rm cont } \left( \varepsilon ( \log ( | x |_{+} + | y |_{ + } ) )^{ \frac{1}{2} } ( \log | x - y |_{+} )^{ \frac{11}{4} } \right)^{ \frac{1}{2} },
\end{align}
where\footnote{Of course, there some arbitrariness in the choice of $ 2 $; every number strictly larger than $ 1 $ would work.} $ | x |_{+} \coloneq \max \{ | x |, 2 \} $.
\end{theorem}

\medskip 

	Our main theorem can be interpreted in the following way: with overwhelming probability, the phase boundaries in a minimizing configuration are flat on scales with logarithm of order $ \varepsilon^{-\frac{4}{13}} $. The flatness is measured by comparing the average normals (in a version that does not depend on the underlying regularization of the noise) at two points $x,y$ in the jump set of the minimizer that have a distance $ \log |x-y| \ll \varepsilon^{-\frac{4}{13}} $, and shows that 
	\begin{align*}
		|\bar\nu_1(x) - \bar\nu_1(y)| \ll 1 ,
	\end{align*}
	i.e.~the oscillation of the average normal along the jump set of the minimizer $m$ is small. This extends the classical (but not universal) small scale regularity, which shows that the normal does not oscillate much on scales of order $ \ll 1 $.
	
	Our main result corresponds to a regularity theory under minimal assumptions and connects to the geometric notion of Reifenberg flatness of the phase boundary, see also \cite[Corollary 1.5]{GNR}. The idea of Reifenberg flatness, see \cite{Toro} for an overview of this concept and its connection to the geometry of measures, is also reflected in our proof strategy, which approximates the jump set of $m$ locally by half-spaces.

	Moreover, our work draws inspiration from classic large-scale regularity in stochastic homogenization, see e.g.~\cite{ArmstrongSmart,gloria_neukamm_otto}. Let us remark that there is an important conceptual difference to these works: In stochastic homogenization regularity emerges down from infinity to some critical length scale, since on large scales the problem is close to the homogenized problem, which admits a good regularity theory.\footnote{For a more in-depth discussion, we refer to \cite{gloria_neukamm_otto} and the monographs \cite{ArmstrongKuusiMourrat,armstrong_kuusi_book}.} In the work at hand, we understand large-scale regularity in the sense that regularity beyond scales of order one persists up to a critical length scale $ \log | x - y | \ll \varepsilon^{ - \frac{ 4 }{ 13 } } $.

	Note that the scale up to which our result holds is much smaller than the correlation length scale determined by Ding--Wirth \cite[Theorem 1.1]{DW}. We believe that the exponent $\frac{4}{13}$ in Theorem~\ref{thm:mr} is not optimal: 

	The factor $(\log R)^{\frac{1}{2}}$ in \eqref{eqn:main} is an artefact of our proof strategy\footnote{More specifically, the spatial pinning in Proposition~\ref{prop:concentration-in-space}.}. Moreover, it is not clear whether the power $\frac{11}{4}$ of the second logarithm is sharp: the core of our current proof strategy is an iteration argument over dyadic scales in combination with a height bound, which is sharp from a deterministic point of view, but may not be optimal for a (random) minimizer due to stochastic cancellations. This is also hinted at in the more recent works\footnote{Both articles were posted on the arXiv while this paper was under review} \cite[Corollary 2]{OPW26} and \cite[Theorem 1.2]{BEPx+1} on the geometrically linearized version of \eqref{eqn:in-energy}. Since our main result implies flatness on scales $ \log | x - y | \ll \varepsilon^{ - \frac{ 4 }{ 13 } } $, a natural step in the analysis of regularity properties of the jump set is to incorporate this refined analysis of the geometrically linearized problem. However, the state-of-the-art exponent $ \frac{ 4 }{ 13 } $ obtained here is already enough to yield interesting connections to the random matching problem, as recently announced in \cite{OP26}.

	At this point, the reader may wonder how the result in Theorem~\ref{thm:mr} depends on the probability $ p_0 $. In view of the above discussion, we believe that the question of finding an optimal exponent for our flatness length-scale is the bigger challenge at the moment; therefore, we did not attempt to optimize our method of proof in that regard. We give a hint at this dependence in Remark~\ref{rmk:dependence-on-p0} after the proof of Theorem~\ref{thm:mr}.

\medskip

As alluded to earlier, our work is  motivated by recent developments in the RFIM. Its Hamiltonian on a box $ \Lambda_L \subset \Z^2 $ of size $ L $ with boundary conditions $ \tau : \Lambda_L^c \to \{ -1, + 1\} $ is given by
\begin{align}\label{eqn:in-rfim}
H( \sigma, \Lambda_L ) = \sum_{ \substack{ x,y \in \Lambda_L \\ x \sim y } } | \sigma_x - \sigma_y |^2 + \sum_{ \substack{ x \in \Lambda_L, y \in \Lambda_L^c \\ x \sim y } } | \sigma_x - \tau_y |^2- \varepsilon \sum_{ x \in \Lambda_L } \xi_x \sigma_x
~~ { \rm for } ~~ \sigma : \Lambda_L \to \{ -1, + 1 \}.
\end{align}
To put our work into this context, let us recall the definition of the order parameter
\begin{align}\label{eqn:in-magnetization-rfim}
m( L ) = \frac{1}{2} \E[ \sigma^{ L, + }_0 - \sigma^{ L, - }_0 ] = \P \{ \sigma^{ L, + }_0 \neq \sigma^{ L, - }_0 \},
\end{align}
where $ \sigma^{ L, \pm } $ denotes the value of the ground state of \eqref{eqn:in-rfim} with boundary conditions $ \tau \equiv \pm 1 $ and $ \E $ is the quenched expectation, i.e.~w.r.t.~to the noise $ \{ \xi_x \}_{ x \in \Z^2 } $. Throughout this introduction, we will assume that $ \{ \xi_x \}_{ x \in \Z^2 } $ are i.i.d.~Gaussian random variables (discrete Gaussian white noise).

Uniqueness of the infinite volume ground state is implied by $ m(L) \rightarrow 0 $ as $ L \rightarrow \infty $, see \cite{AW90}. In fact, precise decay estimates are known: It was shown in \cite{APH,DX19} that $ m(L) $ decays exponentially,
\begin{align*}
m(L) \lesssim \exp( - c L ) \quad { \rm for ~ some ~ constant } ~ 0 < c < \infty.
\end{align*}
A more detailed account on the literature is given in Section \ref{sec:literature}. A natural question, which has not been fully understood yet, is at which length scale $ L $ one starts to observe the exponential decay. First bounds were proven in \cite[see $ \zeta_1 $ therein]{Bar-Nir}. However, the exact scaling remains to be understood, see \cite[Remark 1.5]{DW}.

A second notion of correlation length, which is better understood, asks for the order parameter $ m(L) $ to drop below a given threshold. Let us denote this length scale by $ L_* $. It has the natural interpretation that with overwhelming probability the value at the origin of any ground state configuration of \eqref{eqn:in-rfim} for $ L \geq L_* $ does not depend on the boundary condition $ \tau $.\footnote{This is a well-known consequence of the classical FKG inequality, which ensures that checking  $ \sigma^{ L, + }_0 = \sigma^{ L, - }_0  $ for the extremal boundary conditions $ \tau \equiv \pm 1 $ suffices to conclude $ \sigma^{ L, \tau } = \sigma^{ L, + }_0 = \sigma^{ L, - }_0 $ for any boundary condition $ \tau $, see e.g.~\cite[Section 2.2]{AP} for a recent account on this.} First bounds on that quantity were obtained in \cite[see $ \zeta_2 $ therein]{Bar-Nir}, while the optimal asymptotic behavior was obtained in \cite{DW}, where they prove that 
\begin{align}\label{eqn:in-dw-bound}
\exp ( c^{-1}  \varepsilon^{ -\frac{4}{3} } ) \leq L_* \leq \exp ( c  \varepsilon^{ -\frac{4}{3} } ) 
\quad { \rm for ~ some ~ constant } ~ 0 < c < \infty ~ { \rm independent ~ of } ~ \varepsilon.
\end{align}
Our motivation is to give a finer description of the infinite volume ground state on scales below $ L_* $. More precisely, rather than monitoring the spin at a single point, we are interested in describing the surface that is defined by the boundary between the $ +1 $ and $ -1 $ phase.

\medskip

Right now, we can only answer this question for the model \eqref{eqn:in-energy}. The crucial difference to the RFIM \eqref{eqn:in-rfim} lies in the isotropic surface integrand. Indeed, on length scales $ L \leq L_* $, we can view the energy \eqref{eqn:in-energy} as a perturbation of the perimeter. This suggests that minimizers of \eqref{eqn:in-energy} inherit some regularity from perimeter minimizers corresponding to $ \varepsilon = 0 $ in \eqref{eqn:in-energy}. In two dimensions these are known to be line arrangements.\footnote{The crux is to establish this under minimal regularity assumptions, see e.g.~\cite[Theorem 10.1]{FrankMorgan}.} Therefore, it is natural to measure regularity in terms of a modulus of continuity for the normal to the phase boundary. As explained above, the large-scale regularity statement in Theorem \ref{thm:mr} is phrased in terms of the average normal.

Apart from the difference in surface integrands, our Theorem \ref{thm:mr} is phrased in the same spirit as the result of \cite{DW}: With overwhelming probability, the behavior of the model is well-described up to a given length that grows exponentially as the noise strength becomes small. 
While the result of \cite{DW} is valid up to the larger length scale with logarithm of order $ \varepsilon^{-\frac{4}{3}} $, our Theorem~\ref{thm:mr} provides a finer description of the minimizer on smaller length scales with logarithm of order $ \varepsilon^{-\frac{4}{13}} $ as $\varepsilon \downarrow 0$.
The theorem also reveals one of the main benefits of working in a continuum setting: it allows us to study a finer notion of correlation length in the system related to geometric properties of the minimizer. This correlation length is given by the scale at which the phase interfaces cease to be nearly flat. Indeed, given two points $x, y$ in the jump set of a minimizing configuration $m$ with $|x-y| \gg 1$, then \eqref{eqn:main} implies that 
\begin{align*}
	\log \max \{ | x |_{ + } , | y |_{ + } \} \ll \varepsilon^{ - \frac{ 4 }{ 13 } } 
	\quad \Longrightarrow \quad
	| \bar\nu_1 ( x ) - \bar\nu_1 ( y ) | \ll 1,
\end{align*}
that is, the logarithm of this correlation length is at least of order $\varepsilon^{-\frac{4}{13}}$. 

It is natural to wonder whether our results can be generalized to anisotropic surface integrands, including the continuum limit (as the lattice spacing goes to zero) of the RFIM, corresponding to
\begin{align}\label{eqn:anisotropic}
	E_{RFIM}(m,U) = \int_{\overline{U}} \Phi(\nu) |\nabla m| - \varepsilon\int_{U} \xi m 
\end{align}
with $\Phi(\nu) = |\nu_1| + |\nu_2|$, see \cite[Theorem 4]{ABC06}. The treatment of anisotropic surface energies is a much more delicate question, see e.g.~\cite{TioneInvent} for recent advances and overview of the current state regarding the regularity theory of anisotropic minimal surfaces.
	
From a point of view of regularity theory for minimal surfaces and perturbations thereof, we profit from the fact that for the Imry-Ma phenomenon two dimensions are critical. For curves in two dimensional space, i.e.~surfaces of dimension and co-dimension one, the regularity theory is more explicit, which allows us to work with  ``logarithmic perturbations'' of minimal surfaces, see Proposition \ref{proposition:eta-minimality} for the precise statement. 
This turns out to be the right notion in view of the stochastic estimates on the noise (Proposition~\ref{thm:talagrand-bound}), which are very robust and originate from the discrete setting \cite{DW}.

Appealing to \cite{GNR}, we could have formulated our result for slightly anisotropic, but still uniformly elliptic surface integrands $\Phi$ in \eqref{eqn:anisotropic}. However, the anisotropy $\Phi(\nu) = |\nu_1| + |\nu_2|$ corresponding to the RFIM does not satisfy the strong uniform ellipticity and regularity assumptions of \cite[Definition 2.1]{GNR}. In fact, this makes it a very difficult problem from the point of view of (deterministic) minimal surface theory and currently seems to put an analogous result to our Theorem~\ref{thm:mr} for the energy \eqref{eqn:anisotropic} out of reach.

A detailed account on the  regularity theory for \eqref{eqn:in-energy} \emph{under minimal decay assumptions} is given in Section~\ref{sec:literature} below.

\subsection{Assumptions on the noise (UV cut-off)}\label{sec:assumptions-on-noise} 
Our main Theorem~\ref{thm:mr} covers the following two important model classes of random fields $\xi$, where for concreteness the correlation length is set to unity. 

\medskip

\paragraph{\textit{Regularized white noise}.} The first type of noise is regularized white noise, which we assume to be a stationary Gaussian noise with covariance $ c(x-y) = \E \xi(x) \xi(y) $. The regularization is conveniently imposed in Fourier-space via a UV cut-off, i.e.\
\begin{align*}
	\widehat{c} ( k ) = \frac{1}{2\pi} \boldsymbol{1}_{ \{ | k | \leq \sqrt{4\pi} \} } ( k ) .
\end{align*}
Note that since $ \E \xi(x)^2 = \frac{1}{2\pi} \int \widehat{c}( k ) dk $, we have $ \E \xi(x)^2 = 1 $.

\medskip

\paragraph{\textit{Discretized white noise.}} The second type of noise that we want to consider is a discretized white noise, which is a centered Gaussian noise that is stationary w.r.t.~$ \Z^2 $-shifts; more precisely,
\begin{align*}
	\xi = \sum_{ x \in \Z^2 } \xi_x \boldsymbol{1}_{ Q_{ 1 } ( x ) },
\end{align*}
where $ \{ \xi_x \}_{ x \in \Z^2 } $ are i.i.d.~Gaussians with mean zero and variance one, and, for $ R > 0 $, $ Q_{ R } ( x ) \subset \R^2 $ is an axis-parallel cube of side length $ R $ centered around $ x $. We will also use the notation $Q_R \coloneq Q_R(0)$.

\medskip

In Section \ref{sec:pf-concentration-arguments}, more precisely in Lemmas~\ref{lemma:cameron-martin-continuity} and \ref{lem:wn-sup-bound}, we will see that the common properties of these two centered Gaussian noises are
\begin{align}\label{eqn:noise-assumption}
	\sup_{ Q_1(x) } | \xi | < \infty ~ { \rm almost ~ surely, }
	\quad { \rm and } \quad
	\int \delta \xi(x) f(x) dx \lesssim | \delta \xi |_{ H } \Big( \int f^2 \Big)^{ \frac{1}{2} }
\end{align}
for every Cameron-Martin shift $ \delta \xi $ and Schwartz function $ f $, see Lemma \ref{lemma:cameron-martin-continuity} for details. Note that the former property guarantees that \eqref{eqn:in-energy} is well-defined in a pathwise sense. The latter is a convenient way to state the UV regularization in both models (encoded in their Cameron-Martin space $H$).

\subsection{Definition of local minimizers and their small-scale regularity.}

Our main result is concerned with the regularity of local minimizers to \eqref{eqn:in-energy}. 

\begin{definition}[Local minimizer]\label{defn:in-energy-minimizer}
	For a given realization $ \xi $ of the noise, we call a function of locally bounded variation $ m : \R^2 \rightarrow \{ -1, + 1 \} $ a \emph{local minimizer} of the energy \eqref{eqn:in-energy}, if
	\begin{align*}
		\int_{ \overline B } | \nabla m | - \varepsilon \int_{ B } \xi m \leq \int_{ \overline B } | \nabla m' | - \varepsilon \int_{ B } \xi m'
	\end{align*}
	for every competitor $ m' : \R^2 \rightarrow \{ -1, +1 \} $ such that $ m'=m $ in $\RR^2 \setminus B$ for some ball $ B $.
\end{definition}

At this point, a remark about the regularity of the objects that we consider is in order:
\begin{remark}\label{rem:gmt}
It is always possible to modify $m$ on a set of Lebesgue measure zero (which keeps the energy unchanged) so that the support of the measure $|\nabla m|$ is equal to $\partial\{m=+1\}$, see \cite[Proposition 12.19]{Maggi}.

Since we want to have access to the measure theoretic inner unit normal $\nu(x)$ of the set $\{m = +1\}$, we a priori have to work with the reduced boundary $ \partial^* \{ m = +1 \} \ni x$, see \cite[Definition 3.54]{AmbrosioFuscoPallara} or \cite[Section 15]{Maggi}. 
However, due to the UV regularization of the noise, in particular the $ L^{\infty} $-bound \eqref{eqn:noise-assumption}, we are in fact in a much smoother setting: 

Classical small-scale regularity implies that the reduced boundary $ \partial^* \{ m = +1 \} $ is a $ C^{\frac{3}{2}-} $-curve, see e.g.~\cite[Example 21.2 \& Theorem 21.8]{Maggi}. 
Moreover, since we are working in dimension $d=2$, a finer analysis shows that the singular set $\Sigma =\partial \{ m = +1 \} \setminus \partial^* \{ m = +1 \}$ is empty, that is,
\begin{align*}
	\partial^* \{ m = +1 \} = \partial \{ m = +1 \},
\end{align*}
see \cite[Theorem 28.1]{Maggi}.

In order to emphasize the symmetric role that the sets $\{m=+1\}$ and $\{m=-1\}$ play in our setting, we prefer to work with the jump set $J_m$ of $m$. In our regular setting, $J_m$ is equal to either of the boundaries considered above\footnote{while in general only $\partial^* \{m=+1\} \subset J_m$}. We refer to \cite[Example 3.68 \& Definition 3.67]{AmbrosioFuscoPallara} for details.
\end{remark}

\subsection{Related literature}\label{sec:literature}

In this subsection, we briefly review  important results related to our model. We start with an overview of the literature to the RFIM at zero temperature that motivates our analysis. Afterwards, we  review relevant literature from the regularity theory of perimeter minimizers, before we mention some related result on different models.

\medskip

\paragraph{\textit{The Imry-Ma phenomenon in the context of the Random Field Ising Model.}}

As explained in the introduction, the Imry-Ma phenomenon was first discussed in the physics paper \cite{IM75}. In case of the RFIM, their heuristic discussion was made rigorous in dimensions $ d \geq 3 $ by \cite{BricmontKupiainen88} and in dimension $ d \leq 2 $ by \cite{AW90}. We refer the interested reader to these two references for more links to the physics literature supporting the argument in \cite{IM75}.

The qualitative result \cite{AW90} was recently revisited with a particular focus on the zero temperature case with weak quenched disorder. First, double logarithmic decay of correlations was obtained in \cite{Chatterjee18}, which was then upgraded to algebraic decay with a small exponent in \cite{AP} and more recently culminated in the exponential decay of correlations, see \cite{DX19arxiv, DX19, APH}. For accounts of earlier results considering high temperature or strong disorder we refer to the literature review in \cite{AP}. Moreover, \cite{Chatterjee19} addressed the fluctuations of the free energy in the RFIM by proving a central limit theorem.

\medskip

\paragraph{\textit{Correlation length for the Random Field Ising Model.}}

In the context of the RFIM, multiple notions of correlation length were considered in \cite{Bar-Nir,DW}. Our main focus is on the length scale $ L_* $, cf.~\eqref{eqn:in-dw-bound} above, on which suboptimal bounds were obtained in \cite{Bar-Nir} and sharp bounds in \cite{DW}.

At zero temperature, the authors of \cite{DW} provide optimal upper bounds of order $ \exp( c \varepsilon^{ - \frac{4}{3} } ) $ for some $ c > 0 $ in the noise strength $ \varepsilon > 0 $. An extension to positive temperatures (including the critical temperature\footnote{Here, critical temperature refers to the critical temperature of the classical (deterministic) Ising model.}) is considered in \cite{DWextended,DWperturbative}.

\medskip

\paragraph{\textit{Regularity theory.}}

The classical regularity theory of minimal surfaces, see e.g.~\cite{Maggi} for an overview, is interested in regularity statements on small scales, which are non-universal in our setting. In the language of Definition \ref{def:eta-minimizer}, the work \cite{AP99} corresponds to the case where $ \eta \rightarrow 0 $ as $ R \rightarrow 0 $, while \cite{GNR} is concerned with the case that $ \eta $ has no decay properties. In the end, both works contain $ \varepsilon $-regularity results, meaning that they require initial smallness of some excess or flatness quantity to conclude. To bypass the initial smallness assumption, we rely on the fact that in our two-dimensional setting everything is more explicit, see Section \ref{sec:regularity-theory}.

\medskip

\paragraph{\textit{Related models.}}

Problems related to the Imry-Ma phenomenon on continuum models have already been considered in \cite{DirrOrlandi1,DirrOrlandi2}, where (non-)local Allen-Cahn energies with an external random field are studied. Their quenched random field is essentially what we call discretized white noise. In this setting, they obtain the same uniqueness result as \cite{AW90}.

The recent works \cite{RonDor,BEPx+1} are devoted to the study of geometrically linearization minimal surfaces in a random environment, which roughly speaking  amounts to replacing the area functional on graph configurations by its harmonic approximation. They establish scaling relations for height fluctuations and fluctuations of the ground state energy.

\section{Strategy of proof}\label{sec:strategy-of-proof}

In line with the idea of rough paths, we split our analysis into a purely probabilistic and a purely deterministic part. 
The main ingredient in our stochastic analysis is the \emph{chaining method}, which allows us to give precise bounds on the random field in weak norms. Once the field is controlled stochastically, we can do a path-wise analysis of typical realizations. This reduces the problem to the deterministic question of regularity of almost minimizers of the perimeter problem in two space dimensions.\footnote{This is also reminiscent of the proof of large-scale regularity in stochastic homogenization, see e.g.~\cite[Theorem 1]{gloria_neukamm_otto} or \cite[Theorem 6.1]{armstrong_kuusi_book}: the argument becomes entirely deterministic once a decomposition of the coefficient field with control of the sub-linearity of the corrector pair is assumed. Quantitative control on the latter requires a stochastic argument à la \cite[Theorem 2]{gloria_neukamm_otto} or \cite[Theorem 4.4]{armstrong_kuusi_book}.}

Simplifying the discussion a bit for the moment, we essentially need to control the norm
\begin{align}\label{eqn:weak-norm-to-control}
	\sup \left\{ \frac{1}{ | \partial M | } \int_M \xi ~\middle|~ M \subseteq B_R ~ { \rm Caccioppoli ~ set}  \right\}
\end{align}
of the random field $ \xi $. Recall that a set $ M $ is a Caccioppoli set if the indicator function $ \boldsymbol{1}_M $ is (locally) of bounded variation. Here we write $ | \partial M | $ for $ \int | \nabla \boldsymbol{1}_{ M } | $, which for smooth boundary curves $ \partial M $ reduces to their length, see e.g.~the comment after \cite[Definition 3.5]{AmbrosioFuscoPallara}. It is one of the main insights of \cite{DW} that the discrete analogue of this supremum diverges precisely like $ ( \log R )^{ \frac{3}{4} } $ as $ R \rightarrow \infty $. This is in sharp contrast to the polynomial divergence of the $ L^p $-norms of $ \xi $, which is why we like to think of the above norm as a weak norm of $ \xi $.

Let us assume for the moment that we can control \eqref{eqn:weak-norm-to-control} on some scale $ R $. As alluded to earlier, we use this estimate to show that the behavior of the minimizer of \eqref{eqn:in-energy} on that scale $R$ is dominated by the perimeter term. 
In that sense, the quenched random field is of lower order, which is captured in the following definition (it is instructive to think of $ \eta $ as an upper bound on \eqref{eqn:weak-norm-to-control}).

\begin{definition}[$ \eta $-minimizer]\label{def:eta-minimizer}
Let $ R > 0 $ and $ x \in \R^2 $. If there exists an $ \eta > 0 $, such that for every open ball $ B \subseteq B_R (x) $, and every competitor $ m' : \R^2 \rightarrow \{ -1 , +1 \} $ with $ m = m' $ in $ \R^2 \setminus B $, we have
\begin{align*}
\int_{ \overline{ B } } | \nabla m | \leq ( 1 + \eta ) \int_{ \overline{ B } } | \nabla m' | ,
\end{align*}
we call $ m : \R^2 \rightarrow \{ -1, +1 \} $ an \emph{$ \eta $-minimizer on $ B_R ( x ) $}.
\end{definition}

The main deterministic ingredient is a regularity theory for $ \eta $-minimizers, which has to be as precise as possible in its dependence on $ \eta $. 
Since Definition \ref{def:eta-minimizer} is translation invariant, we restrict the exposition in the following subsection to balls centered around the origin. 

\medskip

\subsection{Deterministic ingredients} Our deterministic regularity theory is based on a Campanato iteration, which we implement as follows: 
We inductively construct a sequence of average normals over dyadic scales by comparing our minimizer to a configuration that has a jump set consisting of a line (the latter having a natural normal associated to it). 
The key difficulty is to ensure that upon passing from one dyadic scale to a smaller, one retains a quantitative control of the error made in the approximation. 

This is a convenient strategy since we can directly relate the $ \eta $-minimizer to a minimizer of the perimeter problem. Here we profit from the two-dimensional setting of our problem: perimeter minimizers $ m_{ \rm line } $ are line configurations, i.e.~the sets $ \{ m_{ \rm line } = 1\} $ and $ \{ m_{ \rm line } = -1 \} $ are half-spaces so that $ J_{ m_{ \rm line } } $ is a line, see Lemma~\ref{lem:implicit-argument-postprocessing}.

Assuming this closeness to a line configuration, it is not difficult to see that on many smaller balls our $ \eta $-minimizer has exactly two jumps (like the line configuration). With this, it follows from $ \eta $-minimality that for a suitable choice of one of these balls the line connecting the two intersection points is a better approximation of our $ \eta $-minimizer on the smaller scale. We formalize this idea in the next proposition. Therein, for ease of notation, we call $ \nu $ the normal to the jump set $ J_{ m_{ \rm line } } $ of a line configuration $ m_{ \rm line } $, since the measure theoretic normal on $ J_{ m_{ \rm line } } $ is constant and agrees with the inwards pointing normal of $ \{ m_{ \rm line } = 1 \} $, see also Remark~\ref{rem:gmt}.

\begin{proposition}[One-step improvement]\label{prop:iteration-step}
There exist $ \varepsilon_0 > 0 $ and $ \eta_0 > 0 $ such that for every $ \varepsilon \leq \varepsilon_0 $ and $ \eta \leq \eta_0 $ the following holds. Let $ m $ be an $ \eta $-minimizer on $ B_R $ such that $ 0 \in J_m $ and
\begin{align}\label{eqn:prop:iteration-step-assumption}
\frac{1}{R^2} \int_{B_R} | m - m_{ \rm line } | \leq \varepsilon
\end{align}
for some $ m_{ \rm line } $ with jump set equal to a line. There exists a radius $ \frac{1}{16} R < r < \frac{15}{16} R $ and another configuration $ m_{ \rm line }' $ with jump set being a line such that
\begin{equation}\label{eqn:tilt-halfspace}
\frac{1}{ r^2 } \int_{ B_{ r } } | m - m_{ \rm line }' |
\lesssim \sqrt{ \eta }
\end{equation}
and
\begin{equation}\label{eqn:tilt-normal}
| \nu - \nu' | \lesssim \frac{1}{ R^2 } \int_{ B_R } | m - m_{ \rm line } |
\end{equation}
where $ \nu $, $ \nu' $ are the normals to $ J_{ m_{ \rm line } } $, resp.~ $ J_{ m_{ \rm line }' } $.
\end{proposition}
Since the statement is conditional on \eqref{eqn:prop:iteration-step-assumption}, iterating Proposition~\ref{prop:iteration-step} leads to $\varepsilon$-regularity of $\eta$-minimizers, i.e.~regularity conditional on the initial closeness to a line configuration. 
Since we are interested in an \emph{unconditional} statement, we need to establish this initial closeness to a line configuration. 

At this point it is convenient to directly compare the $ \eta $-minimizer to a minimizer of the perimeter problem. Using a compactness argument, we show that provided $ \eta \ll 1 $, it cannot be too far from perimeter minimizers in a sufficiently weak norm compared to the natural energy norm:

\begin{proposition}\label{prop:approx-halfspace-l1}
There exists a parameter $ 0 < \theta < 1 $ such that for every $ \varepsilon > 0 $ there exists an $ \eta_0 = \eta_0(\varepsilon) $ with the following property: if $ m $ is an $ \eta $-minimizer on $ B_R $ for some $ \eta \leq \eta_0 $, then there exists a configuration $ m_{ \rm line } $ with jump set given by a line such that
\begin{align*}
\frac{1}{(\theta R)^2} \int_{ B_{ \theta R } } | m - m_{ \rm line } | \leq \varepsilon.
\end{align*}
\end{proposition}

For notational convenience we formulated the above two propositions for $ \eta $-minimizers in balls centered around the origin. Recall that by translation invariance both results also apply to $ \eta $-minimizers centered around any point in $\RR^2$, which will be important in setting up the Campanato iteration in the proof of Theorem \ref{thm:mr}.

It is built into the strategy that, when comparing two neighboring dyadic scales the closeness to the line configuration needs to improve. Of course, this can only happen if on smaller scales, our problem gets closer to the perimeter problem. In particular, Proposition \ref{prop:iteration-step} will only be useful when it is applied to a sequence of $ \eta_R ( x ) $-minimizers, where $ \eta_R ( x ) $ depends on the scale $ R $ and base point $ x $, in such a way that it improves as $ R $ becomes smaller.

\medskip

\subsection{Stochastic ingredients} As motivated in the last paragraph, we will need control over the supremum in \eqref{eqn:weak-norm-to-control} on dyadic scales $ R $ and varying base points $ x $. To this end, we combine moment bounds reminiscent of \cite{DW} with concentration properties of Gaussian random variables. 

As starting point for our estimates we will need the continuum version of Proposition 2.2 in \cite{DW}, see also \cite{Talagrand-GC,Talagrand-ULB}. In our setting their proof needs some technical adjustments.

\begin{proposition}[Continuum version of Proposition 2.2 in \cite{DW}]\label{thm:talagrand-bound}
Let
\begin{align*}
S_R \coloneq \sup \left\{ \Big| \frac{1}{ | \partial M | } \int_M \xi \, \Big| ~\middle|~ M \subseteq B_R ~ { \rm Caccioppoli~set } \right\}.
\end{align*}
Then
\begin{equation}\label{eqn:talagrand-expectation}
\E S_R \lesssim ( \log R )^{ \frac{3}{4} }
\end{equation}
for $ R \geq 2 $, and the variable $ S_R $ satisfies the Sub-Gaussian tail estimate
\begin{equation}\label{eqn:talagrand-concentration}
\P \{ | S_R - \E S_R | \ge t \} \leq 2 \exp ( - \frac{ t^2 }{ 8 \pi } )
\end{equation}
for all $ R > 0 $ and $ t > 0 $.
\end{proposition}
Let us remark that for our purposes the upper bound \eqref{eqn:talagrand-expectation} is enough to establish Theorem~\ref{thm:mr}. In fact, \cite{DW} also prove a lower bound of the same order on the expectation for the discrete model, which could be extended to our setting as well.

The proof of the crucial estimate \eqref{eqn:talagrand-expectation} in the above proposition relies on Talagrand's generic chaining method, see the textbook \cite{Talagrand-GC} for a nice exposition. The main idea, which we elaborate more in Section \ref{sec:pf-talagrand-bound}, is to capitalize on the equivalence of bounding the supremum of Gaussian random variables and appropriate partitions of their associated metric space. In our setting this may be rephrased as a question on quantifying the compactness of embeddings of certain function spaces.  

Starting from Proposition \ref{thm:talagrand-bound}, we now specify what we mean by controlling \eqref{eqn:weak-norm-to-control} over dyadic scales. While the expectation in \eqref{eqn:talagrand-expectation} blows up in $ R $, \eqref{eqn:talagrand-concentration} shows that its variance stays of order one. In particular, upon normalizing \eqref{eqn:weak-norm-to-control} by its expectation, its variance decays. This leads to good concentration properties as made precise in the next proposition.

\begin{proposition}\label{prop:concentration-scales}
We have
\begin{equation}\label{eqn:expectation-bound-sup-x-r}
\E \sup_{ R \ge 2 } ( \log R )^{ - \frac{3}{4} } S_R \lesssim 1
\end{equation}
and
\begin{equation}\label{eqn:tail-bound-sup-x-r}
\P \{ | \sup_{ R \ge 2 } ( \log R )^{ - \frac{3}{4} } S_R - \E \sup_{ R \ge 2 } ( \log R )^{ - \frac{3}{4} } S_R | \ge t  \}
\leq 2 \exp( - \frac{t^2}{8\pi} )
\end{equation}
for every $ t > 0 $.
\end{proposition}

Proposition~\ref{prop:concentration-scales} only solves one of our problems, namely the estimation of \eqref{eqn:weak-norm-to-control} over dyadic scales. We still need to resolve the pinning in different space points.
Heuristically, we may argue as follows: since $ \xi $ is regular, i.e.\ does not vary much on spatial scales of order one, we may replace its supremum over a box of radius one by its value at the center point. Therefore, controlling \eqref{eqn:weak-norm-to-control} over points in a box of radius $ R $ amounts to bounding the quantity for $ \sim R^d $ different center points. Since we are interested in the supremum of these $ R^d $ random variables, any correlation structure between them makes their supremum smaller, so that the worst upper bound is obtained in the independent case. Proposition \ref{prop:concentration-scales} tells us that the random variables of our interest essentially behave like Gaussians with mean and variance of order one. It is well known that the supremum of $ R^d $ of these random variables scales like $ ( \log R^d )^{ \frac{1}{2} } \sim ( \log R )^{ \frac{1}{2} } $ in $ R $. Hence, another concentration argument similar to the one we used above yields the following proposition.

\begin{proposition}\label{prop:concentration-in-space}
For 
\begin{align*}
X =  \sup_{ x \in \R^2 } ( \log | x |_+ )^{ - \frac{1}{2} } \sup_{ R \ge 2 } ( \log R )^{ - \frac{3}{4} } S_R ( \xi ( \cdot - x ) ),
\end{align*}
with $ | x |_+ \coloneq \max \{ | x |, 2 \} $, we have
\begin{equation}\label{eqn:concentration-in-space-expectation}
\E X \lesssim 1
\end{equation}
and
\begin{equation}\label{eqn:concentration-in-space-tail}
\P \{ | X - \E X | \geq t \}
\leq 2\exp( - \frac{t^2}{8\pi} )
\end{equation}
for every $ t > 0 $.
\end{proposition}

We expect that the concentration properties in the dyadic scales (Proposition \ref{prop:concentration-scales}) are essentially sharp. Since the variance of large-scale contributions is close to zero, they essentially behave like order one quantities.
In contrast, Proposition \ref{prop:concentration-in-space} may not be sharp: since the underlying noise $\xi$ has correlation length one, the short-range contributions to the supremum at two distinct space points are independent, whilst there is a nontrivial correlation structure in the longer-range contributions. 

\section{Proof of the deterministic ingredients}\label{sec:regularity-theory}

The goal of this section is to prove the deterministic ingredients in the proof of Theorem~\ref{thm:mr} outlined in Section~\ref{sec:strategy-of-proof}. In particular, we provide the proofs of Propositions~\ref{prop:iteration-step} and \ref{prop:approx-halfspace-l1}. Let us stress again that centering the minimizer around the origin plays no special role; this is done for notational convenience only. 
Moreover, let us emphasize that due to our two-dimensional setting, we can avoid the full machinery from the regularity theory of minimal surfaces and instead use elementary -- and self-contained -- geometric arguments.

Before we start, let us remark that as long as our claim is formulated in a scale invariant way -- as for example in Propositions \ref{prop:iteration-step} and \ref{prop:approx-halfspace-l1} -- we may always assume in the proofs that $ R = 1 $:

\begin{remark}[Rescaling]\label{rmk:eta-minimizer}
Let $ m $ be an $ \eta_R $-minimizer on $ B_R $, then $ m $ is also an $ \eta_R $-minimizer on every smaller ball $ B \subseteq B_R $; furthermore $ m(R \cdot) $ is an $ \eta_R $-minimizer on $ B_1 $.
\end{remark}

After this initial preparation, we start with the proofs of the two Propositions \ref{prop:iteration-step} and \ref{prop:approx-halfspace-l1}. We begin with the latter by employing an implicit argument. Therein, we will use the notion of a \emph{minimizer $ m_{ \rm Per } $ of the perimeter in $ B_R $}, meaning that for any $m': B_R \to \{-1, +1\}$ there holds 
\begin{align}
	\int_{ B_R } | \nabla m_{ \rm Per } | 
	\leq \int_{ B_R } | \nabla m' | 
	\quad \text{provided} ~ \{ m_{\rm Per} \neq m' \} ~ \text{is compactly contained in} ~ B_R ; \label{eqn:per-minimality}
\end{align}
that is, $ m_{ \rm Per } $ minimizes $ m \mapsto \int_{ B_R } | \nabla m | $ under compactly supported perturbations.\footnote{The notion of a $ 0 $-minimizer of the perimeter (in the sense of Definition \ref{def:eta-minimizer}) is slightly stronger since perturbations are allowed up to the boundary of $ B_R $.}

\begin{lemma}\label{lem:implicit-argument}
For every $ \varepsilon > 0 $ there exists an $ \eta_0 > 0 $ such that every $ \eta $-minimizer $ m $ on $ B_R $ with $ \eta \leq \eta_0 $ there exists a minimizer $ m_{ \rm Per } $ of the perimeter in $ B_R $ such that
\begin{align*}
\frac{1}{R^2} \int_{ B_R } | m - m_{ \rm Per } | \leq \varepsilon.
\end{align*}
\end{lemma}

\begin{proof}
By Remark \ref{rmk:eta-minimizer} we assume w.l.o.g.~that $ R = 1 $. Let us assume there exists a sequence $ m_k $ of $ \eta_k $-minimizers on $ B_1 $ such that
\begin{align*}
\int_{ B_1 } | m_k - m_{ \rm Per } | > \varepsilon
\quad {\rm for~every~minimizer~of~the~perimeter~} m_{ \rm Per } ~ { \rm in } ~ B_1
\end{align*}
while $ \eta_k \searrow 0 $ as $ k \rightarrow \infty $. We may compare $ m_k $ to the configuration that is equal to $ 1 $ in $ B_{ 1 } $ to get
\begin{align*}
\int_{ \overline{ B_1 } } | \nabla m_k |
\leq ( 1 + \eta_k ) \int_{ \overline{ B_1 } } | \nabla ( \boldsymbol{1}_{ B_1 } + m_k \boldsymbol{1}_{ \R^2 \setminus B_1 } ) |
\leq 2 ( 1 + \eta_k ) | \partial B_1 |.
\end{align*}
The factor two on the r.h.s.~results from the jump height of two, since $ m \in \{ \pm 1 \} $. Combining this with the fact that $ -1 \leq m_k \leq 1 $, we learn that $ (m_k)_{k \in \N} $ is compact in $ L^1(B_1) $. Hence, we may assume that $ m_k \rightarrow m $ in $ L^1(B_1) $ as $ k \rightarrow \infty $. Following a routine computation, see e.g.~\cite[Theorem 21.14 \& Remark 21.17]{Maggi} or \cite[Lemma 9.1]{Giusti}, one can show that $ m $ is a minimizer of the perimeter in $ B_1 $. This contradicts our assumption. Since the details are somewhat technical, we defer them to Appendix \ref{app:minimality}.
\end{proof}

The next lemma post-processes the implicit argument from above. This is where it matters that we work in two dimensions: the minimizers of the perimeter problem are unions of line segments, see the discussion following the Lemma~\ref{lem:implicit-argument-postprocessing}. Hence, once we zoom in far enough, the previously constructed competitor is induced by a single line. The main difficulty at this place is to make this argument work under our weak assumptions.

\begin{lemma}\label{lem:implicit-argument-postprocessing}
There exists a $ 0 < \theta < 1 $ such that for every perimeter minimizer $ m_{ \rm Per } $ in $ B_R $ the jump set of $ m_{ \rm Per } $ in $ B_{ \theta R } $ is a line.
\end{lemma}

Let us first explain the intuition behind Lemma \ref{lem:implicit-argument-postprocessing}: If the jump set contains a small curved arc, one can replace this arc by a line segment. By the triangle inequality, this reduces the perimeter of the configuration, while keeping its boundary conditions fixed. Hence, the jump set of any minimizer of the perimeter problem should be the union of line segments. The difficulty in making this rigorous is that a priori the jump set is the fairly rough boundary of a Caccioppoli set. In the proof given below, we use the classical, but heavy, machinery of small-scale regularity theory to overcome this difficulty.\footnote{An alternative proof is given in the introductory textbook \cite[Chapter 10]{FrankMorgan} on geometric measure theory. The argument therein avoids the use of heavy regularity theory, however, at the price of working in the context of currents.}

\begin{proof}[Proof of Lemma \ref{lem:implicit-argument-postprocessing}]
By Remark \ref{rmk:eta-minimizer} we assume w.l.o.g.~that $ R = 1 $. Since we are restricted to configurations on $ \R^2 $, we know that the jump set, or reduced boundary (see Remark \ref{rem:gmt}), of $ m_{ \rm Per } $ is the union of line segments. The regularity theory for two-dimensional minimal surfaces allows us to conclude as follows: In \cite[Theorem~on~p.~275]{Maggi} it is shown that any minimizer of the perimeter in dimension $ n \leq 7  $ is an analytic hypersurface of vanishing mean curvature. Since we are interested in the setting $ n = 2 $, this shows that the jump set of any minimizer of the perimeter is the union of non-intersecting analytic curves with curvature zero, i.e.~non-intersecting line segments.

By its finite length, the jump set consists of at most countably many line segments. Moreover, it follows that any open ball $ B \subseteq B_1 $ intersects only finitely many of these line segments. Furthermore, we learn from minimality that these line segments can only intersect on $ \partial B_1 $.

\begin{figure}[ht]
\begin{tikzpicture}[line join=round,>=triangle 45,x=1cm,y=1cm, scale=0.6]
\fill[line width=1pt,fill=black,fill opacity=0.25] (-0.5683502655950378,-3.9594163680519947) -- (-1.1450010149627492,3.8326195579178317) -- (0.96667283568455,3.881435768984151) -- (0.12298621767943245,-3.9981088517273897) -- cycle;
\draw [line width=1pt] (0,0) circle (4cm);
\draw [line width=1pt] (-0.5683502655950378,-3.9594163680519947)-- (-1.1450010149627492,3.8326195579178317);
\draw [line width=1pt] (0.12298621767943245,-3.9981088517273897)-- (0.96667283568455,3.881435768984151);
\draw [line width=1pt] (0,0) circle (1.6151871854679878cm);
\draw [line width=1pt, |-|] (4.8,3.8)-- (4.8,-4);
\draw [line width=1pt] (-0.5683502655950378,-3.9594163680519947)-- (-1.1450010149627492,3.8326195579178317);
\draw [line width=1pt] (-1.1450010149627492,3.8326195579178317)-- (0.96667283568455,3.881435768984151);
\draw [line width=1pt] (0.96667283568455,3.881435768984151)-- (0.12298621767943245,-3.9981088517273897);
\draw [line width=1pt] (0.12298621767943245,-3.9981088517273897)-- (-0.5683502655950378,-3.9594163680519947);
\draw (-0.3525196850136373,3.4752191789883584) node[anchor=north west] {$P$};
\draw [line width=1pt, |-|] (1,4.4)-- (-1.2,4.4);
\draw (-2.709628551131483,0.6342005456542573) node[anchor=north west] {$+$};
\draw (2.207519083485215,0.6342005456542573) node[anchor=north west] {$+$};
\draw (-0.39934966248617726,0.649810538145104) node[anchor=north west] {$-$};
\draw[color=black] (-3.5,2.9) node {$\partial B_1$};
\draw [fill=black] (-0.5683502655950378,-3.9594163680519947) circle (2.5pt);
\draw [fill=black] (-1.1450010149627492,3.8326195579178317) circle (2.5pt);
\draw [fill=black] (0.12298621767943245,-3.9981088517273897) circle (2.5pt);
\draw [fill=black] (0.96667283568455,3.881435768984151) circle (2.5pt);
\draw[color=black] (-2,1) node {$\partial B_{\theta}$};
\draw[color=black] (5.8,0.2595607258739363) node {$L\to 1$};
\draw[color=black] (0,4.802068540710328) node {$\ell \to 0$};
\end{tikzpicture}
\caption{Non-degenerate line configuration considered in the proof of Lemma~\ref{lem:implicit-argument-postprocessing}.}\label{fig:L33}
\end{figure}

Let us assume that at least two distinct line segments intersect $ B_{ \theta } $. All of them must extend as straight lines to the boundary $ \partial B_1 $. We now pick two of these line segments that are adjacent. To conclude, we will build a competitor as depicted in Figure~\ref{fig:L33}.

If $ \theta $ is small enough, the maximal distance $ \ell $ between the closest endpoints on $ \partial B_1 $ of these two lines becomes smaller than the minimal distance $ L $ of the endpoints belonging to the same line segment. In fact we have $ \ell \to 0 $ and $ L \to 1 $ as $ \theta \to 0 $. Further, let $ P \subset \R^2 $ denote the polygon that is spanned by the endpoints of the line segments on $ \partial B_1 $. We now compare $ m_{ \rm Per } $ to the configuration that has $ P $ removed; by the minimality\footnote{Technically speaking we have to first consider polygons $ P' $ compactly contained in $ P $, in which case $  m_{ \rm Per } \boldsymbol{1}_{ \R^2 \setminus P ' } - m_{ \rm Per } \boldsymbol{1}_{ P ' } $ is an admissible competitor in \eqref{eqn:per-minimality}, so that
\begin{align*}
	\int_{ B_1 } | \nabla m_{ \rm Per } |
	&\leq \int_{ B_1 } | \nabla ( m_{ \rm Per } \boldsymbol{1}_{ \R^2 \setminus P ' } - m_{ \rm Per } \boldsymbol{1}_{ P ' } ) | .
\end{align*}
Now, by lower semicontinuity of the perimeter w.r.t.~convergence in $ L^1 $ the above identity also holds for $ P $ upon approximating $ P ' \to P $.} \eqref{eqn:per-minimality} of $ m_{ \rm Per } $ we thus learn
\begin{align*}
	\int_{ B_1 } | \nabla m_{ \rm Per } |
	&\leq \int_{ B_1 } | \nabla ( m_{ \rm Per } \boldsymbol{1}_{ \R^2 \setminus P } - m_{ \rm Per } \boldsymbol{1}_{ P } ) | \\
	&\leq \int_{ B_1 \cap \mathring{P}  } | \nabla m_{ \rm Per }  |
	+ \int_{ B_1 \cap \overline{P}^c } | \nabla m_{ \rm Per }  |
	+ \int_{ \partial P } | m_{ \rm Per } \lfloor_{ \R^2 \setminus P } - m_{ \rm Per } \lfloor_{ P }   | .
\end{align*}
Comparing the last term to $ \frac{ 1 }{ 2 } \int_{ \partial P } | \nabla m_{ \rm Per } | $, we removed at least $ 2 L $ and add at most $ 2 \ell $ to the perimeter. Thus we obtain
\begin{align*}
	\int_{ B_1 } | \nabla m_{ \rm Per } |
	\leq \int_{ B_1 } | \nabla ( m_{ \rm Per } \boldsymbol{1}_{ \R^2 \setminus P } - m_{ \rm Per } \boldsymbol{1}_{ P }  ) |
	\leq \int_{ B_1 } | \nabla m_{ \rm Per } | - 4 L + 4 \ell
	\stackrel{ \theta \to 0 }{ \rightarrow } \int_{ B_1 } | \nabla m_{ \rm Per } | - 4 ;
\end{align*}
a contradiction for small $ \theta $.
\end{proof}

The proof of Proposition \ref{prop:approx-halfspace-l1} now easily follows:

\begin{proof}[Proof of Proposition \ref{prop:approx-halfspace-l1}]
	Let $ \theta $ be the universal constant in Lemma~\ref{lem:implicit-argument-postprocessing}. By Lemma~\ref{lem:implicit-argument} we know that for the choice $ \theta^2 \varepsilon > 0 $ exists an $ \eta_0 = \eta_0 ( \theta, \varepsilon ) $ such that every $ \eta $-minimizer $ m $ on $ B_R $ with $ \eta \leq \eta_0 $ there exists a minimizer $ m_{ \rm Per } $ of the perimeter in $ B_R $ such that
	\begin{align*}
		\frac{1}{R^2} \int_{ B_R } | m - m_{ \rm Per } | \leq \theta^2 \varepsilon.
	\end{align*}
	By the choice of $ \theta $ we know that the jump set of $ m_{ \rm Per } $ in $ B_{ \theta R } $ consists of a (possibly empty) line segment. Thus, the desired $ m_{ \rm line } $ is given by the trivial extension of the former configuration to the whole space.
\end{proof}

Before we continue towards the proof of Proposition~\ref{prop:iteration-step}, let us interpret the proposition we just proved. We think of the quantity $ \frac{1}{R^2} \int_{ B_R } | m - m_{ \rm line } | $ as an excess that measures flatness (in a weak sense). Let us remark, though, that in classical small-scale regularity theory the excess is typically defined in terms of a stronger energy norm. Again, it is the two dimensional setting that enables us to conclude Theorem~\ref{thm:mr} from this weaker notion of excess, see also Remark~\ref{rmk:strong-exc-bound} at the end of this section.

For the proof of Proposition~\ref{prop:iteration-step}, we will need to first prove some auxiliary lemmas. We first show that we can find a ``good'' radius: Provided the excess is small, we can find at least one radius (which does not degenerate to zero) such that on the sphere of that radius our minimizer is close to a line.

To keep the next statements simple, we are a bit imprecise\footnote{This is also justified by the fact that due to the assumptions \eqref{eqn:noise-assumption} on the noise, we are effectively constrained to a class of Caccioppoli sets with smooth boundaries, see again~\cite[Example 21.2 \& Theorem 21.8]{Maggi}.}
 when it comes to boundary traces. This is justified by the following remark.
 
 \begin{remark} For a function of bounded variation the inner trace (which we denote by $ m \lfloor_{ \partial B_r } $) and outer trace (which we denote by $ m \lfloor_{ \partial(\R^2 \setminus B_r) } $), both w.r.t.~$ \partial B_r $, agree for almost every radius $ r $ as $ L^1 $-functions on $ \partial B_r $, see e.g.~\cite[Remark 3.1]{Giusti}.\footnote{This justifies that we do not introduce new notation for the traces of $m$ on spheres.} For these radii we can still identify $ J_m \cap \partial B_r $ and $ \partial^* \{ m = 1 \} \cap \partial B_r $, since they have the same $ \mathscr{H}^{1} $-measure, see Remark~\ref{rem:gmt}.
 \end{remark}
 
In spirit of the above remark, we implicitly assume that whenever we write down a boundary term we work with one of the aforementioned radii.

\begin{lemma}\label{lem:few-jumps}
There exists $ \varepsilon_0 > 0 $ and $ \eta_0 > 0 $ such that for every $ \varepsilon \leq \varepsilon_0 $ and $ \eta \leq \eta_0 $ the following holds: for every $ \eta $-minimizer $ m $ on $ B_R $ such that
\begin{align*}
\frac{1}{R^2} \int_{ B_R } | m - m_{ \rm line } | \leq \varepsilon
\end{align*}
for some $ m_{ \rm line } $ with jump set equal to a line, there exists a radius $ \frac{1}{16} R \leq r \le \frac{15}{16} R $ such that
\begin{align*}
\frac{1}{ r } \int_{ \partial B_r } | m - m_{ \rm line } |
\lesssim \, \frac{1}{R^2} \int_{ B_R } | m - m_{ \rm line } |
\end{align*}
and $ m \lfloor_{ \partial B_r } $ has either zero or two jumps on $ \partial B_r $. 
\end{lemma}

We will actually prove that the set of radii satisfying the condition of the lemma has positive measure. This is important in view of our implicit assumption concerning boundary traces as explained before Lemma~\ref{lem:few-jumps} above.

The main idea behind the lemma is the following. Since for any line configuration $ m_{ \rm line } $ the jump set has two intersection points with a sphere, and $ m $ is close to such a configuration, we expect this property to carry over to sufficiently many spheres. More precisely, since $ \frac{ 1 }{ R^2 } \int_{ B_R } | m - m_{ \rm line } | \leq \varepsilon $, the layer cake formula implies that there exists a ``good'' radius such that $ \frac{ 1 }{ r_1 } \int_{ \partial B_{ r_1 } } | m - m_{ \rm line } | \lesssim \varepsilon $. Thus, by $ \eta $-minimality the perimeter $ \frac{ 1 }{ 2 } \int_{ B_{ r_1 } } | \nabla m | $ behaves like the perimeter $ \frac{ 1 }{ 2 } \int_{ B_{ r_1 } } | \nabla m_{ \rm line } | \leq 2 r_1 $. Geometrically, this implies that the jump set $ J_m $ is short enough, so that it cannot intersect too many spheres in more than two points. 

\begin{figure}[ht]

\begin{tikzpicture}[line cap=round,line join=round,>=triangle 45,x=1cm,y=1cm,scale=.5]

\draw [shift={(-2.362557200696298,-2.0368889612724157)},line width=1pt,color=blue]  plot[domain=2.0119624224748316:3.5425129709957828,variable=\t]({1*1.9381325132715082*cos(\t r)+0*1.9381325132715082*sin(\t r)},{0*1.9381325132715082*cos(\t r)+1*1.9381325132715082*sin(\t r)});
\draw [shift={(-2.3787497316935813,-2.218338965471483)},line width=1pt,color=blue]  plot[domain=0.9801426899712782:1.968025427810441,variable=\t]({1*2.097319423659996*cos(\t r)+0*2.097319423659996*sin(\t r)},{0*2.097319423659996*cos(\t r)+1*2.097319423659996*sin(\t r)});
\draw [shift={(-0.1646193056715446,0.9213684497996446)},line width=1pt,color=blue]  plot[domain=4.069880889737176:5.09549199672038,variable=\t]({1*1.7458536982086001*cos(\t r)+0*1.7458536982086001*sin(\t r)},{0*1.7458536982086001*cos(\t r)+1*1.7458536982086001*sin(\t r)});
\draw [shift={(0.049666708352282304,0.5307416875737344)},line width=1pt,color=blue]  plot[domain=-1.2281301079743177:0.5192317559801743,variable=\t]({1*1.3045099790868204*cos(\t r)+0*1.3045099790868204*sin(\t r)},{0*1.3045099790868204*cos(\t r)+1*1.3045099790868204*sin(\t r)});
\draw [shift={(0.503419150261584,0.7429559272435011)},line width=1pt,color=blue]  plot[domain=0.5699959683328197:2.835362851653719,variable=\t]({1*0.8062969803606828*cos(\t r)+0*0.8062969803606828*sin(\t r)},{0*0.8062969803606828*cos(\t r)+1*0.8062969803606828*sin(\t r)});
\draw [shift={(-0.8942081124498833,1.27055041475496)},line width=1pt,color=blue]  plot[domain=3.254437654806985:5.858290619960291,variable=\t]({1*0.690214041232061*cos(\t r)+0*0.690214041232061*sin(\t r)},{0*0.690214041232061*cos(\t r)+1*0.690214041232061*sin(\t r)});
\draw [shift={(-0.6141846443963681,1.3954773407266599)},line width=1pt,color=blue]  plot[domain=1.6966459167558359:3.3484072812584937,variable=\t]({1*0.9868779750130782*cos(\t r)+0*0.9868779750130782*sin(\t r)},{0*0.9868779750130782*cos(\t r)+1*0.9868779750130782*sin(\t r)});
\draw [shift={(-3.0217338202190613,11.448298105018896)},line width=1pt,color=blue]  plot[domain=4.958947962728468:5.156686099442129,variable=\t]({1*9.356713300724762*cos(\t r)+0*9.356713300724762*sin(\t r)},{0*9.356713300724762*cos(\t r)+1*9.356713300724762*sin(\t r)});
\draw [shift={(0.4211373287902274,3.9564464381716733)},line width=1pt,color=blue]  plot[domain=-1.02654577756614:0.06490823008305485,variable=\t]({1*1.1179766460939127*cos(\t r)+0*1.1179766460939127*sin(\t r)},{0*1.1179766460939127*cos(\t r)+1*1.1179766460939127*sin(\t r)});
\draw [shift={(-1.2188184326368205,3.8855949476395537)},line width=1pt,color=blue]  plot[domain=0.051980845521262906:0.3396891196207035,variable=\t]({1*2.7593051660724943*cos(\t r)+0*2.7593051660724943*sin(\t r)},{0*2.7593051660724943*cos(\t r)+1*2.7593051660724943*sin(\t r)});

\draw [line width=1pt] (0,0) circle (5cm);
\draw [line width=1pt,color=red] (0,0) circle (1.5cm);
\draw [line width=1pt,color=YellowGreen] (0,0) circle (4.5cm);

\end{tikzpicture}

\caption{Choice of ``good'' radius (in {\color{YellowGreen}green}) intersecting the jump set in two points, and a ``bad'' radius (in {\color{red}red}) with more intersection points.}\label{fig:good-radius}	
\end{figure}

To implement the above idea rigorously, we would like to use identity
\begin{align}
	 \int_{B_1} | \nabla m | = \int_0^1\int_{ \partial B_r } | \nabla m | \geq \int_0^1 \int_{ \partial B_r } | \nabla^{ \rm tan } m |, \label{eqn05} 
\end{align}
which holds for any smooth function $ m $, where $ \nabla^{ \rm tan } m (x) $ denotes the tangential gradient, i.e~the orthogonal projection of $ \nabla m $ onto the tangent space of the circle. In our setting, where $ m $ is not smooth but takes values in $ \{ \pm 1 \} $, the estimate \eqref{eqn05} is still true if the r.h.s.~is interpreted as the variation of $ m $ in the angular variable, that is, the number of jumps of $ t \mapsto m ( r e^{ i t } ) $. 

Indeed, for any function of bounded variation we can define the r.h.s.~of \eqref{eqn05} via $ \int_{ \partial B_r } | \nabla^{ \rm tan } m | = V_{ \rm per } ( t \mapsto m ( r \mathrm{e}^{ i t}  ) ) $, where $ V_{ \rm per } ( \cdot ) $ denotes the (periodic) variation, see \cite[Definition 3.4 \& Discussion around (3.99)]{AmbrosioFuscoPallara} for further details. Since the latter expression is lower semicontinuous w.r.t.~convergence in $ L^1( \partial B_r ) $, an approximation as in \cite[Theorem 3.9]{AmbrosioFuscoPallara} applies and \eqref{eqn05} holds. Furthermore, we note that by the one-dimensional characterization of functions of bounded variation, see \cite[Proposition 3.52 \& Theorem 3.28]{AmbrosioFuscoPallara}, $ V_{ \rm per } ( t \mapsto m ( r \mathrm{e}^{ i t}  ) )  $ can be understood as counting the jumps of $ m $ on $ \partial B_r $. For almost every $ r $, the aforementioned jump points agree with $ J_m \cap \partial B_r $, see \cite[Theorem 3.10]{AmbrosioFuscoPallara}.

\begin{proof}[Proof of Lemma \ref{lem:few-jumps}.]
By Remark \ref{rmk:eta-minimizer} we assume w.l.o.g.~that $ R = 1 $. For ease of notation we may pass to a smaller $ \varepsilon $ such that
\begin{align*}
	\frac{1}{R^2} \int_{ B_R } | m - m_{ \rm line } | = \varepsilon
	\quad \text{holds, and the claim becomes} \quad
	\frac{1}{ r } \int_{ \partial B_r } | m - m_{ \rm line } | \lesssim \varepsilon 
\end{align*}
for some $ \frac{1}{16} R \leq r \le \frac{15}{16} R $. To see this, we first note that there exists a radius $ \frac{7}{8} < r_1 < 1 $ such that
\begin{align*}
\int_{ \partial B_{r_1} } |  m - m_{ \rm line } | \leq 16 \varepsilon.
\end{align*}
Indeed, by Markov's inequality
\begin{align*}
\Big| \Big\{ \frac{7}{8} < r_1 < 1 ~\Big|~ \int_{ \partial B_{ r_1 } } | m - m_{ \rm line } | > 16 \varepsilon \Big\} \Big|
\leq \frac{1}{16 \varepsilon} \int_{ \frac{7}{8} } ^1 \int_{ \partial B_{ r_1 } } | m - m_{ \rm line } | d r_1 \leq \frac{1}{16},
\end{align*}
so that the complement has positive Lebesgue-measure. In particular, we can pick a radius $ \frac{ 7 }{ 8 } < r_1 < 1 $ such that inner and outer trace of $ m $ on $ \partial B_{ r_1 } $ agree.

For the fixed $ r_1 $ from above, we now compare $ m $ to $ m_{ \rm line } $ on $ B_{ r_1 } $ to obtain
\begin{align*}
\int_{ \overline{ B_{r_1} } } | \nabla m |
\leq ( 1 + \eta ) \left( \int_{ B_{r_1} } | \nabla m_{ \rm line } | + \int_{ \partial B_{r_1} } |  m \lfloor_{ \R^2 \setminus B_{ r_1 } } - m_{ \rm line } | \right)
\leq ( 1 + \eta ) ( 4 r_1 + 16 \varepsilon ).
\end{align*}
Since for almost every $ 0 < r_2 < r_1 $, the integral $ \int_{ \partial B_{ r_2 } } | \nabla^{ \rm tan } m | $ is equal to twice (recall that the jump height is $ 2 $) the number of jumps (which is an even number) of $ m $ on $ \partial B_r $, we have
\begin{align*}
8 \Big| \Big\{ 0 < r_2 < r_1 ~ \Big| ~ \int_{ \partial B_{ r_2 } } | \nabla^{ \rm tan } m | \notin \{ 0, 4 \} \Big\} \Big|
& \leq \int_0^{r_1} \left( \int_{ \partial B_{ r_2 } } | \nabla^{ \rm tan } m | \right) dr_2 \\
& \leq \int_0^{r_1} \left( \int_{ \partial B_{ r_2 } } | \nabla m | \right) dr_2
\leq \int_{ B_{r_1} } | \nabla m |,
\end{align*}
which implies
\begin{align*}
\Big| \Big\{ 0 < r_2 < \frac{7}{8} ~ \Big| ~ \int_{ \partial B_{ r_2 } } | \nabla_{ \partial B_{ r_2 } } m | \notin \{ 0, 4 \} \Big\} \Big|
\leq \frac{1}{2} ( 1 + \eta ) ( 1 + 4 \varepsilon ).
\end{align*}
Furthermore, we have
\begin{align*}
\Big| \Big\{ 0 < r_2 < \frac{7}{8} ~ \Big| ~ \int_{ \partial B_{ r_2 } } | m - m_{ \rm line } | \geq 8 \varepsilon \Big\} \Big|
\leq \frac{1}{8}.
\end{align*}
Hence, for small enough $ \eta, \varepsilon \ll 1 $ the union of the above two sets has at most Lebesgue measure $ \frac{3}{4} $. Therefore, its complement in $ [0, \frac{7}{8}] $ has Lebesgue measure larger than $ \frac{1}{8} $, so that it must intersect with $ [ \frac{1}{16}, \frac{7}{8} ] $. Hence, we can find a radius $ \frac{1}{16} \leq r_2 \leq \frac{15}{16} $ such that the claim of Lemma \ref{lem:few-jumps} holds.
\end{proof}

Next, we will prove an analog of what is known as the height bound in the classical regularity theory of minimal surfaces, see e.g.~\cite[Section 22]{Maggi}. In our two dimensional setting, this can be done using fairly elementary geometric arguments. In essence, we solve the following optimization problem: given two points in the plane connected by a curve of a given length, how far can this curve deviate from the line connecting the two points?

\begin{lemma}[Height bound]\label{lem:height-bound}
Let $ \gamma $ be a curve in $ B_R $ that starts and ends at points $ A, B \in \partial B_R $ and satisfies
\begin{align*}
{ \rm Length } (\gamma) \leq ( 1 + \eta ) | A - B |
\end{align*}
for some $ 0 \leq \eta \leq 1 $; then
\begin{align*}
\frac{ { \rm dist } ( \gamma, { \rm conv } ( A, B ) ) }{ R }
\lesssim \sqrt{ \eta } \frac{| A - B |}{R},
\end{align*}
where $  { \rm dist } ( \gamma, { \rm conv } ( A, B ) ) $ refers to the Hausdorff distance\footnote{For two sets $M, N$ the (two-sided) Hausdorff distance is given by \begin{align*}\mathrm{dist}(M,N) \coloneq \max\big\{\sup_{p \in M} \mathrm{dist}(p,N), \sup_{q \in N} \mathrm{dist}(q,M)\big\}.\end{align*}} between the curve $ \gamma $ and the line segment connecting $ A $ and $ B $.
\end{lemma}

\begin{proof}
We assume w.l.o.g.~that $ R = 1 $. 
Since $ { \rm conv } (A,B) $ is a line segment with the same start and end point as $\gamma$, $ h \coloneq { \rm dist } ( \gamma, { \rm conv } (A, B) ) $ is in fact equal to the one-sided Hausdorff distance,~i.e.
\begin{align}\label{eqn:height-bound-2}
h = \sup_{ y \in \gamma } { \rm dist } ( y, { \rm conv } ( A, B ) ).
\end{align}

First, let us note that the maximum value of $ h $ (taken over all admissible curves $ \gamma $) is equal to the maximum of $ h $ taken over all curves that consist of two line segments, i.e.
\begin{align}
	\sup_{ { \rm admissible } ~ \gamma } h &= \sup \{ { \rm dist } ( C, { \rm conv } ( A, B ) ) ~ | ~ C \in B_1 ~ { \rm s.t. } ~ | A - C | + | B - C | \leq ( 1 + \eta ) | A - B | \}.  \label{eqn:height-bound-2-b}
\end{align}
\begin{figure}
\centering

\begin{minipage}{.48\textwidth}
\centering
\begin{tikzpicture}[line join=round,>=triangle 45,x=1cm,y=1cm]
\clip(-3,-3) rectangle (3.5,2.8);
\draw [shift={(-1.3409464800701152,-0.15905351992988478)},line width=1pt,color=blue]  plot[domain=1.317647092817942:2.9047845613130194,variable=\t]({1*0.6779746044905054*cos(\t r)+0*0.6779746044905054*sin(\t r)},{0*0.6779746044905054*cos(\t r)+1*0.6779746044905054*sin(\t r)});
\draw [shift={(-1.0482116589162211,1.4842785543808061)},line width=1pt,color=blue]  plot[domain=4.588470332215275:5.917146629095273,variable=\t]({1*0.9945921752743889*cos(\t r)+0*0.9945921752743889*sin(\t r)},{0*0.9945921752743889*cos(\t r)+1*0.9945921752743889*sin(\t r)});
\draw [shift={(0.46667181722723083,0.8534370808551299)},line width=1pt,color=blue]  plot[domain=0.05150111521724326:2.703136442553309,variable=\t]({1*0.6474213281782358*cos(\t r)+0*0.6474213281782358*sin(\t r)},{0*0.6474213281782358*cos(\t r)+1*0.6474213281782358*sin(\t r)});
\draw [shift={(3.211017859025946,1.1332939321536868)},line width=1pt,color=blue]  plot[domain=3.258574753614191:3.707926381952951,variable=\t]({1*2.1122193101324527*cos(\t r)+0*2.1122193101324527*sin(\t r)},{0*2.1122193101324527*cos(\t r)+1*2.1122193101324527*sin(\t r)});
\draw [shift={(1.7142857142857144,0.2217322961534863)},line width=1pt,color=blue]  plot[domain=3.8015665185708345:5.623211442198545,variable=\t]({1*0.3616598736640848*cos(\t r)+0*0.3616598736640848*sin(\t r)},{0*0.3616598736640848*cos(\t r)+1*0.3616598736640848*sin(\t r)});
\draw [line width=1pt] (0,0) circle (2cm);
\draw [line width=1pt,dash pattern=on 1pt off 1pt] (0.5,0)-- (0.5,1.5);
\draw [line width=1pt] (-2,0)-- (0.5,1.5);
\draw [line width=1pt] (0.5,1.5)-- (2,0);
\draw [line width=1pt] (-2,0)-- (2,0);
\begin{scriptsize}
\draw [fill=black] (-2,0) circle (1pt);
\draw[color=black] (-2.2,0) node {$A$};
\draw [fill=black] (2,0) circle (1pt);
\draw[color=black] (2.2,0) node {$B$};
\draw [fill=black] (0.5,1.5) circle (1pt);
\draw[color=black] (0.5,1.75) node {$C$};
\draw[color=black] (0.65,0.75) node {$h$};
\draw[color=black] (0,-0.3) node {$\mathrm{conv}(A,B)$};
\draw[color=blue] (-1.2,0.8) node {$\gamma$};
\end{scriptsize}
\end{tikzpicture}
\end{minipage}
\hfill
\begin{minipage}{.48\textwidth}
\centering
\begin{tikzpicture}[line cap=round,line join=round,>=triangle 45,x=1cm,y=1cm]
\clip(-3,-3) rectangle (3.5,2.8);
\draw [line width=1pt] (0,0) circle (2.0987428108210455cm);
\draw [line width=1pt] (-1.8361832612315765,1.0164410544374896)-- (2.0987428108210455,0);
\draw [rotate around={-14.483620895443744:(0.13145057576063224,0.5081764071741245)},line width=1pt,color=DarkGreen,fill=DarkGreen,fill opacity=0.15] (0.13145057576063224,0.5081764071741245) ellipse (2.2350709328755056cm and 0.9303903518687997cm);
\draw [line width=1pt,color=DarkGreen] (0.13127977479473452,0.5082205272187448)-- (0.36397340636002584,1.4090423049789844);
\draw [line width=1pt,color=DarkGreen] (0.13127977479473452,0.5082205272187448)-- (2.295489114423676,-0.050822052721874456);
\begin{scriptsize}
\draw [fill=black] (-1.8361832612315765,1.0164410544374896) circle (1pt);
\draw[color=black] (-1.7,.8) node {$A$};
\draw [fill=black] (2.0987428108210455,0) circle (1pt);
\draw[color=black] (1.9351736038303906,-0.12119276492332029) node {$B$};
\draw[color=DarkGreen] (0.4,.9) node {$h$};
\draw[color=DarkGreen] (1,-.1) node {$\frac{(1+\eta)|A-B|}{2}$};
\end{scriptsize}
\end{tikzpicture}
\end{minipage}

\caption{Illustration of the proof of Lemma~\ref{lem:height-bound}; Left: Reduction to curves consisting of two line segments $ { \rm conv } ( A , C ) $ and $ { \rm conv } ( C , B ) $; Right: Admissible points $ C $ subject to the constraint $ | A - C | + | B - C | \leq ( 1 + \eta ) | A - B | $.}\label{fig:lem35}
\end{figure}%
The idea behind the characterization \eqref{eqn:height-bound-2-b} is illustrated in Figure~\ref{fig:lem35} (left); its rigorous justification is deferred to the end of the proof.

Equipped with \eqref{eqn:height-bound-2-b} we are able to conclude by elementary geometric arguments. By convexity, we know that the maximum point $ C $ must lie on the boundary of the admissible set, i.e.~satisfy
\begin{align*}
	 | A - C | + | B - C | = ( 1 + \eta ) | A - B | ;
\end{align*}
that is, $ C $ lies on the ellipse with focal points $ A $ and $ B $ and semi-major axis of length $ \frac{ 1 }{ 2 } ( 1 + \eta ) | A - B | $, see Figure~\ref{fig:lem35} (right). Moreover, by optimality of $ C $ we know that $ { \rm dist } ( C, { \rm conv } ( A , B ) ) $ is precisely the length of the semi-minor axis and thus by Pythagoras
\begin{align*}
	 { \rm dist } ( C, { \rm conv } ( A , B ) ) ^2 
	 =  \frac{ 1 }{ 4 } ( 1 + \eta )^2 | A - B |^2 - \frac{ 1 }{ 4 } | A - B |^2
	 =  \frac{ 1 }{ 4 } ( \eta^2 + 2 \eta ) |A-B|^2 .
\end{align*}
Since $ \eta $ is small we have $  \frac{ 1 }{ 4 } ( \eta^2 + 2 \eta ) \lesssim \eta $, which, inserted in \eqref{eqn:height-bound-2-b}, implies the claim.

Finally, here comes the justification of \eqref{eqn:height-bound-2-b}: If for some admissible $ \gamma $ the height $ h $ is attained at a point $ C \in \gamma $, we know that the curve $ \gamma' = { \rm conv }(A,C) \cup { \rm conv }(C,B) $, satisfies
\begin{align*}
	h
	= \sup_{ y \in \gamma ' } { \rm dist } ( y, { \rm conv } ( A, B ) )
	= { \rm dist } ( C, { \rm conv } ( A , B ) ) .
\end{align*}
Moreover, by the triangle inequality we have $ { \rm Length } ( \gamma ' ) \leq  { \rm Length } ( \gamma ) $ and thus also
\begin{align*}
	{ \rm Length } ( \gamma ' ) \leq { \rm Length } ( \gamma ) \leq ( 1 + \eta ) | A - B | .
\end{align*}
Since $ { \rm Length } ( \gamma ' ) = | A - C | + | C - B | $, the value of $ h $ for our choice of $ \gamma $ can always be recovered by a curve joining two line segments. Therefore equality in \eqref{eqn:height-bound-2-b} holds.
\end{proof}

The next lemma is well-known in the regularity theory of minimal surfaces, and there are many versions of it. In two dimensions, there is a clear picture to have in mind: zooming in on sufficiently small scales, any minimizer of the perimeter looks like a line. Hence it evenly divides any ball centered on a point of the surface into two pieces of equal volume, which is therefore comparable to the volume of the ball. In terms of estimates, this behavior carries over to $ \eta $-minimizers:

\begin{lemma}[Density estimates]\label{lem:density-bounds}
Let $ m $ be an $ \eta $-minimizer in $ B_R $ such that $ 0 \in J_{ m } $, then
\begin{equation}\label{eqn:volume-density-estimate}
	\frac{ 1 }{ ( 2 + \eta )^2 } | B_r |
	\leq \frac{1}{2} \int_{ B_r } | m \pm 1 |
	\leq \left( 1 - \frac{ 1 }{ ( 2 + \eta )^2 } \right) | B_r | 
\end{equation}
for every $ 0 < r < R $. Furthermore, we have
\begin{equation}\label{eqn:surface-density-estimate}
	\frac{r}{(2 + \eta) } \lesssim \int_{ B_r } | \nabla m |
	\lesssim ( 1 + \eta ) r.
\end{equation}
\end{lemma}

The proof is quite standard, see for example \cite[Theorem 21.11]{Maggi} or \cite[Proposition 1.5]{AP99}. We include it for the reader's convenience. We also want to mention the reference \cite[Proposition 3.1]{GNR} for another argument in the anisotropic (but ``uniformly elliptic'') setting.

\begin{proof}
We split the proof in two steps.

\medskip

\textit{Step 1 (Argument for \eqref{eqn:volume-density-estimate}).} We consider the function
\begin{align*}
I( R ) = \frac{1}{2} \int_{ B_R } | 1 - m | = \frac{1}{2} \int_0^R \left( \int_{ \partial B_r } | 1 - m | \right) dr 
\quad \text{with} \quad
I'( R ) = \frac{ 1 }{ 2 } \int_{ \partial B_R } | 1 - m | .
\end{align*}
Note that $0<I(r)<|B_r|$ for every $0<r<R$, since $0\in J_m$.
Below we will argue that
\begin{align}\label{eqn:density-estimate-1}
	\text{for a.e.} \quad 0 < r < R \quad \text{it holds} \quad \frac{ | B_1 |^{\frac{1}{2}} }{ 2 + \eta } \leq \frac{1}{2} \frac{ I'(r) }{ I(r)^{ \frac{1}{2} } } ,
\end{align}
which we can use to conclude \eqref{eqn:volume-density-estimate}: Since the r.h.s.~of \eqref{eqn:density-estimate-1} is equal to $  \frac{1}{2} I ' ( r ) / I( r )^{ \frac{1}{2} }  = \frac{ d }{ d r } I( r )^{\frac{1}{2}} $, we can integrate \eqref{eqn:density-estimate-1} to obtain the claimed lower bound for $ \int_{ B_1 } | 1 - m | $ in \eqref{eqn:volume-density-estimate}. The same argument yields the lower bound for $ 1 - m $ replaced by $ - 1 - m $. Using the fact that
\begin{align*}
| 1 - m | + | - 1 - m | = 2,
\end{align*}
we also conclude the upper bounds in \eqref{eqn:volume-density-estimate}.

Now comes the argument for \eqref{eqn:density-estimate-1}: Note that for a.e. $ 0 < r < R $ we know the inner and outer trace of $ m $ on $ \partial B_r $ are equal. Let us fix such a radius and compare $ m $ to the configuration that is equal to $ m $ in $ B_r^c $ and $ 1 $ in $ B_r $ to get
\begin{align*}
\int_{ \overline{ B_r } } | \nabla m | \leq ( 1 + \eta ) \int_{ \overline{B_r} } | \nabla ( \boldsymbol{1}_{ B_r } + m \boldsymbol{1}_{ \R^2 \setminus B_r } ) | \leq(  1 + \eta ) \int_{ \partial B_r } | 1 - m |.
\end{align*}
Adding $ \int_{ \partial B_r } | 1 - m | $ to both sides, we get
\begin{align*}
\int | \nabla (  ( 1 - m ) \boldsymbol{1}_{ B_r }  ) | \leq ( 2 + \eta ) \int_{ \partial B_r } | 1 - m | .
\end{align*}
The isoperimetric inequality (relating the volume and perimeter of the set $ \{ 1 - m = 2 \}$) and the previous estimate imply
\begin{align*}
	2 | B_1 |^{ \frac{1}{2} } \left( \frac{ 1 }{ 2 } \int_{ B_r } | 1 - m | \right)^{ \frac{1}{2} }
	\leq \frac{ 1 }{ 2 } \int | \nabla ( ( 1 - m ) \boldsymbol{1}_{ B_r } ) |
	\leq ( 2 + \eta ) \frac{ 1 }{ 2 } \int_{ \partial B_r } | 1 - m |. 
\end{align*}
By definition of $ I ( r ) $, and since $I(0) = 0$ and $I(r)>0$ for all $0<r<R$, this gives \eqref{eqn:density-estimate-1}.

\medskip

\textit{Step 2 (Argument for \eqref{eqn:surface-density-estimate}).} Observe that by Poincaré's inequality
\begin{equation}\label{eqn:density-estimate-3-prev}
\int_{ B_r } | m - \fint_{B_r} m |^2
\lesssim \left( \int_{ B_r } | \nabla m | \right)^2,
\end{equation}
where the implicit constant is by scaling independent of $ r $. Since $ m $ only takes values $ \pm 1 $, $ m - \fint_{B_r} m $ is a simple function and we may rewrite the left hand side as
\begin{align}
\int_{ B_r } | m - \fint_{B_r} m |^2
&= | 1 - \fint_{B_r} m |^2 \, | B_r \cap \{ m = 1 \} |+ | 1 + \fint_{B_r} m |^2 \, | B_r \cap \{ m = - 1 \} |  \nonumber \\
&= \frac{1}{2} | 1 - \fint_{B_r} m |^2 \, \int_{B_r} | 1 + m | + \frac{1}{2} | 1 + \fint_{B_r} m |^2 \, \int_{B_r} | 1 - m | . \label{eqn:density-estimate-3}
\end{align}
Since $ -1 \leq \fint_{B_r} m \leq 1 $, we have
\begin{align*}
\frac{1}{2} | 1 - \fint_{B_r} m |^2 + \frac{1}{2} | 1 + \fint_{B_r} m |^2 
= 1 + ( \fint_{B_r} m )^2 \geq 1 ,
\end{align*}
so that \eqref{eqn:density-estimate-3-prev} and \eqref{eqn:density-estimate-3}, which we estimate by \eqref{eqn:volume-density-estimate} (recall that $0 \in J_m$), combine to the lower bound in \eqref{eqn:surface-density-estimate}.

The upper bound \eqref{eqn:surface-density-estimate}, follows from comparing $ m $ to the configuration that is equal to $ 1 $ in $ B_r $. More precisely, we have
\begin{align}\label{eqn:surface-density-estimate-precise}
\int_{ B_r } | \nabla m | \leq ( 1 + \eta  ) \int_{ \overline B_r } | \nabla ( \boldsymbol{1}_{ B_r } + m \boldsymbol{1}_{ \R^2 \setminus B_r } ) |
\leq 2 ( 1 + \eta ) | \partial B_r |,
\end{align}
as claimed.
\end{proof}

It was convenient in the above argument that our excess is defined in terms of line configurations. What we really need to control in the end is the tilt of the normal after readjusting the line configuration. The next lemma provides the necessary estimate for this. 

\begin{lemma}\label{lem:normal-estimate}
Consider two lines $ \mathscr{L} $ and $ \mathscr{L}' $ such that
\begin{align*}
{ \rm dist }_{ \partial B_1 } ( \mathscr{L} \cap \partial B_1 , \mathscr{L}' \cap \partial B_1 ) \leq \frac{ 1 }{ 4 } ,
\quad
\mathscr{L} \cap \overline{ B_{ \frac{1}{4} } } \neq \emptyset, 
\end{align*}
where $ { \rm dist }_{ \partial B_1 } $ denotes the two-sided Hausdorff distance on $ \partial B_1 $; then
\begin{align*}
| \nu - \nu' | \lesssim { \rm dist }_{ \partial B_1 } ( \mathscr{L} \cap \partial B_1 , \mathscr{L}' \cap  \partial B_1 ) ,
\end{align*}
where $ \nu $, $ \nu' $ are normals to $ \mathscr{L} $,~resp.~$ \mathscr{L}' $, with $ \nu \cdot \nu' > 0 $.
\end{lemma}

Notice that since we assume $ d \coloneq { \rm dist }_{ \partial B_1 } ( \mathscr{L} \cap  \partial B_1 , \mathscr{L}' \cap  \partial B_1 ) \leq \frac{ 1 }{ 4 }$, we have that $ d $ equals the maximum of the distances of closest endpoints of $ \mathscr{L} $ and $ \mathscr{L} ' $ on $ \partial B_1 $.

\begin{proof}[Proof of Lemma \ref{lem:normal-estimate}.]
Let us write $ d \coloneq { \rm dist }_{ \partial B_1 } ( \mathscr{L} \cap  \partial B_1 , \mathscr{L}' \cap \partial B_1  ) $. W.l.o.g.~we may assume that $ \mathscr{L} $ is not parallel to $ \mathscr{L}' $. Since $ d \leq \frac{1}{4} $, we know that
\begin{align*}
\mathscr{L} \cap \overline{ B_{\frac{1}{2}} } \neq \emptyset
\quad \text{and} \quad
\mathscr{L}' \cap \overline{ B_{\frac{1}{2}} } \neq \emptyset.
\end{align*}
Based on that, we distinguish two cases: either $ \mathscr{L} $ and $ \mathscr{L}' $ intersect in $ \overline{ B_{\frac{1}{2}} } $ (case 1), or they meet in $ \R^2 \setminus \overline{B_{ \frac{1}{2} }} $ (case 2).

\medskip

\textit{Case 1.} Let us assume that $ \{ p \} = \mathscr{L} \cap \mathscr{L}' \subseteq \overline{ B_{ \frac{1}{2} } } $. Furthermore, we denote by $ A \in \mathscr{L} \cap \partial B_1 $ and $ A' \in \mathscr{L}' \cap \partial B_1 $ two points in $ \mathscr{L} $, resp.~$ \mathscr{L}' $, lying on the boundary $ \partial B_1 $ such that $ ( A - p ) \cdot ( A' -p ) > 0 $. Such points exists since $ d < \frac{ \pi }{ 2 } $ and they satisfy
\begin{equation}\label{eqn:normal01}
\sphericalangle( \nu, \nu' )
= \sphericalangle( A - p, A' - p ).
\end{equation}
By an elementary geometric consideration, we have
\begin{align*}
\sin \sphericalangle( A - p, A' - p )
= \frac{ { \rm dist }(A', \mathscr{L}) }{ | p - A' | }
= \frac{ { \rm dist }(A, \mathscr{L}') }{ | p - A | } ,
\end{align*}
so that
\begin{equation}\label{eqn:normal02}
\sin \sphericalangle( A - p, A' - p )
\leq \min \left\{ \frac{ { \rm dist }_{ \partial B_1 } (A', \mathscr{L}) }{ | p - A' | },
 \frac{ { \rm dist }_{ \partial B_1 } (A, \mathscr{L}') }{ | p - A | }
 \right\}
 \leq 2 d.
\end{equation}
Since
\begin{align*}
| \nu - \nu' |^2
= 2 - 2 \cos \sphericalangle( \nu, \nu') 
= 4 \sin^2 \frac{ \sphericalangle( \nu, \nu' ) }{2},
\end{align*}
we may combine \eqref{eqn:normal01} and \eqref{eqn:normal02} to get
\begin{align*}
| \nu - \nu' |
\leq 2 \sin \sphericalangle( A - p, A' - p )
\leq 4 d,
\end{align*}
which proves the claim if $ \mathscr{L} $ and $ \mathscr{L}' $ meet in $ \overline{ B_{\frac{1}{2}} } $.

\medskip

\textit{Case 2.} We now consider the case that $ \mathscr{L} $ and $ \mathscr{L}' $ intersect in $ \R^2 \setminus \overline{ B_{ \frac{1}{2} } } $, and both $ \mathscr{L} $ and $ \mathscr{L}' $ go through $ \overline{ B_{ \frac{1}{2} } } $.  As above, we write $ \{ p \} = \mathscr{L} \cap \mathscr{L}' $. Since $ \mathscr{L} $ passes through $ B_{ \frac{1}{4} } $, there exists points $ A , B \in \mathscr{L} \cap \partial B_{ \frac{1}{2} } $ such that $ B $ lies between $ A $ and $ p $. Observe that (again since $ \mathscr{L} $ intersects $ B_{ \frac{1}{4} } $) we have $ | A - B | \gtrsim 1 $. Furthermore, let us note that via parallel translation
\begin{align*}
	{ \rm dist } ( A, \mathscr{L}' )
	\leq { \rm dist } ( \mathscr{L} \cap \partial B_1 , \mathscr{L}' )
	\leq  { \rm dist } ( \mathscr{L} \cap \partial B_1 , \mathscr{L}' \cap \partial B_1 )
	\leq  { \rm dist }_{ \partial B_1 } ( \mathscr{L} \cap \partial B_1 , \mathscr{L}' )
	= d .
\end{align*}
Using these two facts, we estimate
\begin{align*}
\sin \sphericalangle( \nu, \nu' ) 
= \frac{ { \rm dist } ( A, \mathscr{L}' ) }{ | A - p | }
\leq  \frac{ d }{ | A - B | }
\lesssim d 
\end{align*}
As in Case~1, this yields
\begin{align*}
| \nu - \nu' | \lesssim d,
\end{align*}
which proves the lemma when $ \mathscr{L} $ and $ \mathscr{L}' $ intersect in $ \R^2 \setminus \overline{ B_{ \frac{1}{2} } } $.
\end{proof}

After the preparations above, we are finally ready to prove the one step improvement (Proposition~\ref{prop:iteration-step}) of the Campanato iteration in the proof of Theorem~\ref{thm:mr}. As announced in Section~\ref{sec:strategy-of-proof}, the final goal is to use the small excess to find a ``good'' radius so that the jump set of the minimizer has two intersection points with the circle of that radius. Afterwards, we connect these points to construct an improved configuration.

\begin{proof}[Proof of Proposition \ref{prop:iteration-step}.] By Remark \ref{rmk:eta-minimizer} we assume w.l.o.g.~that $ R = 1 $. We proceed in three steps.

\medskip

\textit{Step 1.} Let us fix some radius $ r_0 $ (to be chosen later). We claim that there exist $ \varepsilon_0, \eta_0 > 0 $ such that for every $ \varepsilon \le \varepsilon_0 $, $ \eta \leq \eta_0 $, and $ m $, $ m_{ \rm line } $ satisfying \eqref{eqn:prop:iteration-step-assumption}, the jump set of $ m_{ \rm line } $ intersects $ B_{ r_0 } $.

Indeed, observe that by using the pointwise inequality $ | | 1 - m_{ \rm line } | - | 1 - m | | \leq | m - m_{ \rm line } | $, we obtain after integration
\begin{align*}
\frac{1}{2} \int_{ B_{ r_0 } } | 1 - m | - \frac{1}{2} \int_{B_1} | m - m_{ \rm line } |
\leq \frac{1}{2} \int_{ B_{ r_0 } } | 1 - m_{ \rm line } |
\leq \frac{1}{2} \int_{ B_{ r_0 } } | 1 - m | + \frac{1}{2} \int_{B_1} | m - m_{ \rm line } |.
\end{align*}
By \eqref{eqn:volume-density-estimate} this estimate implies
\begin{align*}
\frac{ |  B_{ r_0 } | }{ ( 2 + \eta )^2 } - \frac{1}{2} \varepsilon
\leq \frac{1}{2} \int_{ B_{ r_0 } } | 1 - m_{ \rm line } |
\leq ( 1 - \frac{ 1 }{ ( 2 + \eta )^2 } ) |  B_{ r_0 } | + \frac{1}{2} \varepsilon.
\end{align*}
Whenever $ \eta, \varepsilon \ll_{ r_0 } 1 $, we may infer
\begin{align*}
c | B_{ r_0 } |
\leq \frac{1}{2} \int_{ B_{ r_0 } } | 1 - m_{ \rm line } |
\leq ( 1 - c ) | B_{ r_0 } |
\quad \text{for~some} \quad \frac{1}{8} < c < \frac{1}{4}.
\end{align*}
We conclude with
\begin{align*}
\frac{1}{2} \int_{ B_{ r_0 } } | 1 - m_{ \rm line } | = | B_{r_0} \cap \{ m_{ \rm line } = - 1 \} |
\end{align*}
that the jump set of $ m_{ \rm line } $ must intersect with $ B_{ r_0 } $.

\medskip

\textit{Step 2.} We claim that for some radius $ \frac{1}{16} < r < \frac{5}{16} $ the jump set of $ m $ has two intersections with $ \partial B_r $ and
\begin{equation}\label{eqn:linf-bdry-close}
{ \rm dist }_{ \partial B_r } ( J_{ m }, J_{ m_{ \rm line } } ) \lesssim \int_{B_1} | m - m_{ \rm line } |.
\end{equation}

To this end, we choose a radius $ \frac{1}{16} < r < 1 $ according to Lemma \ref{lem:few-jumps}. Hence, $ m \lfloor_{ \partial B_r } $ has either zero or two jumps on $ \partial B_r $ and
\begin{equation}\label{eqn:l1-bdry-close-iii}
\int_{ \partial B_r } | m - m_{ \rm line } |
\lesssim \int_{ B_1 } | m - m_{ \rm line } |
\lesssim \varepsilon.
\end{equation}
We will proceed in two steps: first, we argue that $ m \lfloor_{ \partial B_r } $ has exactly two jumps provided $ \varepsilon \ll 1 $. Secondly we upgrade \eqref{eqn:l1-bdry-close-iii} to \eqref{eqn:linf-bdry-close}.

We start with the former: $ m $ can't have zero jumps on $ \partial B_r $. By the first step, where for now $ r_0 = \frac{1}{32} $ would be sufficient, but for later purposes (in view of Lemma \ref{lem:normal-estimate}) $ r_0 = \frac{1}{64} $ is convenient, we may assume that the jump set of $ m_{ \rm line } $ intersects non-trivially with $ B_{ \frac{1}{64} } $ whenever $ \eta < \eta_0 $ and $ \varepsilon < \varepsilon_0 $. Since the jump set of $ m_{ \rm line } $ intersects $ B_{ \frac{1}{32} } $, we can lower bound
\begin{align*}
| \{  m_{ \rm line } = -1  \} \cap \partial B_{ \frac{1}{16} }  | \gtrsim 1.
\end{align*}
Hence, we can bound
\begin{align*}
| \{  m_{ \rm line } = -1  \} \cap \partial B_r |
\geq 16 r \, | \{  m_{ \rm line } = -1  \} \cap \partial B_{ \frac{1}{16} }  |
\gtrsim r.
\end{align*}
Similarly, we may argue for the set $ \{   m_{ \rm line } = 1 \} $ to conclude
\begin{equation}\label{eqn:density-bound-m-per-i}
| \{ m_{ \rm line } = -1  \} \cap \partial B_r | \gtrsim r
\quad \text{and} \quad
| \{ m_{ \rm line } = 1  \} \cap \partial B_r | \gtrsim r.
\end{equation}
Furthermore, we may write
\begin{equation}\label{eqn:l1-bdry-close}
| \{ m = 1, m_{ \rm line } = -1 \} \cap \partial B_r | + | \{ m = -1, m_{ \rm line } = 1 \} \cap \partial B_r |
= \frac{1}{2} \int_{ \partial B_r } | m - m_{ \rm line } |
\lesssim \varepsilon.
\end{equation}
If $ m \lfloor_{ \partial B_r } $ had zero jumps on $ \partial B_r $, say $ m $ is constantly equal to $ 1 $, then
\begin{align*}
| \{ m_{ \rm line } = -1 \} \cap \partial B_r |
= | \{ m = 1, m_{ \rm line } = -1 \} \cap \partial B_r |
\lesssim \varepsilon.
\end{align*}
Since $ r > \frac{1}{16} $, this contradicts \eqref{eqn:density-bound-m-per-i} if $ \varepsilon \ll 1 $. Hence, by our choice of the radius $ r $ through Lemma \ref{lem:few-jumps} $ m \lfloor_{ \partial B_r } $ has exactly two jumps 

We now argue, that \eqref{eqn:l1-bdry-close-iii} implies \eqref{eqn:linf-bdry-close}. The jump set of $ m_{ \rm line } $ has two intersection points with $ \partial B_r $, say $ A', B' \in J_{ m_{ \rm line } } \cap \partial B_r $. In this notation, we may write
\begin{align*}
{ \rm dist }_{ \partial B_r }( A', B' ) = \min \{  | \{ m_{ \rm line } = - 1 \} \cap \partial B_r | ,  | \{ m_{ \rm line } = 1 \} \cap \partial B_r | \},
\end{align*}
so that \eqref{eqn:density-bound-m-per-i} turns into
\begin{equation}\label{eqn:density-bound-m-per-ii}
1 \lesssim r \lesssim { \rm dist }_{ \partial B_r }( A', B' ) \lesssim r \lesssim 1.
\end{equation}
Since also the jump set of $ m $ has exactly two intersection points $ A $, $ B $ with $ \partial B_r $, we can write
\begin{align*}
{ \rm dist }_{ \partial B_r }(A, A') = | \{ m = 1, m_{ \rm line } = - 1 \} \cap \partial B_r |, \quad
{ \rm dist }_{ \partial B_r }(B, B') = | \{ m = -1, m_{ \rm line } = 1 \} \cap \partial B_r |.
\end{align*}
Hence \eqref{eqn:l1-bdry-close} becomes
\begin{equation}\label{eqn:l1-bdry-close-ii}
{ \rm dist }_{ \partial B_r }(A, A') + { \rm dist }_{ \partial B_r }(B, B')
\lesssim \frac{1}{2} \int_{ \partial B_r } | m - m_{ \rm line } |.
\end{equation}
Now recall that $ r > \frac{1}{16} $, $ \varepsilon \ll 1 $, so that \eqref{eqn:l1-bdry-close-ii} and \eqref{eqn:density-bound-m-per-ii} show
\begin{align*}
\min \{ { \rm dist }_{ \partial B_r }(B, A') ,  { \rm dist }_{ \partial B_r }(A, B') \} \gtrsim 1, \quad
\max \{ { \rm dist }_{ \partial B_r }(A, A') , { \rm dist }_{ \partial B_r }(B, B') \} \lesssim \varepsilon,
\end{align*}
so that \eqref{eqn:l1-bdry-close-ii} turns into \eqref{eqn:linf-bdry-close}.

\medskip

\textit{Step 3.} We now construct the competitor $ m_{ \rm line }' $. To this end, we denote by $ \mathscr{L} $ the line along which $ m_{ \rm line } $ jumps and by $ \nu $ the outwards pointing normal $ \nu $ to the half space $ \{ m_{ \rm line } = 1 \} $. Furthermore, we let $ \mathscr{L}' $ denote the line that contains the two intersection points of the jump set of $ m $ and $ \partial B_r $ where $ r $ was chosen in Step 2. We can now define $ m_{ \rm line }' $ to be the configuration with jump set $ \mathscr{L}' $ such that the halfspace $ \{ m_{ \rm line }' = 1 \} $ has outwards pointing normal $ \nu' $ with $ \nu \cdot \nu' \ge 0 $. 

First, we argue that \eqref{eqn:tilt-halfspace} holds true. Since $ J_m $ has exactly two intersections with $ B_r $, we may decompose it, up to sets of measure $ \mathscr{H}^1 $, into curves
\begin{align*}
J_m \cap \overline{ B_r } = \bigcup_{ k \in \N } \gamma_k,
\end{align*}
where $ \gamma_1 $ intersects $ \partial B_r $ in two points, and every $ \gamma_k $ is a simple closed curve in $ B_r $. Recall\footnote{Strictly speaking, this relies on the small-scale regularity of a minimizer in the sense of Definition \ref{defn:in-energy-minimizer}, which is justified by our later application. For genuine $ \eta $-minimizers, one can proceed via approximation through smooth boundaries, see e.g.~\cite[Theorem 3.42]{AmbrosioFuscoPallara}.}
that since $ \xi $ is bounded on small scales, see \eqref{eqn:noise-assumption}, the small-scale regularity theory, see \cite[Example 21.2 \& Theorem 21.8]{Maggi}, shows that these curves are (sufficiently) smooth. Hence there is no subtlety in the above decomposition.

Let us now denote by $ { \rm int } ( \gamma_k ) $ the domain enclosed by the curve $ \gamma_k $. Hence, we can bound
\begin{align}\label{eqn:tilt-half-space-estimate-1}
\begin{aligned}
\int_{ B_r } | m - m_{ \rm line }' | 
&\leq \int_{ B_r \cap J_{ m _ { \rm line } ' } } \sup_{ y \in \gamma_1 } | x - y |  \, d \mathscr{H}^1 (x)  + \sum_{ k \ge 2 } | { \rm int } ( \gamma_k ) | \\
&\leq r \, { \rm dist } ( J_{ m_{ \rm line }' }, \gamma_1 ) + \sum_{ k \ge 2 } | { \rm int } ( \gamma_k ) |
\end{aligned}
\end{align}
The first term on the r.h.s is controlled via Lemma \ref{lem:height-bound}, which yields
\begin{align}\label{eqn:tilt-half-space-estimate-2}
{ \rm dist } ( J_{ m_{ \rm line }' }, \gamma_1 ) \lesssim r \sqrt{ \eta }.
\end{align}
To the latter, we may apply the isoperimetric inequality to obtain
\begin{align*}
\sum_{ k \ge 2 } | { \rm int } ( \gamma_k ) | 
\lesssim \sum_{ k \ge 2 } | \gamma_k |^2 
\leq \left( \sum_{ k \ge 2 } | \gamma_k | \right)^2
= \left( \int_{ B_r } | \nabla m | - | \gamma_1 | \right)^2
\leq \left( \int_{ B_r } | \nabla m | - \int_{ B_r } | \nabla m_{ \rm line } ' | \right)^2.
\end{align*}
The $ \eta $-minimality of $ m $ therefore implies upon comparing $ m $ to $ m_{ \rm line }' $ that
\begin{align}\label{eqn:tilt-half-space-estimate-3}
\sum_{ k \ge 2 } | { \rm int } ( \gamma_k ) | \leq \eta^2  \left( \int_{ B_r } | \nabla m_{ \rm line } ' | \right)^2 \leq r^2 \eta^2.
\end{align}
Together \eqref{eqn:tilt-half-space-estimate-1}, \eqref{eqn:tilt-half-space-estimate-2}, and \eqref{eqn:tilt-half-space-estimate-3} imply \eqref{eqn:tilt-halfspace}, as desired.

Finally, we show that also \eqref{eqn:tilt-normal} holds. By \eqref{eqn:linf-bdry-close}, we have
\begin{align*}
{ \rm dist }_{ \partial B_r } ( \mathscr{L} , \mathscr{L}' )
= { \rm dist }_{ \partial B_r } ( J_m, \mathscr{L}' )
\lesssim \int_{ B_1 } | m - m_{ \rm line } |
\lesssim \varepsilon,
\end{align*}
so that Lemma \ref{lem:normal-estimate} (recall that in Step 2 we intentionally chose to have $ \emptyset \neq \mathscr{L} \cap \overline{ B_{ \frac{1}{64} } } \subseteq \mathscr{L} \cap \overline{ B_{ \frac{r}{4} } }  $ with $ \frac{1}{16} < r < 1 $) implies
\begin{align*}
| \nu - \nu' | \lesssim \int_{ B_1 } | m - m_{ \rm line } |,
\end{align*}
proving \eqref{eqn:tilt-normal}.
\end{proof}

Although at this point we have proven the main deterministic ingredient, namely the one step improvement in Proposition~\ref{prop:iteration-step}, the next remark will be crucial. It relates the normal constructed in the proof above to the average normal in Theorem~\ref{thm:mr}.

\begin{remark}\label{rmk:strong-exc-bound} 
Let us remark that in the situation of Proposition \ref{prop:iteration-step}, we could even bound
\begin{equation}\label{eqn:strong-exc-bound}
\frac{1}{r} \int_{ B_r } | \nu - \nu' |^2 | \nabla m |
\lesssim \eta,
\end{equation}
where $ r $ is the same radius as above. The left-hand side of \eqref{eqn:strong-exc-bound} provides a bound on the \emph{excess} $\inf_{\nu'\in\mathbb{S}^1} \frac{1}{r} \int_{ B_r } | \nu - \nu' |^2 | \nabla m |$ in the regularity theory of minimal surfaces. 
Indeed, since $ m_{ \rm line }' $ has a jump set equal to a line, we have
\begin{align*}
\nu' = \frac{ \int_{ B_r } \nabla m_{ \rm line }' }{ \int_{ B_r } | \nabla m_{ \rm line }' | },
\end{align*}
so that in the situation of Step 3 in the proof of Proposition \ref{prop:iteration-step}, we can compute
\begin{align*}
\frac{1}{2} \int_{ B_r } | \nu - \nu' |^2 | \nabla m |
&= \int_{ B_r } | \nabla m | - \int_{ B_r } \nu' \cdot \nabla m \\
&= \int_{ B_r } | \nabla m | - \int_{ B_r } | \nabla m_{ \rm line }' |
- \int_{ B_r } \nu' \cdot \nabla ( m - m_{ \rm line }' ).
\end{align*}
The last integral vanishes, because $ m = m_{ \rm line }' $ on $ \partial B_r $. Indeed, let $ J $ denote a counterclockwise rotation by $ \frac{\pi}{2} $, so that $ J\nu $ of a normal $ \nu $ is a tangent. Since the jump sets of $ m $ and $ m_{ \rm line }' $ intersect $ \partial B_r $ in two common points, we learn from integrating their tangents that
\begin{align*}
\int_{ B_r } J \nu | \nabla m | = \int_{ B_r } J \nu' | \nabla m_{ \rm line } ' |
\quad { \rm so ~ that } \quad
\int_{ B_r } \nabla m =  \int_{ B_r } \nu | \nabla m | = \int_{ B_r } \nu' | \nabla m_{ \rm line } ' | = \int_{ B_r } \nabla m_{ \rm line } ' .
\end{align*}
Taking the inner product with $ \nu' $ yields $ \int_{ B_r } \nu' \cdot \nabla ( m - m_{ \rm line }' ) = 0 $. Therefore, by $ \eta $-minimality, the above identity simplifies to
\begin{align*}
\frac{1}{2} \int_{ B_r } | \nu - \nu' |^2 | \nabla m |
= \int_{ B_r } | \nabla m | - \int_{ B_r } | \nabla m_{ \rm line }' |
\leq \eta \int_{ B_r } | \nabla m_{ \rm line }' |,
\end{align*}
which implies \eqref{eqn:strong-exc-bound}.\end{remark}

\section{Proof of the stochastic ingredients}

In this section, we make the stochastic ingredients needed for the proof of Theorem \ref{thm:mr} rigorous. In particular, we prove Propositions~\ref{thm:talagrand-bound}, \ref{prop:concentration-scales}, and~\ref{prop:concentration-in-space} that we introduced in Section~\ref{sec:strategy-of-proof}. 
Since the first one is an adaption of the results in \cite{DW} to our continuum setting, we start with the latter two in Section~\ref{sec:pf-concentration-arguments}. The proof of Proposition~\ref{thm:talagrand-bound} is then given in Section~\ref{sec:pf-talagrand-bound} together with a short primer on the chaining method. 
As in Section \ref{sec:regularity-theory}, we try to give (to some reasonable extent) self-contained arguments.  

\subsection{Concentration arguments}\label{sec:pf-concentration-arguments}

The proofs of this section rely on concentration properties of Gaussian random variables. More precisely, we use that the application of Lipschitz-continuous functions to a Gaussian random variable enjoys good concentration properties. In particular, this implies concentration estimates for the supremum of a Gaussian processes. Indeed, for a finite-dimensional centered Gaussian vector $ \xi \in \R^N $, it is known that
\begin{align}\label{eqn-gaussian-concentration-plausibel-1}
	\P \big\{ \sup_{ i = 1, \hdots , N } \xi_i  > \E[ \sup_{ i = 1, \hdots , N } \xi_i ] + t \big\} \leq \exp ( - \frac{t^2}{2 \sigma^2} ),
	\quad { \rm where } \quad
	\sigma^2 = \sup_{ i = 1, \hdots, N }  \E \xi_i^2 
\end{align}
see e.g.~\cite[Theorem 2.1.1 \& Lemma 2.1.6]{AdlerTaylor-RF-and-Geometry} (or \cite[Lemma 3.1 \& (3.2)]{LedouxTalagrand} for a closely related statement). 
This naturally generalizes to infinite-dimensional Gaussian vectors. 
To treat the two classes of noise we introduced in Section \ref{sec:assumptions-on-noise} in a unified way, it is convenient to instead argue via the abstract concentration argument in \cite[Theorem 4.5.6 \& Theorem 4.5.7]{BogachevGM}.

To make this precise for our infinite dimensional Gaussian $ \xi $, we need to specify a norm to measure Lipschitz-continuity. In fact, it turns out, see \cite{BogachevGM}, that Lipschitz-continuity w.r.t.~Cameron-Martin shifts is the right notion in infinite dimensions. For convenience, we recall the notion of a Cameron-Martin space from \cite{BogachevGM}: viewing $ \xi $ as a $ \mathscr{S}'(\R^2) $-valued\footnote{$ \mathscr{S}'$ denotes the space of Schwartz distributions.} random variable, the definition of the Cameron-Martin space $ H $ (of the law of $ \xi $) is
\begin{align}\label{eqn:cm-pf-03}
H = \{ T \in \mathscr{S}'(\R^2) ~|~ | T |_{ H } < \infty \}
\quad { \rm with } \quad
| T |_{ H } = \sup \{ \ell(T)  ~|~ \ell \in (  \mathscr{S}'(\R^2) )', ~ \E[ \ell(\xi)^2 ] \leq 1 \} ,
\end{align}
see \cite[Definition 2.2.7]{BogachevGM}. Using this definition we can prove the following lemma.\footnote{In the statement of Lemma~\ref{lemma:cameron-martin-continuity}, we are intentionally a bit vague with notation: we write $ \delta\xi $ for an element in the Cameron-Martin space $ H $, which we think of as a distribution, and use the same symbol for its defining function,~i.e.~$ \delta\xi(f) = \int \delta\xi(x) f(x) dx $.}

\begin{lemma}\label{lemma:cameron-martin-continuity}
	Let $ \xi $ denote either of the two noises in Section \ref{sec:assumptions-on-noise}. The map $ \xi \mapsto \frac{1}{| \partial M |} \int_{ M } \xi $ is $ 2 \sqrt{ \pi } $-Lipschitz-continuous under Cameron-Martin shifts; that is 
	\begin{align}\label{eqn:cm-pf-02}
	\left| \frac{1}{| \partial M |} \int_{ M } ( \xi + \delta \xi ) - \frac{1}{| \partial M |} \int_{ M } \xi \, \right|
	\leq 2 \sqrt{ \pi } | \delta \xi |_{ H }
	\end{align}
	provided $ \delta\xi \in H $.
\end{lemma}

\begin{proof}[Proof of Lemma \ref{lemma:cameron-martin-continuity}.]
	First, we observe that either of the two noises statisfies
	\begin{align}\label{eqn:cm-pf-01}
	\int \delta \xi(x) f(x) dx \leq | \delta \xi |_{ H } \left( \int f^2 \right)^{ \frac{1}{2} }
	\end{align}
	for any Schwartz function $ f \in \mathscr{S}(\R^2) $. Indeed, to see this, we recall that the space of Schwartz functions $ \mathscr{S}(\R^2) $ embeds into its bi-dual $ ( \mathscr{S}' (\R^2) )' $ via the map $ f \mapsto ( \xi \mapsto \xi(f) ) $. Hence, by the definition of the norm $ | \cdot |_{ H } $, see \eqref{eqn:cm-pf-03}, we have
	\begin{align*}
	\int \delta \xi (x) f(x) dx
	&\leq \E[ \xi(f)^2 ] \sup \left\{ l(\delta \xi) ~\middle|~ \ell \in ( \mathscr{S}' (\R^2) )' , ~ \E[ \ell(\xi)^2 ] \leq 1 \right\} \\
	&= \E[ \xi(f)^2 ] | \delta\xi |_{ H }.
	\end{align*}
	In particular, \eqref{eqn:cm-pf-01} follows for the regularized white noise upon invoking the estimate
	\begin{align*}
	\E[ \xi(f)^2 ] = 2 \pi \int | \widehat{f} |^2 \widehat{c} 
	\leq 2 \pi ( \sup \widehat{c} ) \int | \widehat{f} |^2
	\leq \int f^2
	\end{align*}
	For the discretized white noise we observe that Jensen's inequality gives
	\begin{align*}
	\E[ \xi(f)^2 ] = \E \Big[ \sum_{ x,y \in \Z^2 } \xi_x \xi_y \Big( \int_{ Q_1 (x) } f \Big) \Big( \int_{ Q_{ 1 } (y) } f \Big) \Big] = \sum_{ x \in \Z^2 } \Big( \int_{ Q_{ 1 } (x) } f \Big)^2 \leq \int f^2 ,
	\end{align*}
	so that \eqref{eqn:cm-pf-01} is established in both cases.
	
	Via a density argument, we can apply \eqref{eqn:cm-pf-01} to $ f = \boldsymbol{1}_M $, which yields
	\begin{align*}
	\int_M \delta \xi(x) dx \leq | \delta \xi |_{ H }  | M |^{ \frac{1}{2} }.
	\end{align*}
	By the isoperimentric inequality this implies
	\begin{align*}
	\frac{1}{| \partial M |} \int_M \delta \xi(x) dx \leq 2 \sqrt{\pi} | \delta \xi |_{ H }.
	\end{align*}
	By linearity of the map $ \xi \mapsto \frac{1}{| \partial M |} \int_M \xi $, this already implies \eqref{eqn:cm-pf-02}.
\end{proof}

Equipped with the Lipschitz-continuity estimates in Lemma \ref{lemma:cameron-martin-continuity}, we appeal to the concentration estimate \cite[Theorem 4.5.7]{BogachevGM} asserting that
\begin{align} \label{eqn:Gaussian-concentration}
	\P \{ | X - \E X | > \alpha \} \leq 2 \exp ( - \frac{\alpha^2}{2 C_{ \rm CM } } ),
\end{align}
provided $ X = X(\xi) $ is Lipschitz-continuous under Cameron-Martin shifts with constant $ C_{ \rm CM } < \infty $, see also \cite[Theorem 4.5.6]{BogachevGM} for a one-sided tail estimate with improved constants. Using this, we can prove Proposition~\ref{prop:concentration-scales}, which is concerned with the pinning over dyadic scales. 
Note that in case of the discretized white noise, Step 1 is more explicit: the random variables $ \pm \frac{1}{ | \partial M | } \int_{ M } \xi $ are Gaussian with variance of order one, hence  \eqref{eqn-gaussian-concentration-plausibel-1} also applies directly to their supremum.

\begin{proof}[Proof of Proposition~\ref{prop:concentration-scales}.] We split the proof in two parts, one for \eqref{eqn:tail-bound-sup-x-r} and one for \eqref{eqn:expectation-bound-sup-x-r}.

\medskip

\textit{Step 1 (Concentration).} We start with the argument for \eqref{eqn:tail-bound-sup-x-r}. Let us observe that Lemma \ref{lemma:cameron-martin-continuity} states that
\begin{align*}
\xi \mapsto \frac{1}{ | \partial M | } \Big| \int_{ M } \xi \, \Big|
\quad { \rm is } \quad 2 \sqrt{ \pi } - { \rm Lipschitz ~ under ~ Cameron\text{-}Martin ~ shifts }.
\end{align*}
In particular, the same holds true after taking the two suprema under consideration (recall that the supremum of uniformly Lipschitz-continuous functions is Lipschitz-continuous with the same constant),~i.e.
\begin{align*}
\xi \mapsto \sup_{ R \ge 1 } ( \log R )^{ - \frac{3}{4} } S_R
\quad { \rm is } \quad 2 \sqrt{ \pi } - { \rm Lipschitz ~ under ~ Cameron\text{-}Martin ~ shifts }.
\end{align*}
The pre-factor $ ( \log R )^{ - \frac{3}{4} } $ is not harmful since it decays for large $ R $, even improving the continuity estimate for $ ( \log R )^{ - \frac{3}{4} } S_R $.
By the Gaussian concentration estimate \eqref{eqn:Gaussian-concentration} with $ C_{ \rm CM } = 2 \sqrt\pi $, we obtain \eqref{eqn:tail-bound-sup-x-r}. For later purposes let us record that same argument as above may be used to deduce
\begin{equation}\label{eqn:concentration-x-r}
\P \{ | S_R - \E S_R | \geq \alpha \}
\leq 2 \exp ( - \frac{1}{ 8 \pi } \alpha^2 )
\end{equation}
for every $ R > 0 $.

\medskip

\textit{Step 2 (Moment bound).} For the expectation we have
\begin{equation}\label{eqn:sup-x-r-expectation-reduction}
\E \sup_{ R \ge 2 } ( \log R )^{ - \frac{3}{4} } S_R
= \E \sup_{ k \in \N_0 } \sup_{ 2^k \leq R \leq 2^{k+1} } ( \log R )^{ - \frac{3}{4} } S_R
\sim \E \sup_{ R \ge 2 ~ { \rm dyadic } } ( \log R )^{ - \frac{3}{4} } S_R.
\end{equation}
Furthermore, we may estimate
\begin{equation}\label{eqn:sup-x-r-expectation-split}
\E \sup_{ R \ge 2 ~ { \rm dyadic } } ( \log R )^{ - \frac{3}{4} } S_R
\leq \sup_{ R \ge 2 ~ { \rm dyadic } } ( \log R )^{ - \frac{3}{4} } \E S_R
 + \E \sup_{  R \ge 2 ~ { \rm dyadic } } \log ( R )^{ - \frac{3}{4} }  ( S_R - \E S_R ).
\end{equation}
We estimate the second term on the r.h.s.~of \eqref{eqn:sup-x-r-expectation-split} by using \eqref{eqn:concentration-x-r}. More precisely, we have

\begin{align*}
\E \sup_{ R \ge 2 ~ { \rm dyadic } }  ( \log R )^{ - \frac{3}{4} }  ( S_R - \E S_R )
& \leq 1 + \sum_{ R \ge 2 ~ { \rm dyadic } } \int_{ 1 }^{ \infty } \P \{ S_R - \E S_R > t ( \log R )^{ \frac{3}{4} } \} dt \\
& \lesssim 1 + \sum_{ R \ge 2 ~ { \rm dyadic } } \int_{ 1 }^{ \infty } \exp( - \frac{1}{8 \pi} t^2 ( \log R )^{ \frac{3}{2} } ) dt \\
& \lesssim 1 + \sum_{ R \ge 2 ~ { \rm dyadic } } ( \log R )^{ - \frac{3}{4} }  \int_{  ( \log R )^{ \frac{3}{4} } }^{ \infty } \exp( - \frac{1}{8 \pi} t^2 ) dt.
\end{align*}
Using the tail bound
\begin{align}\label{eqn:gaussian-tail-bound}
\int_{ x }^{ \infty } \exp( - \frac{1}{8 \pi} t^2 ) dt
\lesssim \frac{1}{ x } \exp( - \frac{1}{8 \pi} x^2 ),
\end{align}
see e.g.~\cite[(2.1.1)]{AdlerTaylor-RF-and-Geometry}, we obtain
\begin{equation}\label{eqn:sup-x-r-expectation-split-ii}
\E \sup_{ R \ge 2 ~ { \rm dyadic } } ( \log R )^{ - \frac{3}{4} } ( S_R - \E S_R )
\lesssim 1 + \sum_{ R \ge 2 ~ { \rm dyadic } } ( \log R )^{ - \frac{3}{2} } \exp( - \frac{1}{8 \pi}  ( \log R )^{ \frac{3}{2} } )
\end{equation}
and the last sum converges. Hence, \eqref{eqn:sup-x-r-expectation-reduction} and \eqref{eqn:sup-x-r-expectation-split} imply
\begin{align*}
\sup_{ R \ge 2 ~ { \rm dyadic } } ( \log R )^{ - \frac{3}{4} } \E S_R
\lesssim \E \sup_{ R > 0 } ( \log R )^{ - \frac{3}{4} } S_R
\lesssim 1 +  \sup_{ R \ge 2 ~ { \rm dyadic } } ( \log R )^{ - \frac{3}{4} } \E S_R
\end{align*}
and the latter supremum is finite due to Theorem \ref{thm:talagrand-bound}. This finishes the argument for \eqref{eqn:expectation-bound-sup-x-r}.
\end{proof}

Equipped with Proposition~\ref{prop:concentration-scales}, we now use another concentration argument to perform the pinning in space. The proof proceeds analogously to the previous proof, which makes it a bit hard to recognize the heuristic in Section~\ref{sec:strategy-of-proof}. In its core, our argument literally applies to an abstract set of Gaussian random variables; see also the proof of Lemma \ref{lem:wn-sup-bound} below.

\begin{proof}[Proof of Proposition \ref{prop:concentration-in-space}.]
Let us write
\begin{align*}
\mathscr{S} (x) =  \sup_{ R \ge 2 } \, ( \log R )^{ - \frac{3}{4} } \, S_R ( \xi( \cdot - x ) )
\quad { \rm for } ~ x \in \R^2.
\end{align*}
The reader will notice that the argument relies on the sub-Gaussianity of $ \mathscr{S} (x) $ proven in Proposition \ref{prop:concentration-scales} and hence follows essentially the same strategy as the proof of the aforementioned Proposition.

\medskip

\textit{Step 1 (Concentration).} As in the proof of Proposition \ref{prop:concentration-scales}, we learn from \eqref{eqn:cm-pf-02} that
\begin{align*}
\xi \mapsto \sup_{ x \in \R^2 } \, ( \log | x |_{+})^{ - \frac{1}{2} } \, \mathscr{S} (x)
\quad { \rm is } \quad 2 \sqrt{ \pi } - { \rm Lipschitz ~ under ~ Cameron\text{-}Martin ~ shifts },
\end{align*}
so that by \cite[Theorem 4.5.7]{BogachevGM} it follows that the concentration property \eqref{eqn:concentration-in-space-tail} holds true.

\medskip

\textit{Step 2 (Estimate for the expectation).} We start with the following observation. Given some set $ M \subseteq B_R(x) $, we have $ M \subseteq B_{ R + 1 } ( y ) $ whenever $ | x - y | \leq 1 $. This implies
\begin{align*}
\sup_{ x \in \R^2 } \, ( \log | x | _{+} )^{ - \frac{1}{2} } \mathscr{S} (x)
\lesssim \sup_{ y \in \Z^2 } ( \log | y |_{+} )^{ - \frac{1}{2} } \mathscr{S} (y)
\lesssim \sup_{ y \in \Z^2, | y | \ge 2 } ( \log | y |_{+} )^{ - \frac{1}{2} } \mathscr{S} (y).
\end{align*}
Hence, by \eqref{eqn:expectation-bound-sup-x-r}, we obtain
\begin{align*}
&\E \sup_{ x \in \R^2 } \, ( \log | x |_{+} )^{ - \frac{1}{2} } \mathscr{S} (x) \\
&\qquad \lesssim \sup_{ x \in \Z^2, | x | \ge 2 } \, ( \log | x | )^{ - \frac{1}{2} } \E \mathscr{S} (x) + \E \sup_{ x \in \Z^2, | x | \ge 2 } ( \log | x | )^{ - \frac{1}{2} } ( \mathscr{S} (x) - \E \mathscr{S} (x) ) \\
&\qquad \lesssim 1 + \E \sup_{ x \in \Z^2, | x | \ge 2 } \, ( \log | x | )^{ - \frac{1}{2} } ( \mathscr{S} (x) - \E \mathscr{S} (x) ).
\end{align*}
Similar to \eqref{eqn:sup-x-r-expectation-split-ii}, we may show that for an arbitrary cut-off $ t > 0 $ (that was chosen to be $ 1 $ before), we have
\begin{align*}
&\E \sup_{ x \in \Z^2, | x | \ge 2 } ( \log | x | )^{ - \frac{1}{2} } ( \mathscr{S} (x) -  \E \mathscr{S} (x) ) \\
&\qquad \leq t + \sum_{ x \in \Z^2, | x | \ge 2 }  ( \log | x | )^{ -1 } \exp( - \frac{ t }{8 \pi} \log | x | ).
\end{align*}
The exact form of the exponents is not so important here, since we have the freedom to choose $ t $: note that for $ | x | \ge 2 $ the latter summand simplifies to $ | x | ^{ - \frac{ t }{8 \pi} } ( \log | x | )^{ - 1 } $, which is summable for large $ t \gg 1 $. Together, the last two estimates imply \eqref{eqn:concentration-in-space-expectation}.
\end{proof}

\subsection{Chaining arguments}\label{sec:pf-talagrand-bound}

The goal of this section is to prove Proposition~\ref{thm:talagrand-bound}. Let us stress again that this is the continuum version of \cite[Proposition 2.2]{DW}, which carries over to our continuum setting due to the UV regularization. In the two Lemmas~\ref{lem:wn-sup-bound} and \ref{lem:deformation} below we collect the main adaptations to the continuum setting. 

\medskip

\paragraph{\textit{Preliminaries on the Chaining Method.}}
In this paragraph we give an overview of the main ingredient in Proposition \ref{thm:talagrand-bound}: the \emph{chaining method}. 
We start with an exposition of the chaining method, where we recall some useful results from \cite{Talagrand-GC} that we will need later on. 
For more details on this technique, we refer the reader to \cite{VanHandel} and the books \cite{Talagrand-GC} and \cite{Talagrand-ULB}.

The chaining method is based on the insight that the supremum of centered Gaussian random variables $ \{ X_t ~|~ t \in T \} $ is related to covering numbers with respect to the metic 
\begin{align*}
	d(s,t)^2  \coloneqq \E | X_s - X_t |^2 ,
\end{align*}
and provides a far reaching generalization of Dudley's bound \cite{DudleyBound} or \cite[(1.18)]{Talagrand-GC}. 
Due to Gaussianity of the random variables, this metric provides a strong control on their distance, which is encoded in the tail estimate
\begin{align*}
	\P \{ | X_s - X_t | > \alpha \} \leq \exp( - \frac{\alpha^2}{2 d(s,t)^2} ) . 
\end{align*}
In a kind of coarse graining procedure, one bounds the supremum of the family $ \{ X_t ~|~ t \in T \} $ by grouping the random variables provided they are close in the distance $d$. 
In the Dudley  bound this is translated in a successive covering of the metric space $ T $ by balls. The main insight of Talagrand is to consider general partitions instead; optimizing them leads to sharp estimates, see \cite[Section 1.2]{Talagrand-GC} for a short exposition of that topic. Using that strategy one can prove that
\begin{align}\label{eqn:chaining-bound}
\E \sup_{ t \in T } X_t \sim \inf \Big\{ \sup_{ t \in T } \sum_{ k \geq 0 } 2^{ \frac{k}{2} } { \rm diam } (  A_k(t) ) ~\big|~ { \rm admissible ~ partition ~ chain }\footnotemark ~ \{ \mathscr{A}_k \}  \Big\}, 
\end{align}
see \cite[Theorem 1.2.6]{Talagrand-GC} and \cite[Theorem 2.10.1]{Talagrand-ULB}.\footnotetext{A sequence of partitions $ \{ \mathscr{A}_k \} $ of the metric space $ (T,d) $ is admissible provided it increases, but its size growths like $ | \mathscr{A}_k | \leq 2^{ 2^k } $. For any such partition $ \mathscr{A}_k $ the set $ A_k (t) $ denotes the unique element $ A \in \mathscr{A}_k $ such that $ t \in A $.} 
Note that the r.h.s.~of \eqref{eqn:chaining-bound} only depends on the geometry of the metric space $ (T,d) $; due to its crucial role in the chaining method it is usually denoted by $ \gamma_2 (T,d) $. Furthermore, if the process $ \{ X_t \, \mid \, t \in T \} $ is symmetric, meaning that $  \{ X_t \, \mid \, t \in T \} $ and $ \{ - X_t \, \mid \, t \in T \} $ have the same law, then \eqref{eqn:chaining-bound} may be postprocessed into an estimate on the supremum of $ | X_t | $, namely
\begin{align}\label{eqn:chaining-bound-b}
0 \leq \E \sup_{ t \in T } | X_t | - \E | X_{ t_0 } |  \lesssim \gamma_2 ( T, d ) 
\end{align}
for any $ t_0 \in T $, see \cite[Lemma 2.2.1]{Talagrand-ULB}. 

Since in our case we are only interested in upper bounds on the supremum of the Gaussians on the l.h.s.~of \eqref{eqn:chaining-bound}, let us make the useful observation that instead of studying the metric space $ (T,d) $, we can parametrize it through a surjective Lipschitz-continuous map $ f : (M,d_M) \rightarrow (T,d) $ provided we understand the metric geometry of $ (M,d_M) $ in the sense of the $ \gamma_2 $-functional. More precisely, in this case we have
\begin{align}\label{eqn:gamma-lipschitz-composition}
	\E \sup_{ t \in T } X_t \lesssim [ f ]_{ \rm Lip } \, \gamma_{2} ( M, d_M ) , 
\end{align}
see \cite[Theorem 1.3.6]{Talagrand-GC}. Hence, to prove Proposition \ref{thm:talagrand-bound}, we need to find a convenient way of parametrizing the Gaussians $ \frac{1}{ | \partial M | } \int_{M} \xi $ through a metric space that we understand well.

In \cite{DW}, a proof of the discrete analog of Proposition \ref{thm:talagrand-bound} is suggested that employs the chaining bound in a multi-scale way, in the sense that they study the integrals $ \frac{1}{ | \partial M | } \int_{M} \xi $ for $ | \partial M | \sim L $ over a range of dyadic scales $ L $. In this way, the metric takes the more customary form
\begin{align*}
\E \Big| \frac{1}{L} \int_{ M } \xi - \frac{1}{L} \int_{ N } \xi \Big|^2
= \left\{ \begin{aligned}
& ~ \frac{2\pi}{L^2} \int | \widehat{ \boldsymbol{1}_M } (k) -  \widehat{ \boldsymbol{1}_N  }(k) |^2 \widehat{c}(dk) ~ \\
& ~ \frac{1}{L^2} \sum_{ x \in \Z^2 } \Big( \int_{ Q_1(x) } (  \boldsymbol{1}_M -  \boldsymbol{1}_N ) \Big)^2 ~
\end{aligned} \right\}
\leq \frac{1}{L^2} \int |  \boldsymbol{1}_M -  \boldsymbol{1}_N |^2
\end{align*}
for the two noises under consideration in Section \ref{sec:assumptions-on-noise}. Let us note that passing to the upper bound is justified by \eqref{eqn:gamma-lipschitz-composition} above. In this way, we are led to consider the space $ \{ M \subseteq B_R ~|~ \int | \nabla \boldsymbol{1}_M | \leq L \} $ of Caccioppli sets w.r.t.~the $ L^2 $-metric. Note that the former set is a compact subset of the space $ BV( B_R ) $ equipped with the $ L^1 $-topology (which in terms of topologies agrees with the $ L^2 $-topology induced on that set).

We can use \eqref{eqn:gamma-lipschitz-composition} to pull the estimation procedure back into a Hilbert space setting, where the Euclidian geometry allows for explicit computations. In our case, the following parametrization will be useful.  

\begin{lemma}[Proposition 3.4.5 in \cite{Talagrand-GC}]\label{lem:paramterization-through-h1}
	Denote by $ L = 2^k $ (for some $ k \in \N $) a length scale and let
	\begin{align*}
	\Gamma_{ L } \coloneqq \big\{ M \subseteq \R^2 ~ { \rm simply ~ connected } ~ \big| ~ \partial M ~ { \rm contained ~ in ~ the ~ square ~ lattice }, ~ | \partial M | \leq L, ~ 0 \in \partial M \big\}
	\end{align*}
	There exists a surjective map $ M : { \rm dom } ( T ) \subseteq \{ f \in H^1 ( \mathbb{T} ) ~|~ \int_{ \mathbb{T} } | f |^2 + \int_{ \mathbb{T} } | f' |^2 \leq 1 \} \rightarrow \Gamma_L $	 such that
	\begin{align*}
		\int | \boldsymbol{1}_{ M( f ) } - \boldsymbol{1}_{ M ( g ) } | \lesssim L^2 \Big( \int_{ \mathbb{T} } | f - g |^2 \Big)^{ \frac{1}{2} } .
	\end{align*}
\end{lemma}

For a proof of the lemma, we refer the reader to the proof given in Talagrand's book, see \cite[Proof of Proposition 3.4.5 \& Lemma 3.4.7 for the Fourier space representation]{Talagrand-GC}. The idea originates from the matching community, see again \cite[Theorem 3.4.1]{Talagrand-GC} for an exposition of this topic. Since the elements of $ \Gamma_L $ are uniquely determined by their boundary curves, one may use a piecewise continuous function to parametrize it. The crux is to prove the above estimate.

Note that since our original metric measures distances on indicator functions, this bound translates to the statement
\begin{align}\label{eqn:upper-bound-gaussian-metric}
\Big( \E \Big| \frac{1}{L} \int_{ M(f) } \xi - \frac{1}{L} \int_{ M(g) } \xi \Big|^2 \Big)^{ \frac{1}{2} }
\lesssim \Big( \frac{1}{L} \int | \boldsymbol{1}_{ M( f ) } - \boldsymbol{1}_{ M ( g ) }  | \Big)^{ \frac{1}{2} } \lesssim \Big( \int_{ \mathbb{T} } | f - g |^2 \Big)^{ \frac{1}{4} } 
\end{align}
for the $ L^2 $-metric considered above. Note that the exponent on the r.h.s.~is $ \frac{1}{4} $, which is from a metric point of view not a problem: the square root of a metric is still a metric.

It is the next lemma from \cite{Talagrand-GC} that recovers the correct exponent in the framework of the chaining method. The lemma relates the $ \gamma_2 $-functional of interest to another $ \gamma $-functional, see \cite[(3.4)]{Talagrand-GC}, namely
\begin{align}\label{eqn:gamma_1-2}
	\gamma_{1,2} (T,d) \coloneqq \inf \Big\{ \sup_{ t \in T } \Big( \sum_{ k \geq 0 } \big( 2^{k} { \rm diam } (  A_k(t) ) \big)^2 \Big)^{ \frac{1}{2} } ~\big|~ { \rm admissible ~ partition ~ chain } ~ \{ \mathscr{A}_k \}  \Big\}.
\end{align}
At a first glance, this looks a bit artificial, but it turns out that this is exactly the right notion of measuring the compactness of the embedding $ H^1 ( \mathbb{T} ) \hookrightarrow L^2 ( \mathbb{T} ) $, more precisely we shall consider $\gamma_{1,2} ( H^1 ( \mathbb{T} ) \hookrightarrow L^2 ( \mathbb{T} )  ) \coloneq \gamma_{1,2} ( \{  \| f \|_{ H^1(\T) } \leq 1 \} , \| \cdot \|_{ L^2(\T) } ) $; in view of our original metric, corresponding to the square-root of the $ L^2 ( \mathbb{T} ) $-norm, the new functional replaces the $ \ell^{ \frac{1}{2} } $-summability (which by the naive argument above we expect to diverge) on the r.h.s.~of \eqref{eqn:chaining-bound} by an $ \ell^2 $-summability condition in the $ \gamma_{1,2} $-functional (which borderline converges). In Lemma \ref{chain-key-lemma} we will refine the aforementioned strategy and explicitly construct a partition of $ H^1( \mathbb{T} ) $ as a competitor in \eqref{eqn:gamma_1-2}. Since we will use $ L^2 $-balls of slightly varying radii this may indeed be seen as a refined statement on the compactness of $ H^1( \mathbb{T} ) \hookrightarrow L^2 ( \mathbb{T} ) $.

The first of the above mentioned lemma is the following.

\begin{lemma}[Lemma 3.4.6 in \cite{Talagrand-GC}]\label{lem:emergernce-gamma_1-2}
	Let $ (T, d_T) $, $ (M, d_M) $ be two metric spaces related through a Lipschitz-continuous surjective map $ f : (M, d_M^{\frac{1}{2}}) \rightarrow (T,d_T) $. Assume that $ T $ is finite (with $ | T | \geq 2 $ to exclude pathological cases); then
	\begin{align*}
		\gamma_2 ( T, d_T ) \lesssim ( \log \log | T | )^{ \frac{3}{4} } \gamma_{ 1, 2 } ( M, d_M ) ^{ \frac{1}{2} } .
	\end{align*}
\end{lemma}

The following proof is essentially the same as in \cite{Talagrand-GC}, we include it for the reader's convenience, since we have changed the formulation of the statement a little lit.

\begin{proof}
	W.l.o.g.~we may assume that $ M = f^{-1} ( T ) $. Otherwise, we restrict $ M $ and observe after the proof that $ \gamma_{1,2} $ is monotone up to multiplicative constants, see \cite[Theorem 1.3.6]{Talagrand-GC}.

	Modulo the Lipschitz parameterization that may be handled as in  \cite[Theorem 1.3.6]{Talagrand-GC}, see also the comment above, what we really need to show is
	\begin{align*}
		\gamma_2 ( M, d_M^{\frac{1}{2}} ) \lesssim ( \log \log | M | )^{ \frac{3}{4} } \gamma_{ 1, 2 } ( M, d_M ) ^{ \frac{1}{2} } .
	\end{align*}
	Since $ | M | = | T | \geq 2 $, we can write $ 2^{ 2^{n-1} } \leq | M | \leq 2^{ 2^n } $ for some $ n \in \N $. Note that in particular,
	\begin{align*}
	\log ( \log 2 ) + ( n-1 ) \log 2 
	\leq \log ( \log | M | )
	\leq \log ( \log 2 ) + n \log 2,
	\quad { \rm so ~ that } \quad
	\log \log | M | \sim n .
	\end{align*}
	For the desired estimate, we start with an admissible partition $ \{ \mathscr{A}_k \} $ that almost attains the infimum for $ \gamma_{1,2}(T,d) $, i.e.
	\begin{align*}
		 \gamma_{1,2} (M, d_M) 
		 \leq \sup_{t \in T} \Big( \sum_{ k \geq 0 } \big( 2^k { \rm diam }_{ d } ( A_k (t) ) \big)^2 \Big)^{ \frac{1}{2} }
		 \leq \gamma_{1,2} (M, d_M) + \varepsilon
	\end{align*}
	Note that the remaining supremum above becomes smaller -- and hence closer to the actual infimum -- if we redefine $ \mathscr{A}_n = \{ \{ t \} ~|~ t \in M \} $ (note that $ | \mathscr{A}_n | \leq | M | \leq 2^{ 2^n } $). This procedure ensures that the sum is finite; it stops at $ k = n - 1 $. Hence, by Hölder's inequality we have
	\begin{align*}
		 \gamma_{2}( M, d_M^{ \frac{1}{2} } )
		 \leq \sum_{ k \geq 0 } \big( 2^k { \rm diam }_{ d } ( A_k(t) ) \big)^{ \frac{1}{2} } 
		 \leq n^{ \frac{3}{4} } \Big( \sum_{ k \geq 0 } \big( 2^k { \rm diam }_{ d } ( A_k(t) ) \big)^{2} \Big)^{ \frac{1}{4} }
		 \leq n^{ \frac{3}{4} } ( \gamma_{1,2} ( M, d_M) + \varepsilon )^{ \frac{1}{2} } .
	\end{align*}
	Using $ | M | \sim n $ this concludes the proof upon sending $ \varepsilon \rightarrow 0 $.
\end{proof}

To make use of the last lemma, we of course need control over the quantity $ \gamma_{1,2} ( H^1( \mathbb{T} ) \hookrightarrow L^2( \mathbb{T} ) ) $ appearing therein. This is handled by the next lemma.

\begin{lemma}[Lemma 3.4.7 in \cite{Talagrand-GC}]\label{chain-key-lemma}
	We have $ \gamma_{1,2} ( H^1( \mathbb{T} ) \hookrightarrow L^2( \mathbb{T} ) ) < \infty $.
\end{lemma}

Again, we refer to the book \cite{Talagrand-GC} for a detailed proof utilizing the partition schemes in \cite[Corollary 3.1.4 \& Theorem 3.1.5]{Talagrand-GC}. Another more recent proof is given in \cite[Theorem 4.2.1]{Talagrand-ULB}, which simplifies the argument by using the contraction principle developed in \cite[Theorem 3.1]{VanHandel}; specializing this construction to our setting leads to an explicit construction of the admissible partition in \eqref{eqn:gamma_1-2}.

\medskip

\paragraph{\textit{Adaptations to the continuum setting.}}

In our proof of Proposition \ref{thm:talagrand-bound}, we follow the general strategy by \cite{DW} and \cite{Talagrand-GC}. Due to our continuum setting, some adaptations of the proof are needed. Briefly summarized, we need some type of continuity of the functional $ M \mapsto \frac{1}{ | \partial M | } \int_{ M } \xi $ to reduce the maximization in \eqref{thm:talagrand-bound} to a set parameterized by discretely many curves. The exponent $ \frac{3}{4} $ then emerges from Lemma~\ref{lem:emergernce-gamma_1-2}. 

In view of the bound of order $ ( \log R )^{ \frac{3}{4} } $ in \eqref{thm:talagrand-bound}, the following estimate on the supremum of $ \xi $ suffices for our purposes. It is this lemma where we need the assumption on the noise $ \xi $ in Section \ref{sec:assumptions-on-noise}. We will see that, next to the Gaussianity, we rely on the first property in \eqref{eqn:noise-assumption}, which we verify along the lines.

\begin{lemma}\label{lem:wn-sup-bound}
Let $ \xi $ denote either of the two noises in Section \ref{sec:assumptions-on-noise}. It holds
\begin{equation}\label{eqn:wn-sup-bound}
\E \sup_{Q_R} | \xi | \lesssim ( \log R )^{ \frac{1}{2} }
\end{equation}
for every $ R \ge 2 $.
\end{lemma}

Before we prove the lemma, let us comment on why it is natural to expect the bound \eqref{eqn:wn-sup-bound}: we think of the assumptions of Section~\ref{sec:assumptions-on-noise} as guaranteeing that the noise $ \xi $ is sufficiently regular on small scales. Therefore, the supremum of $ | \xi | $ on small balls behaves, at least in terms of estimates, like $ | \xi | $ at the center point. In view of the white noise character of $ \xi $ on large scales this shows that the supremum of $ | \xi | $ behaves like the supremum of $ R^d $ independent random variables, for which the logarithmic bound in \eqref{eqn:wn-sup-bound} is sharp. To make this rigorous, surprisingly few regularity assumptions on the noise are needed.

\begin{proof}[Proof of Lemma \ref{lem:wn-sup-bound}.]
We split the proof in two steps.

\medskip

\textit{Step 1 (Core argument).} Note that 
\begin{align*}
\sup_{Q_R} | \xi | \leq \sup_{p \in \Z^2 \cap Q_R} \sup_{ Q_1(p) } | \xi | .
\end{align*}
Below we will argue that
\begin{align}\label{eq:sup-xi-small-ball}
\E \sup_{ Q_1( p ) }  | \xi | \lesssim 1
\quad \text{and} \quad
\P \{  \sup_{ Q_1 ( p ) }  | \xi | \geq \E \sup_{ Q_1( p ) } | \xi  | + \alpha \} \lesssim \exp \left( - \frac{\alpha^2}{2 \sigma^2} \right)
\end{align}
for some positive constant $ \sigma = \sigma(d) > 0 $. The claim now follows by standard arguments.

Indeed, applying the Gaussian tail estimate \eqref{eqn:gaussian-tail-bound} yields
\begin{align*}
\E \sup_{ p \in \Z^2 \cap Q_R } ( \sup_{ Q_1(p) } | \xi | - \E  \sup_{ Q_1(p) } | \xi | )
& \leq \int_0^{\infty} \P \{ \sup_{ Q_1(p) } | \xi | - \E  \sup_{ Q_1(p) } | \xi | > t ~ \text{for~some} ~ p \in \Z^2 \cap Q_R \} dt \\
& \le \delta + \sum_{ p \in \Z^2 \cap Q_R } \int_{\delta}^{\infty} \P \{  \sup_{ Q_1(p) } | \xi | > t + \E  \sup_{ Q_1(p) } | \xi | \} dt \\
& \lesssim \delta + \sigma R^d \int_{\frac{\delta}{\sigma}}^{\infty} \exp \left(- \frac{t^2}{2} \right) dt
\le \delta + R^d \frac{ \sigma^2 }{ \delta } \exp \left( - \frac{1}{2} \frac{\delta^2}{\sigma^2} \right)
\end{align*}
for all $ \delta > 0 $. In particular, we may choose $ \delta = \sigma \sqrt{2 \log R^d} \sim \sqrt{ \log R } $ to obtain
\begin{align*}
\E \sup_{p \in \Z^2 \cap Q_R } \sup_{ Q_1 ( p ) } | \xi |  
\lesssim \sup_{p \in \Z^2 \cap Q_R } \big( \E \sup_{ Q_1 ( p ) } | \xi |  \big) + ( \log R )^{ \frac{1}{2} } +  ( \log R )^{ - \frac{1}{2} }
\lesssim ( \log R )^{ \frac{1}{2} }, 
\end{align*}
where the last inequality holds if $ R $ is strictly bounded away from $ 1 $.

\medskip

\textit{Step 2 (Argument for \eqref{eq:sup-xi-small-ball}).} By $ \Z^2 $-stationarity we restrict ourselves to the case $ p = 0 $. The estimates in \eqref{eq:sup-xi-small-ball} are a direct consequence of the Borell-TIS-inequality, see e.g.~\cite[Theorem 2.1.1]{AdlerTaylor-RF-and-Geometry}, which applies since any of the two noises satifies
\begin{align}\label{eqn-borell-tis-assumption}
	\sup_{ Q_1 } | \xi | < \infty \quad { \rm almost ~ surely },
\end{align}
see argument below. Therefore, the Borell-TIS-inequality implies
\begin{align*}
\E \sup_{ Q_1 } \xi < \infty
\quad { \rm and } \quad
\P \big\{ \sup_{ Q_1 } \xi - \E \sup_{ Q_1 } \xi > t \big\}
\le \exp \left( - \frac{t^2}{2 \sigma^2} \right)
\end{align*}
with constant $ \sigma^2 = \sup_{x \in Q_1} \E \xi(x)^2 < \infty $. Since $ \xi $ is a centered Gaussian, we have $ \xi =_{ \rm law } - \xi $ so that by $ \sup_{ Q_1 } | \xi | = \sup_{ Q_1 } \max \{ \xi , - \xi \} $ we have
\begin{align*}
\P \{  \sup_{ Q_1 } | \xi | > t  \} \leq 2 \, \P \{ \sup_{ Q_1 }  \xi > t  \} ,
\end{align*}
which in particular implies
\begin{align*}
\P \{  \sup_{ Q_1 } | \xi | > \E \sup_{ Q_1 } | \xi | + t \}
\leq 2 \, \P \{  \sup_{ Q_1 }  \xi > \E \sup_{ Q_1 } | \xi | + t \}
\leq 2 \, \P \{ \sup_{ Q_1 }  \xi > \E \sup_{ Q_1 } \xi + t \} .
\end{align*}
This shows that although in disguise, the above estimate already contains \eqref{eq:sup-xi-small-ball}.

\medskip

\textit{Step 3 (Argument for \eqref{eqn-borell-tis-assumption}).} To conclude the lemma, it is left to establish \eqref{eqn-borell-tis-assumption} for our two noises. Note that this also implies the first item in \eqref{eqn:noise-assumption}. Furthermore, let us note that the claim is almost immediate for the discretized white noise. Therefore, we focus only on white noise in the presence of the UV cut-off.

By $ \Z^2 $-stationary, we may restrict ourselves to the case $ p = 0 $. We first argue that
\begin{align*}
\E | \xi(x) - \xi(y) |^2
= 2 (c(0) - c(x-y))
\leq C^2 | x - y |^{2 \alpha} ,
\end{align*}
for some constant $ C > 0 $ that will appear explicitly further down in the proof. 
Indeed, to obtain the last inequality we used the UV regularization in form of
\begin{align*}
| c(0) - c(x-y) |
\sim \int | e^{0} - \exp(i (x-y) \cdot q ) | \hat{c}(q) dq
\lesssim |x - y|^{2 \alpha} \int | q |^{2 \alpha} \hat{c}(q) dq ,
\end{align*}
and $ C^2 $ is obtained multiplying the implicit constants and $  \int | q |^{2 \alpha} \hat{c}(q) dq $.

Next, Dudley's bound, see \cite[(1.18)]{Talagrand-GC} or \cite{DudleyBound}, can be applied using the above continuity statement, i.e.
\begin{align}\label{eqn:lem-wn-sup-bound-i}
\E \sup_{B_1} | \xi |
\lesssim \int_0^{ \infty } \left( \log N(B_1, C | \cdot |^{\alpha}; \varepsilon) \right)^{\frac{1}{2}} d \varepsilon
\end{align}
It is well-known that in two space-dimensions
\begin{align*}
\left( \frac{1}{ \varepsilon } \right)^2 \leq N( B_1, | \cdot |; \varepsilon )  \leq \left( \frac{ 2 + \varepsilon }{ \varepsilon } \right)^2 ,
\end{align*}
see~\cite[Exercise 2.5.9 \& (2.47)]{Talagrand-ULB}. Therefore, the last integrand is given by
\begin{align*}
\log N(B_1, C | \cdot |^{\alpha} ; \varepsilon )
= \log N(B_1, | \cdot |;  ( \frac{ \varepsilon }{ C } )^{\frac{1}{\alpha}})
\lesssim \left\{ \begin{aligned}
 \log \frac{2 + ( \frac{ \varepsilon }{ C } )^{ \frac{1}{\alpha} } }{ ( \frac{ \varepsilon }{ C } )^{\frac{1}{\alpha}} } \quad & { \rm for } ~ \varepsilon < C \\
0 \quad & { \rm for } ~ \varepsilon \ge C,
\end{aligned}\right.
\end{align*}
so that we may estimate \eqref{eqn:lem-wn-sup-bound-i} by
\begin{align*}
\E \sup_{B_1} | \xi |
\lesssim \int_0^1 \left( \log ( 1 +  2 \varepsilon^{ - \frac{1}{\alpha} } ) \right)^{ \frac{1}{2} } d\varepsilon
\lesssim_{ \alpha } \int_1^{\infty} \left( \log ( 1 + 2 x ) \right)^{ \frac{1}{2} } \frac{ dx} { x^{ 1 + \alpha } }
\lesssim_{ \alpha } 1,
\end{align*}
which establishes \eqref{eqn-borell-tis-assumption} after covering $ Q_1 $ with finitely many balls.
\end{proof}

Equipped with Lemma~\ref{lem:wn-sup-bound}, we are ready to prove that in the supremum considered in Proposition~\ref{thm:talagrand-bound} we can restrict to finitely many curves. To this end, let us recall from Lemma~\ref{lem:paramterization-through-h1} the set of curves
\begin{align*}
\Gamma_{ L } \coloneqq \big\{ M \subseteq \R^2 ~ { \rm simply ~ connected } ~ \big| ~ \partial M ~ { \rm contained ~ in ~ the ~ square ~ lattice }, ~ | \partial M | \leq L, ~ 0 \in \partial M \big\}
\end{align*}
of which we have a good parametrization. The key idea of the next lemma is to deform admissible curves to curves contained in $ \Gamma_L $. In view of the supremum bound in Lemma~\ref{lem:wn-sup-bound} and our claim in Proposition~\ref{thm:talagrand-bound}, we can allow for a deformation of order $ L ( \log L )^{ \frac{1}{4} } $ around the boundary (which is much more space than we actually need).

\begin{lemma}[Deformation to discrete curves]\label{lem:deformation}
Given some length scale $ 1 \leq L < \infty $, we consider a set $ \Gamma $ of simply connected sets (in $ \R^2 $) with smooth boundaries; additionally assume that $ | \partial M | \leq L $ and $ 0 \in B_1 (\partial M) $ for $ M \in \Gamma $. There exists a universal constant $ 0 < c_{ \rm def } < \infty $, such that for every $ M \in \Gamma $, there exists an $ N \in \Gamma_{ c_{ \rm def } L } $ such that
	\begin{align*}
	\int | \boldsymbol{1}_{M} - \boldsymbol{1}_{ N } | \lesssim L .
	\end{align*}
\end{lemma}

\begin{proof}[Proof of Lemma \ref{lem:deformation}.]
Given $ M \in \Gamma $, let us write
\begin{align*}
B_1 ( \partial M ) \coloneq \{ p \in \R^2 ~|~ { \rm dist } ( p, \partial M ) \leq 1 \}.
\end{align*}
for the tube of radius one around $ \partial M $. Now let $ \gamma $ the innermost curve fully contained in the square lattice that contains $ B_1( \partial M ) $. We denote by $ { \rm int } ( \gamma ) $ the domain enclosed by this curve. Note that this curve is contained in $ B_3 ( \partial M ) $, which contains at most order $ | \partial M | $ many points from $ \Z^2 $. This follows by a volume argument. Hence, $ { \rm Length } ( \gamma ) \lesssim L $. Furthermore, by construction we have
\begin{align*}
| \boldsymbol{1}_{ { \rm int } ( \gamma ) } - \boldsymbol{1}_{ M } | \lesssim | B_3 ( \partial M ) | \lesssim L.
\end{align*}
By redirecting the curve $ \gamma $ over finitely many lattice sides, we may ensure that $ 0 \in \gamma $. Since $ L \geq 1 $ the above bound and the estimate $ { \rm Length } ( \gamma ) \lesssim L $ remain true upon chaining the implicit constant.
\end{proof}

\paragraph{\textit{Proof of Proposition \ref{thm:talagrand-bound}.}}

We are finally ready to prove Proposition \ref{thm:talagrand-bound}. We closely follow the argument given by \cite[Theorem 1.6 \& Proposition 2.2]{DW} in the discrete situation of the RFIM. 

Note that the concentration property \eqref{eqn:talagrand-concentration} already appeared in \eqref{eqn:concentration-x-r} in the section collecting the concentration type argument. Hence, our main concern for now is the upper bound \eqref{eqn:talagrand-expectation}. We proceed in several steps.

\medskip

\textit{Step 1 (Reduction to simple closed boundary curves).} As in \cite[Proof of Theorem 1.6]{DW}, we now reduce the maximization problem in \eqref{eqn:talagrand-expectation} to simply connected sets with smooth boundarys. Additionally we need to cut-off too short or too boundary curves, which happens automatically in the discrete setting. More precisely, we argue that
\begin{equation}\label{eqn:talagrand-expectation-reduction}
\begin{aligned}
& \E \sup \left\{ \Big| \frac{1}{ | \partial M | } \int_M \xi \Big| ~\middle|~ M \subseteq B_R ~ { \rm Caccioppoli ~ set}  \right\} \\
& \qquad \lesssim \E \sup \left\{ \Big| \frac{1}{ | \partial M | } \int_M \xi \Big| ~\middle|~ \begin{matrix} M \subseteq B_R ~ { \rm simply ~ connected}, \\ \partial M ~ { \rm smooth }, 1 \leq | \partial M | \leq R^2 \end{matrix} \right\} + ( \log R )^{ \frac{1}{2} }.
\end{aligned}
\end{equation}
In view of our claim the term $ ( \log R )^{ \frac{1}{2} } $ on the right is of lower order. Note that this term is not present in the discrete setting of \cite{DW} and comes from the reduction from (possibly very rough) Caccioppoli sets to smooth sets, and from the restriction to short boundaries that is automatic on the lattice. 

First, we want to argue that we may restrict to sets  $ M $ with smooth boundary. To this end, let us start with the continuity estimate
\begin{align*}
\left| \Big| \frac{ 1 }{ | \partial M | } \int_M \xi \Big| - \Big| \frac{ 1 }{ | \partial N | } \int_N \xi \Big| \right|
&\leq \left| \frac{ 1 }{ | \partial M | } \int_N \xi - \frac{ 1 }{ | \partial N | } \int_N \xi \, \right|
+ \frac{ 1 }{ | \partial M | } \left| \int_N \xi - \int_M \xi \, \right| \\
&\lesssim \Big| \frac{  | \partial N | }{ | \partial M | } - 1 \Big| \sup \left\{ \frac{1}{ | \partial N | } \int_N \xi ~\middle|~ N \subseteq B_R, \partial N ~ { \rm smooth } \right\} \\
& \qquad + \frac{ 1 }{ | \partial M | } \int | \boldsymbol{1}_M - \boldsymbol{1}_N | \sup_{ B_R } | \xi |.
\end{align*}
Since for any Caccioppoli set $ M $, we can find a set $ N $ with smooth boundary s.t.
\begin{align*}
\int | \boldsymbol{1}_M - \boldsymbol{1}_N | \lesssim | \partial M |
\quad { \rm and } \quad
| \, | \partial M | - | \partial N | \, |  \lesssim | \partial M |,
\end{align*}
see e.g.~\cite[Theorem 3.42 with a compactly supported mollifier in its proof]{AmbrosioFuscoPallara}
the above estimate shows
\begin{equation}\label{eqn:talagrand-expectation-reduction-i}
\begin{aligned}
& \E \sup \left\{ \Big| \frac{1}{ | \partial M | } \int_M \xi \Big| ~\middle|~ M \subseteq B_R ~ { \rm Caccioppoli } \right\} \\
& \qquad \lesssim \E \sup \left\{ \Big| \frac{1}{ | \partial M | } \int_M \xi \Big| ~\middle|~ M \subseteq B_R, \partial M ~ { \rm smooth } \right\} + \E \sup_{ B_R } | \xi |.
\end{aligned}
\end{equation}
Lemma \ref{lem:wn-sup-bound} implies that the last term is of lower order in view of our claimed estimate \eqref{eqn:talagrand-expectation}.

Now we employ the same argument as \cite[Proof of Theorem 1.6]{DW} to further restrict to connected sets. To this end, we write $ M = \cup_i M_i $ where $ M_i $ are the connected components of $ M $. By \eqref{eqn:talagrand-expectation-reduction-i}, we may already assume that $ \partial M $ is smooth. Hence, the boundary segments of the $ M_i $'s do not overlap (which is trivial in the discrete setting of \cite{DW}) and we have
\begin{align*}
\sum_{ i } \vert \partial M_i \vert = \vert \partial M \vert.
\end{align*}
Therefore, we may estimate
\begin{align*}
\Big| \frac{1}{ | \partial M | } \int_{M} \xi \Big|
\leq \sum_i \frac{| \partial M_i| }{| \partial M | } \Big| \frac{1}{| \partial M_i |} \int_{ M_i } \xi \Big|
\leq \sup \left\{ \Big| \frac{1}{ | \partial N | } \int_{ N } \xi \Big| ~\middle|~ \begin{matrix} N \subseteq B_R ~ { \rm connected}, \\ \partial N ~ { \rm smooth } \end{matrix} \right\}.
\end{align*}
By a similar argument, we may further restrict the supremum to simply connected sets. If $ M $ is connected, then we may write $ M = S \setminus \cup_i H_i $ where $ S $ is simply connected and the $ H_i $'s denote the possible holes of $ M $. Hence, we may again decompose the above integral
\begin{align*}
\Big| \frac{1}{ | \partial M | } \int_{M} \xi \Big|
\leq \frac{ | \partial S| }{ | \partial M | } \Big| \frac{1}{|\partial S|} \int_{S} \xi \Big| + \sum_{i} \frac{ | \partial H_i| }{ | \partial M | } \Big| \frac{1}{|\partial H_i|}  \int_{H_i} \xi \Big|.
\end{align*}
Note that these holes may as well degenerate to points and lines, where the corresponding integrals vanish. Since the holes $ H_i $ are simply connected and by the smoothness of $ \partial M $ we also have $ | \partial M | = | \partial S | + \sum_{i} | \partial H_i | $, and we may estimate
\begin{align*}
\Big| \frac{1}{ | \partial M | } \int_{M} \xi \Big|
\leq  \sup \left\{ \Big| \frac{1}{ | \partial N | } \int_{ N } \xi \, \Big| ~\middle|~ \begin{matrix} N \subseteq B_R ~ { \rm simply ~ connected}, \\ \partial N ~ { \rm smooth } \end{matrix} \right\}.
\end{align*}
Overall, we have shown
\begin{equation}\label{eqn:talagrand-expectation-reduction-ii}
\begin{aligned}
&\E \sup \left\{ \Big| \frac{1}{ | \partial M | } \int_M \xi  \, \Big| ~\middle|~ M \subseteq B_R, \partial M ~ { \rm smooth } \right\} \\
&\qquad \leq \E \sup \left\{ \Big| \frac{1}{ | \partial N | } \int_{ N } \xi \, \Big| ~\middle|~ \begin{matrix} N \subseteq B_R ~ { \rm simply ~ connected}, \\ \partial N ~ { \rm smooth } \end{matrix} \right\},
\end{aligned}
\end{equation}
which is reminiscent of the discrete case, see \cite[(11)]{DW}.

Lastly, we want to establish that we can restrict the maximization problem on the r.h.s.~of \eqref{eqn:talagrand-expectation} to short boundary curves in the sense of $ | \partial M | \leq R^2 $ and long boundary curves in the sense that $ | \partial M | \geq 1 $. More precisely, we show that
\begin{equation}\label{eqn:talagrand-expectation-reduction-iii}
\E \sup \left\{ \Big| \frac{1}{ | \partial M | } \int_M \xi \Big| ~\middle|~ \begin{matrix} M \subseteq B_R ~ { \rm simply ~ connected}, \\ \partial M ~ { \rm smooth }, | \partial M | \leq 1 ~ { \rm or } ~ | \partial M | \geq R^2  \end{matrix} \right\} 
\lesssim \E \sup_{ B_R } | \xi | + \E \fint_{ B_R } | \xi |
\end{equation}
Let us start with a set $ M $ with long boundary curve, i.e.~$ | \partial M | \geq R^2 $. Here, we employ the estimate
\begin{align*}
\Big| \frac{1}{ | \partial M | } \int_{ M } \xi \Big| \leq \frac{1}{R^2} \int_{ B_R } | \xi | ,
\end{align*}
If on the contrary, we have $ | \partial M | \leq 1 $, we know that $ { \rm diam }( M ) \leq | \partial M | \leq 1 $, so that $ M $ is fully contained in a ball of radius $ 1 $. Combining this with the isoperimetric inequality in form of $ \frac{ | M | }{ | \partial M | } \lesssim  | \partial M | \lesssim 1 $, we obtain
\begin{align*}
\Big| \frac{1}{ | \partial M | } \int_{ M } \xi \Big|
\lesssim \frac{ | M | }{ | \partial M | } \sup_{ B_R } | \xi |
\lesssim \sup_{ B_R } | \xi |.
\end{align*}
Combining the last two equations we conclude \eqref{eqn:talagrand-expectation-reduction-iii}, which in combination with \eqref{eqn:talagrand-expectation-reduction-i}, \eqref{eqn:talagrand-expectation-reduction-ii}, and Lemma \ref{lem:wn-sup-bound} implies \eqref{eqn:talagrand-expectation-reduction}.

\medskip

\textit{Step 2 (Reduction to finitely many boundary curves).} In view of the argument for the analogous upper bound \eqref{eqn:talagrand-expectation} in \cite[Lemma 4.4]{DW} and \eqref{eqn:talagrand-expectation-reduction}, we need to show that there exists some universal constant $ C_{ 0 } > 0 $ such that
\begin{equation}\label{eqn:talagrand-one-scale-bound-tail}
\P \bigg\{  \sup \left\{ \Big| \frac{1}{L} \int_M \xi \Big| ~\middle|~ \begin{matrix} M \subseteq B_R ~ { \rm simply ~ connected}, \\ \partial M ~ { \rm smooth }, \frac{L}{2} \leq | \partial M | \leq L, 0 \in B_1( \partial M) \end{matrix}\right\} \geq C_{ 0 } ( \log L )^{ \frac{3}{4} } + t \bigg\}
\leq \exp( - \frac{1}{8\pi} t^2 )
\end{equation}
for any $ t > 0 $. Once more, we make use of the continuity estimate in Lemma \ref{lemma:cameron-martin-continuity} (with $ | \partial M | $ on the l.h.s.~replaced by $ L $) that by employing the Gaussian concentration estimate from Bogachev's book, see \cite[Theorem 4.5.6]{BogachevGM}, reduces the problem to bounding the expectation
\begin{equation}\label{eqn:talagrand-one-scale-bound-expectation}
\E \sup \left\{ \Big| \frac{1}{L} \int_{ M }  \xi \Big| ~\middle|~ \begin{matrix} M \subseteq B_R ~ { \rm simply ~ connected}, \\ \partial M ~ { \rm smooth }, \frac{L}{2} \leq | \partial M | \leq L, 0 \in B_1( \partial M) \end{matrix} \right\}
\lesssim ( \log L )^{ \frac{3}{4} } .
\end{equation}

By Lemma \ref{lem:deformation}, there exists a universal constant $ 0 < c_{ \rm def } < \infty $ such that for any of the competing sets above, we find some set $ N \in \Gamma_{ c_{ \rm def } L } $ (see the definition before Lemma \ref{lem:deformation}) such that
\begin{align*}
\frac{1}{L} \int | \boldsymbol{1}_{ N } - \boldsymbol{1}_{ M } | \lesssim 1.
\end{align*}
Since
\begin{align*}
\left| \Big| \frac{1}{L} \int_{ M }  \xi \, \Big| - \Big| \frac{1}{L} \int_{ N } \xi \, \Big| \right|
\lesssim \frac{1}{L} \int | \xi | | \boldsymbol{1}_{ M } - \boldsymbol{1}_{ M }  |
\lesssim \sup_{ B_{ 2 c_{ \rm def } L } } | \xi | \, \frac{1}{L} \int | \boldsymbol{1}_{ N } - \boldsymbol{1}_{ M } |
\lesssim \sup_{ B_{ 2 c_{ \rm def } L } } | \xi | ,
\end{align*}
Lemma \ref{lem:wn-sup-bound} on scale $ R = 2 c_{ \rm def } L $ implies
\begin{align*}
\E \sup \left\{ \Big| \frac{1}{L} \int_{ M } \xi \, \Big| ~\middle|~ \begin{matrix} M \subseteq B_R ~ { \rm simply ~ connected}, \\ \partial M ~ { \rm smooth }, \frac{L}{2} \leq | \partial M | \leq L, 0 \in B_1( \partial M) \end{matrix} \right\}
\lesssim \E \sup_{ N \in \Gamma_{ c_{ \rm def } L } } \Big| \frac{1}{L} \int_{ N }  \xi \, \Big| + (\log L)^{ \frac{1}{2} }.
\end{align*}
Now the same argument\footnote{Note that by construction the set $ \Gamma_{ cL } $ contains only sets with boundary curves on the lattice. Hence for the discretized white noise as considered in Section \ref{sec:pf-talagrand-bound}, we are exactly in the situation of \cite{DW}.} as in \cite[Lemma 4.4]{DW} and \cite[Section 3.4]{Talagrand-GC} applies: using the chaining bounds \eqref{eqn:chaining-bound} and \eqref{eqn:chaining-bound-b}, followed by the Lemma \ref{lem:paramterization-through-h1} (keeping also \eqref{eqn:upper-bound-gaussian-metric} in mind) and Lemma \ref{lem:emergernce-gamma_1-2}, we obtain
\begin{align}\label{eqn:talagrand-one-scale-bound-expectation-ii}
\begin{aligned}
& \E \sup \left\{ \Big| \frac{1}{L} \int_{ M }  \xi \, \Big| ~\middle|~ \begin{matrix} M \subseteq B_R ~ { \rm simply ~ connected}, \\ \partial M ~ { \rm smooth }, | \partial M | \leq L, 0 \in B_1( \partial M) \end{matrix} \right\} \\
& \qquad \lesssim 1 +  ( \log L )^{ \frac{1}{2} } + ( \log \log | \Gamma_{ c_{  \rm def } L } | )^{ \frac{3}{4} } \gamma_{1,2} ( H^1( \mathbb{T} ) \hookrightarrow L^2( \mathbb{T} ) ) ^{ \frac{1}{2} } .
\end{aligned}
\end{align}
Since $ | \Gamma_{ c_{ \rm def } L } | \leq 4^{ c_{ \rm def } L } $, and in view of Lemma \ref{chain-key-lemma}, this implies the claim.

\medskip

\textit{Step 3 (Conclusion).} The only thing that is left, is to use the estimate provided in Step 2 to conclude that
\begin{align*}
 \E \sup \left\{ \Big| \frac{1}{ | \partial M | } \int_M \xi \, \Big| ~\middle|~ \begin{matrix} M \subseteq B_R ~ { \rm simply ~ connected}, \\ \partial M ~ { \rm smooth }, 1 \leq | \partial M | \leq R^2 \end{matrix} \right\}  
 \lesssim ( \log R )^{ \frac{3}{4} } .
 \end{align*}
In view of what we proved in Step 1, this completes the argument for Proposition \ref{thm:talagrand-bound}.

Here, we again rely on the argument by \cite[Proposition 2.2]{DW}. For notational simplicity, let $ \Gamma_{ L, 0 }^* $ denote the set of admissible $ M $'s in the supremum in \eqref{eqn:talagrand-one-scale-bound-expectation}. Once the constraint $ 0 \in B_1 ( \partial M ) $ is replaced by $ p \in B_1 ( \partial M ) $, we write $ \Gamma_{ L, p }^* $ instead. Equipped with this notation, we can bound
\begin{align*}
& \E \sup \left\{ \Big| \frac{1}{ | \partial M | } \int_M \xi \, \Big| ~\middle|~ \begin{matrix} M \subseteq B_R ~ { \rm simply ~ connected}, \\ \partial M ~ { \rm smooth }, 1 \leq | \partial M | \leq R^2 \end{matrix} \right\} \\
& \qquad \leq \E \sup_{  p \in \Z^2 \cap Q_R } \sup_{ L \leq R^2 } \sup \left\{ \Big| \frac{1}{ | \partial M | } \int_M \xi \, \Big| ~\middle|~ M \in \Gamma_{ L, p }^*, ~ | \partial M | \geq 1 \right\} .
\end{align*}
The latter supremum is conveniently bounded by using the $ \Z^2 $-stationarity: with the constant $ C_0 > 0 $ from \eqref{eqn:talagrand-one-scale-bound-tail} we have
\begin{align*}
&\P \bigg\{ \sup_{  p \in \Z^2 \cap Q_R } \sup_{ L \leq R^2 } \sup \left\{ \Big| \frac{1}{ | \partial M | } \int_M \xi \, \Big| ~\middle|~ M \in \Gamma_{ L, p }^*, ~ | \partial M | \geq 1 \right\} \ge 2 C_{ 0 } (\log R)^{\frac{3}{4}} + t \bigg\} \\
&\qquad \le \sum_{p \in \Z^2 \cap Q_R} \sum_{2^k \le R^2} \P \bigg\{ \sup \left\{ \Big| \frac{1}{ 2^{k} } \int_M \xi \, \Big| ~\middle|~ M \in \Gamma_{2^k, p}^* \right\}  \ge C_{ 0 }  (\log 2^k)^{\frac{3}{4}} + \frac{t}{2} \bigg\} \\
&\qquad \le C_1 R^d \log (R) \exp( - \frac{t^2}{16 \pi} ) 
= \exp ( d \log(R) + \log (C_1 \log (R)) -  \frac{t^2}{32 \pi} ) ,
\end{align*}
where $ C_1 $ denotes another universal constant. Therefore, inserting $ t = C_{ 2 } ( \log R )^{\frac{3}{4}} + t' $ for some large constant $ C_{ 2 } > 0 $ absorbs the lower order terms and yields the bound
\begin{align*}
\P \bigg\{ \sup_{  p \in \Z^2 \cap Q_R } \sup_{ L \leq R^2 } \sup \left\{ \Big| \frac{1}{ | \partial M | } \int_M \xi \, \Big| ~\middle|~ M \in \Gamma_{ L, p }^*, ~ | \partial M | \geq 1 \right\} \ge C_{ 3 } (\log R)^{\frac{3}{4}} + t \bigg\}
\leq \exp( - \frac{ \alpha^2 }{ 16 \pi } ) 
\end{align*}
for another universal constant $ C_{ 3 } > 0 $.
The argument to conclude is standard again, we have
\begin{align*}
	\sup_{  p \in \Z^2 \cap Q_R } \sup_{ L \leq R^2 } \sup \Big\{ \Big| \frac{1}{ | \partial M | } \int_M \xi \, \Big| ~\Big|~ M \in \Gamma_{ L, p }^*, ~ | \partial M | \geq 1 \Big\}
	\leq C_{ 3 } (\log R)^{\frac{3}{4}} + \int_{ C_{ 3 } (\log R)^{\frac{3}{4}} }^{ \infty } \exp( - \frac{t^2}{ 32 \pi } ) dt,
\end{align*}
which implies the claim. \hfill $\blacksquare$

\section{Proof of the main result}

In this section, we prove the main result of this paper, namely Theorem~\ref{thm:mr}. We begin by showing that any minimizer of \eqref{eqn:in-energy} is an approximate minimizer of the perimeter in the sense of $ \eta $-minimality, see Definition \ref{def:eta-minimizer}.

\begin{lemma}[$ \eta $-minimality of local minimizers]\label{proposition:eta-minimality}
	For any $ 0 < p_0 < 1 $ there exists a constant $ C_{ \eta } = C_{ \eta } ( p_0 ) < \infty $ such that with probability at least $ p_0 $ the following holds: let $ m $ be a local minimizer of the energy \eqref{eqn:in-energy} for some $ 0 < \varepsilon < 1 $, then $ m $ is an $ \eta_R(x) $-minimizer of the perimeter in $ B_R(x) $, with
	\begin{align}\label{eqn:def-eta-r-p}
	\eta_R(x) \leq \varepsilon C_{ \eta } ( \log | x |_{+} )^{ \frac{1}{2} } ( \log R )^{ \frac{3}{4} }
	\end{align}
	for $ x \in \R^2 $ and $ R \ge 2 $ provided the latter is less than $ 1 $.
\end{lemma}

\begin{proof}[Proof of Lemma \ref{proposition:eta-minimality}] By Proposition \ref{prop:concentration-in-space}, we may assume that with probability at least $ p_0 $, the estimate
\begin{align}\label{eqn:proof-mr-01}
S_R(x) = \sup \left\{ \Big| \frac{1}{ | \partial M | } \int_M \xi \, \Big| ~\middle|~ M \subseteq B_R(x) ~ { \rm Caccioppoli~set } \right\}
\lesssim_{ p_0 } ( \log | x |_{+} )^{ \frac{1}{2} } ( \log R )^{ \frac{3}{4} }
\end{align}
holds for every point $ x \in \R^2 $ and scale $ R \ge 2 $. We use this to argue that $ m $ is an $ \eta_R(x) $-minimizer of the perimeter in $ B_R(x) $, see Definition \ref{def:eta-minimizer} provided the r.h.s.~of \eqref{eqn:proof-mr-01} including the implicit constant is small.

Indeed, fix $ R $ and $ x $ and consider an open ball $ B \subseteq B_R(x) $. From the minimality of $ m $ we know that
\begin{align}\label{eqn:proof-mr-02}
\int_{ \overline{ B } } | \nabla m | \leq \int_{ \overline{ B } } | \nabla m' | + \varepsilon \int_{ B } \xi ( m - m' )
\end{align}
for every competitor $ m' : \R^2 \rightarrow \{ \pm 1 \} $ such that $ m = m' $ in $ \R^2 \setminus B $. We can split the field term on the r.h.s.~of \eqref{eqn:proof-mr-02} and use that $ \{ m = 1, \, m' = -1  \} \cup \{ m = -1, \, m' = 1 \} \subseteq B $, together with the definition of $ S_R(x) $ in \eqref{eqn:proof-mr-01} to obtain
\begin{align*}
\int_{ B } \xi ( m - m' )
&= 2 \int_{ B } \xi \boldsymbol{1}_{ \{ m = 1, \, m' = -1 \} } - 2 \int_{ B } \xi \boldsymbol{1}_{ \{ m = -1, \, m' = 1 \} } \\
&\leq 2 S_R(x) \Big( \int_{ \R^2 } | \nabla \boldsymbol{1}_{ \{ m = 1, \, m' = -1 \} } | +  \int_{ \R^2 }  | \nabla \boldsymbol{1}_{ \{ m = - 1, \, m' = 1 \} } | \Big).
\end{align*}
Recall that for any smooth functions $ 0 \leq f, g \leq 1 $ we have $ \int | \nabla f g  | \leq \int | \nabla f | + \int | \nabla g | $, so that by approximation a la \cite[Proposition 3.38]{AmbrosioFuscoPallara}, we learn that for every strictly larger open ball $ \overline{ B } \subseteq B' $
\begin{align*}
\int_{ B' } | \nabla \boldsymbol{1}_{ \{ m = 1, \, m' = -1 \} } |
\leq \int_{ B' } | \nabla \boldsymbol{1}_{ \{ m = 1 \} }  | + \int_{ B' } | \nabla \boldsymbol{1}_{ \{ m' = -1 \} }  |
\leq \frac{1}{2} \int_{ B' } | \nabla m | + \frac{1}{2} \int_{ B' } | \nabla m' | .
\end{align*}
By continuity of the measures from above (and a symmetric argument for the second term), this implies
\begin{align*}
\int_{ \R^2 } | \nabla \boldsymbol{1}_{ \{ m = 1, \, m' = -1 \} } | + \int_{ \R^2 } | \nabla \boldsymbol{1}_{ \{ m = - 1, \, m' = 1 \} } |
\leq \int_{ \overline{B} } | \nabla m | + \int_{ \overline{B} } | \nabla m' | .
\end{align*}
Hence, the above estimate implies
\begin{align*}
\int_{ B } \xi ( m - m' ) \leq 2 S_R(x) \Big( \int_{ \overline{B} } | \nabla m | + \int_{ \overline{B} } | \nabla m' | \Big).
\end{align*}
Inserting this in \eqref{eqn:proof-mr-02} yields
\begin{align*}
\int_{ \overline{B} } | \nabla m | \leq \frac{ 1 + 2 \varepsilon S_R(x) }{ 1 - 2 \varepsilon S_R(x) } \int_{ \overline{B} } | \nabla m' |.
\end{align*}
Provided $ 2 \varepsilon S_R(x) \leq \frac{1}{2} $, we may estimate $ \frac{ 1 + 2 \varepsilon S_R(x) }{ 1 - 2 \varepsilon S_R(x) } \leq 1 + 8 \varepsilon S_R(x) $, which in view of \eqref{eqn:proof-mr-01} implies the claim.
\end{proof}

We have now collected all ingredients to prove Theorem \ref{thm:mr}. The only thing that is left is to set up the Campanato iteration that we explained in Section \ref{sec:strategy-of-proof}. Its implementation is somewhat technical compared to its core idea:

Focusing on one point $ x \in J_{ m } $ for the moment, we will repeatedly apply Proposition \ref{prop:iteration-step} on scales $ \theta^k R $ for $ \theta = \frac{ 1 }{ 16 } $, $k=0, \ldots, K$ to approximate $ J_{ m } $ by line segments. In that way, we manage to transfer the initial flatness on scale $ R $ down to the smaller scales $ \theta^k R $ as illustrated in Figure~\ref{fig:step31}.

To conclude, we need to apply two post-processing steps. Firstly, $K$ has to be chosen in such a way that the smallest scale $ \theta^K R $ is of order $1$, so that we can relate the normal of the approximating line segment to the average normal $ \bar\nu_1 ( x ) $. This is achieved by Remark~\ref{rmk:strong-exc-bound}. Secondly, we have to perform a similar iteration around a second base point $ y \in J_{m}$. To relate the two iterations, the largest scale $ R $ must be comparable to the distance of $x$ and $y$, such that the two normals obtained on scale $ R $ are comparable. The latter is achieved by performing one more iteration step around the midpoint $z \in J_m$ between $x$ and $y$.

This procedure is carried out in Step 3.1 in the proof of Theorem~\ref{thm:mr} below, where the difference of the two normals is estimated by the cumulated error from the iteration steps via the triangle inequality.

\begin{figure}
\begin{tikzpicture}[scale=3,line cap=round,line join=round,>=triangle 45,x=1cm,y=1cm]
\draw [shift={(-0.530449571128189,-0.5995825470896056)},line width=1pt,color=blue]  plot[domain=1.1625154523356056:2.2521039531784046,variable=\t]({1*0.7719120437323149*cos(\t r)+0*0.7719120437323149*sin(\t r)},{0*0.7719120437323149*cos(\t r)+1*0.7719120437323149*sin(\t r)});
\draw [shift={(0.24276932834762824,0.9413162115949204)},line width=1pt,color=blue]  plot[domain=4.201368691064911:4.877893863252976,variable=\t]({1*0.9543572152604437*cos(\t r)+0*0.9543572152604437*sin(\t r)},{0*0.9543572152604437*cos(\t r)+1*0.9543572152604437*sin(\t r)});
\draw [line width=1pt,color=gray] (0.06846557909837847,0.003011411800459385) circle (0.1111111111111111cm);
\draw [line width=1pt,color=gray] (0.06846557909837847,0.003011411800459385) circle (0.3333333333333333cm);
\draw [line width=1pt,color=gray] (0.06846557909837847,0.003011411800459385) circle (1cm);
\draw [shift={(-1.5068300010443225,0.6948442036115884)},line width=1pt,color=blue]  plot[domain=4.432612115647871:5.326806770432807,variable=\t]({1*0.8503688210535949*cos(\t r)+0*0.8503688210535949*sin(\t r)},{0*0.8503688210535949*cos(\t r)+1*0.8503688210535949*sin(\t r)});
\draw [shift={(-4.210755731687733,27.48118143929112)},line width=1pt,color=blue]  plot[domain=4.8786193994580485:4.902748433000367,variable=\t]({1*27.865290267939613*cos(\t r)+0*27.865290267939613*sin(\t r)},{0*27.865290267939613*cos(\t r)+1*27.865290267939613*sin(\t r)});
\draw [line width=1pt,color=gray] (-0.0394100506618949,0.029629591981871137)-- (0.17870386101763322,-0.010888236102872196);
\draw [line width=1pt,color=gray] (-0.24460621212628153,0.11745425627093736)-- (0.4,0);
\draw [line width=1pt,color=gray] (-0.9298536317280512,0.060966163934298434)-- (1.061687857689121,0.11924180933489588);
\draw [line width=1pt] (-0.7747407660133021,0.132653764056327) circle (0.1111111111111111cm);
\draw [line width=1pt] (-0.7747407660133021,0.132653764056327) circle (0.3333333333333333cm);
\draw [line width=1pt] (-0.7747407660133021,0.132653764056327) circle (1cm);
\draw [line width=1pt] (-0.8773367144396307,0.08999523174627505)-- (-0.6670832252869577,0.16014072345712535);
\draw [line width=1pt] (-1.064792261901897,-0.03160591098261664)-- (-0.4432217787711678,0.16738519375711322);
\draw [line width=1pt] (-1.741651864610051,-0.12245979381629618)-- (0.21464975404898542,-0.01262665035276039);
\begin{scriptsize}
\draw [fill=black] (-0.7747407660133021,0.132653764056327) circle (.5pt);
\draw[color=black] (-0.77,0.2) node {$x$};
\draw [fill=gray] (0.06846557909837847,0.003011411800459385) circle (.5pt);
\draw[color=gray] (0.07,-0.06) node {$y$};
\draw[color=blue] (-1.2,-0.2) node {$J_m$};
\end{scriptsize}
\end{tikzpicture}

	\caption{Illustration of the Campanato iteration around the points $x$ and $y$; approximation of $J_m$ through line segments on different scales.}\label{fig:step31}
\end{figure}

\begin{proof}[Proof of Theorem \ref{thm:mr}.]
By Lemma \ref{proposition:eta-minimality}, there exists a constant $ C_{ \eta } = C_{ \eta } (p_0) \in ( 1 , \infty ) $ such that with probability at least $ p_0 $, any local minimizer $ m $ is an $ \eta_R(x) $-minimizer of the perimeter in $ B_R(x) $, with
\begin{align}\label{eqn:def-eta-r-p-repeat}
\eta_R(x) \leq \varepsilon C_{ \eta } ( \log | x |_{+} )^{ \frac{1}{2} } ( \log R )^{ \frac{3}{4} },
\end{align}
for $ x \in \R^2 $ and $ R \ge 2 $ provided the latter is less than $ 1 $, cf.~\eqref{eqn:def-eta-r-p}.

\medskip

\textit{Step 1 (Jump set is a single arc).} Next, we capitalize on $ \eta_R $-minimality in the ball $ B_R $ (centered around the origin) for $ R =  \exp( c_{ \rm exp } \varepsilon^{ - \frac{4}{3} } ) $. Provided the constant $ 0 < c_{ \rm exp } < \infty $ is small enough (so that the r.h.s.~of \eqref{eqn:def-eta-r-p} is sufficiently small), we claim that there exists a $ \frac{1}{16} \leq \theta_{ 0 } \leq 1 $, such that the jump set of $ m $ in $ B_{ \theta_{ 0 } R } $ consists of a single arc. 

First, we argue that we can choose $ \theta_{ 0 } $ such that the jump set of $ m $ intersects $ \partial B_{ \theta_{ 0 } R } $ in either zero or two points. This follows from Lemma \ref{lem:few-jumps}. Its assumptions are satisfied by Proposition \ref{prop:approx-halfspace-l1} for a suitable choice of $ 0 < c_{ \rm exp }  < \infty $.

Next, we show that $ m $ can't have bubbles, i.e.~connected components of $ \{ m = 1 \} $ or $ \{ m = -1 \} $, fully contained in $ B_R $. Indeed, the presence of a bubble would contradict \eqref{eqn:proof-mr-01}: assume that our minimizer $ m $ has a bubble supported in a set $ M \subseteq B_R $. Now let $ m' $ denote the configuration where the bubble is removed. By local minimality, we know that
\begin{align*}
\int_{ \overline{ B_R } } | \nabla m | - \varepsilon \int_{B_R} \xi m 
\leq \int_{ \overline{ B_R } } | \nabla m' | - \varepsilon \int_{B_R} \xi m' ,
\end{align*}
which, since $ \frac{1}{2} ( m - m'  ) = \pm \boldsymbol{1}_M $, translates into
\begin{align*}
| \partial M |
= \int_{ \overline{ B_R } } | \nabla m | - \int_{ \overline{ B_R } } | \nabla m'  |
\leq \varepsilon \int_{ B_R } \xi ( m - m' )
\leq 2 \varepsilon \left| \int_M \xi \, \right|.
\end{align*}
This contradicts \eqref{eqn:proof-mr-01} provided $ R = \exp( c_{ \rm exp } \varepsilon^{ - \frac{4}{3} } ) $ for a possibly smaller choice of $ 0 < c_{ \rm exp }  < \infty $.

Summarizing the above, the jump set of $ m $ intersects $ \partial B_{ \theta_{ 0 } R } $  in either zero or two points, and contains no bubbles. Hence, the jump set in $ B_{ \theta_{ 0 } R } $ either consists of a single arc with start and endpoint on $ \partial B_{ \theta_{ 0 } R } $, or is empty.

\medskip

\textit{Step 2 (Distance on the jump set).} We now argue that for any two points $ x, y \in B_{ \frac{1}{64} R } $ in the jump set of $ m $ we have
\begin{align}\label{eqn:dist-comparable}
{ \rm dist }_{ J_m } (x,y) \sim | x - y |,
\end{align}
provided that $ R = \exp( c_{ \rm exp }  \varepsilon^{ - \frac{4}{3} } ) $ as in Step 1. Note carefully that the implicit constants in \eqref{eqn:dist-comparable} improve if $ c_{ \rm exp }  $ becomes smaller, since they only rely on the fact that $ S_R(0) $ is made small enough by the choice of $ c_{ \rm exp }  $.

To this end, let us fix two such points $ x, y $. We note that due to Remark \ref{rmk:eta-minimizer}, the same argument that we used in Step 1 also applies on the ball $ B_{ 16 | x - y | }( x ) \subseteq B_R $: we find a $ \frac{1}{16} \leq \theta_{ 1 } \leq 1 $, such that the jump set of $ m $ in $ B_{ 16 \theta_{ 1 } | x - y | } (x) $ consists of a single arc. Therefore, since $ x,y \in B_{ 16 \theta_{ 1 } | x - y | } (x) $, we can estimate
\begin{align*}
| x - y | \leq { \rm dist }_{ J_m } ( x,y) \leq  { \rm dist }_{ J_m \cap  B_{ 16 \theta_{ 1 } | x - y | }( x ) } ( x, y ) \leq \frac{1}{2} \int_{ B_{ 16 \theta_{ 1 } | x - y | }( x )  } | \nabla m |.
\end{align*}
We now appeal to the $ \eta_{R}(0) $-minimality of $ m $ in $ B_R $ (and therefore in $B_{16 |x-y|(x)}$, see Remark~\ref{rmk:eta-minimizer}), and invoke the density bounds of Lemma~\ref{lem:density-bounds} (in its more precise form of \eqref{eqn:surface-density-estimate-precise}) on the last integral to obtain 
\begin{align*}
	| x - y | \leq { \rm dist }_{ J_m } ( x,y) \leq (1+\eta_R(0)) |\partial B_{1}| 16 \theta_{ 1 } |x-y| \lesssim |x-y|,
\end{align*}
provided that $\eta_R(0)$ (due to the choice of $c$) is small enough. This implies the claim.

\medskip

\textit{Step 3 (Iteration).} We now come to the core argument. Recall $ R = \exp( c_{ \rm exp } \varepsilon^{ - \frac{4}{3} } ) $ and fix $\frac{1}{16}< \theta_{ 0 } <1$ as in Step 1. We may assume that the jump set of $ m $ in $ B_{ \theta_{ 0 } R} $ consists of a single arc; otherwise the jump set is empty and the conclusion of Theorem~\ref{thm:mr} holds trivially.

Since we are now in the situation that the jump set is a single arc (without bubbles on smaller scales), the variational problem \eqref{eqn:average-normal} always admits the unique solution 
\begin{align*}
\bar \nu_1 ( x ) = \frac{  \int_{ B_1(x) } \nabla m }{ \big| \int_{ B_1(x) } \nabla m \big| }.
\end{align*}
Indeed, the integral $\int_{ B_1(x) } \nabla m = \int_{B_1(x)} \nu |\nabla m|$, which (up to rotating by $\frac{\pi}{2}$) amounts to integrating the tangent vector to the curve $J_m \cap B_1(x)$, vanishes if and only if the curve is closed. The latter is excluded by Step 1.

We now fix a universal constant $0<\delta<1$ (to be chosen later) and take two points $ x, y \in B_{ \delta R} $ in the jump set of $ m $.

\medskip

\textit{Step 3.1 (The case $ | x - y | \gtrsim 1 $)}. The main work is to be done in the case $ | x - y | \geq 16^2 $ that we treat now. The condition is such that the iteration below makes sense.

We may pick a point $ z $ in $ J_m $ such that
\begin{align*}
{ \rm dist }_{ J_m }( x, z ) = { \rm dist }_{ J_m }( y, z ) = \frac{1}{2} { \rm dist }_{ J_m }(x,y) .
\end{align*}
By Step 2, $ { \rm dist }_{ J_m }(x,y) $ is comparable to $ | x - y | $, so that there exists a universal constant $ M > 0 $ such that $ B_{ | x - y | }(x), B_{ | x - y | }(y) \subseteq B_{ M | x - y | } ( z ) $. We can now appeal to Proposition \ref{prop:approx-halfspace-l1} on $ B_{ 16 M | x - y | } ( z ) \subseteq B_{ R } $ (the inclusion is guaranteed by choosing $\delta$ sufficiently small at this point), followed by Proposition \ref{prop:iteration-step} that asserts that
\begin{align*}
\frac{ 1 }{ ( M | x - y | )^2 } \int_{ B_{ M | x - y | }(z) } | m - m_{ \rm line } | \lesssim ( \eta_{ | x - y | } ( z ) )^{ \frac{1}{2} }
\end{align*}
for some configuration $ m_{ \rm line } $ with jump set being a line. To the price of a worse constant, we can pass this smallness to the balls around $ x $ and $ y $, that is
\begin{align}\label{eqn:iteration-01}
\max \left\{ \frac{ 1 }{  | x - y |^2 } \int_{ B_{ | x - y | }(x) } | m - m_{ \rm line } |,
\frac{ 1 }{  | x - y |^2 } \int_{ B_{ | x - y | }(y) } | m - m_{ \rm line } | \right\}
\lesssim ( \eta_{ | x - y | } ( z ) )^{ \frac{1}{2} }.
\end{align}

We now repeatedly apply Proposition \ref{prop:iteration-step} on balls centered around $ p \in \{ x, y \} $ to find configurations $ m_{ \theta^k | x - y |, p } $ with jump sets being a line and corresponding normal $ \nu_{ \theta^k | x - y | }( p ) $ such that
\begin{align}\label{eqn:iteration-02}
| \nu_{ \theta^{ k + 1 } | x - y | }( p ) - \nu_{ \theta^k | x - y | }( p ) |
\lesssim \left\{ \begin{aligned}
& ( \eta_{ \theta^{ k -	1 } | x - y | }( p ) )^{ \frac{1}{2} } \quad  & { \rm for } ~ k \geq 1 \\
& ( \eta_{ | x - y | }( z ) )^{ \frac{1}{2} }  & { \rm for } ~ k = 0
\end{aligned}\right.
\end{align}
and $ \theta = \frac{1}{16} $, see Figure~\ref{fig:step31}. We can additionally achieve that $ \nu_{ | x - y | } ( x ) = \nu_{ | x - y | } ( y ) $. This follows by an induction.

We start with the case $ k = 0 $, and observe that Proposition \ref{prop:iteration-step} yields the existence of $ m_{ \theta | x - y |, p } $ with associated normal $ \nu_{ \theta | x - y | }( p ) $ such that
\begin{align}\label{eqn:iteration-03}
\frac{ 1 }{ ( \theta | x - y | )^2 } \int_{ B_{ \theta | x - y | } ( p ) } | m - m_{ \theta | x - y | , p } | \lesssim ( \eta_{ | x - y | }( p ) )^{ \frac{1}{2} }
\end{align}
and (already inserting \eqref{eqn:iteration-01})
\begin{align*}
| \nu_{ \theta | x - y | }( p ) - \nu_{ | x - y | }( p ) |
\lesssim \frac{ 1 }{ | x - y |^2 } \int_{ B_{ | x - y | } ( p ) } | m - m_{ \rm line } |
\lesssim ( \eta_{ | x - y | } ( z ) )^{ \frac{1}{2} },
\end{align*}
which shows \eqref{eqn:iteration-02} for $ k = 0 $. Additionally, by \eqref{eqn:iteration-01} we know that $ \nu_{ | x - y | } ( x ) $ and $ \nu_{ | x - y | } ( y ) $ agree with the normal to the line segment approximating $ J_m $ in $ B_{ M | x - y | } ( z ) $, and thus are equal.

Next, we claim that there exists $ \nu_{ \theta^k | x - y | }(p) $'s with $ k \ge 1 $ as in \eqref{eqn:iteration-02} together with configurations $ m_{ \theta^k | x - y | , p } $ such that
\begin{align}\label{eqn:iteration-04}
\frac{ 1 }{ ( \theta^{k+1} | x - y | )^2 } \int_{ B_{ \theta^{k+1} | x - y | } ( p ) } | m - m_{ \theta^{k+1} | x - y | , p } |
\lesssim ( \eta_{ \theta^{ k } | x - y | }( p ) )^{ \frac{1}{2} },
\end{align}
For $ k = 1 $, \eqref{eqn:iteration-04} follows Proposition \ref{prop:iteration-step} on scale $ R = \theta | x - y | $ that we can apply due to \eqref{eqn:iteration-03}. Furthermore, the proposition shows that
\begin{align*}
| \nu_{ \theta^2 | x - y | }(p) - \nu_{ \theta | x - y | }(p) |
\lesssim \frac{ 1 }{ ( \theta | x - y | )^2 } \int_{ B_{ \theta | x - y | } ( p ) } | m - m_{ \theta | x - y | , p } |,
\end{align*}
which by \eqref{eqn:iteration-03} shows that \eqref{eqn:iteration-04} holds for $ k = 1 $.

To pass from scale $ \theta^k | x - y | $ to $ \theta^{ k + 1 } | x - y | $ we proceed similarly, again Proposition \ref{prop:iteration-step} can be applied on scale $ R = \theta^k | x - y | $ due to the induction hypothesis \eqref{eqn:iteration-04}. This shows that \eqref{eqn:iteration-04} is true for the next scale $ \theta^{ k + 1 } | x - y | $. Furthermore, we obtain together with \eqref{eqn:iteration-04} that
\begin{align*}
| \nu_{ \theta^{ k + 1 } | x - y | }( p ) - \nu_{ \theta^k | x - y | }( p ) |
& \lesssim \frac{ 1 }{ ( \theta^k | x - y | )^2 } \int_{ B_{  \theta^k | x - y | } ( p ) } | m - m_{ \theta^k | x - y | , p } | \\
& \lesssim ( \eta_{ \theta^{ k - 1 } | x - y | } ( p ) )^{ \frac{1}{2} } \phantom{ \int_{B_R} }
\end{align*}
so that also \eqref{eqn:iteration-02} is true on the next level.

We now use \eqref{eqn:iteration-02} up to $ k \leq N $, where $ \theta^{N+1} | x - y | <  1 \leq \theta^N | x - y | $, to estimate
\begin{align}\label{eqn:iteration-07}
| \nu_{ \theta^N | x - y | }(p) - \nu_{ | x - y | }(p)  |
\leq \sum_{ k = 1 }^N | \nu_{ \theta^k | x - y | } - \nu_{ \theta^{k-1} | x - y | } (p)  |
\lesssim ( \eta_{ | x - y | }( z ) )^{ \frac{1}{2} } + \sum_{ k = 2 }^{ N } ( \eta_{ \theta^{ k - 2 } | x - y | } ( p ) )^{ \frac{1}{2} }.
\end{align}
Since $ N \sim \log | x - y | $, we have
\begin{align*}
\sum_{ k = 2 }^{ N } ( \log \theta^{k-2} | x - y | )^{ \frac{3}{8} }
\leq \sum_{ j = 3 }^{ N + 1  } ( \log \theta^{ -j } )^{ \frac{3}{8} } 
\sim \sum_{ k = 3 }^{ N + 1  } k^{ \frac{3}{8} }
\lesssim N^{ 1 + \frac{3}{8} } |
\sim ( \log | x - y | )^{ 1 + \frac{3}{8} }.
\end{align*}
Hence, also in view of \eqref{eqn:def-eta-r-p}, \eqref{eqn:iteration-07} turns into
\begin{equation}\label{eqn:iteration-06}
\begin{aligned}
| \nu_{ \theta^N | x - y | }(p) - \nu_{ | x - y | }(p)  |
& \lesssim \left( \varepsilon ( \log | z |_{+} )^{ \frac{1}{2} } ( \log | x - y | )^{ \frac{3}{4} } \right)^{ \frac{1}{2} }
+ \left( \varepsilon ( \log | p |_{+} )^{ \frac{1}{2} } ( \log | x - y | )^{ 2 + \frac{3}{4} } \right)^{ \frac{1}{2} } \\
& \lesssim \left( \varepsilon ( \log ( | x |_{+} + | y |_{ + } ) )^{ \frac{1}{2} } ( \log | x - y | )^{ \frac{11}{4} } \right)^{ \frac{1}{2} } ,
\end{aligned}
\end{equation}
where the implicit constant depends on $ C_{ \eta } $ from \eqref{eqn:def-eta-r-p} and hence also on $ p_0 $.

Finally, it is left to observe that since $ \theta^{N+1} | x - y | \leq 1 \leq \theta^N | x - y | $ by Remark \ref{rmk:strong-exc-bound} and \eqref{eqn:strong-exc-bound} therein, we have
\begin{align*}
\int_{ B_1(p) } | \nu_{ \theta^N | x - y | }(p) - \nu |^2 | \nabla m | \lesssim \eta_{ \theta^{ N - 2 } | x - y | } ( x ) \lesssim \eta_{ 1 } ( x ),
\end{align*}
so that in view of the density estimates \eqref{eqn:surface-density-estimate} and the definition of $ \bar \nu_1 ( p ) $, we have
\begin{align*}
| \nu_{  \theta^N | x - y | } (p)  - \bar \nu_1 ( p ) |^2
&\lesssim \int_{ B_1(p) } | \nu_{ \theta^N | x - y | }(p) - \nu |^2 | \nabla m | +  \int_{ B_1(p) } | \bar \nu_1 (p) - \nu |^2 | \nabla m | \\
&\lesssim \int_{ B_1(p) } | \nu_{ \theta^N | x - y | }(p) - \nu |^2 | \nabla m | \lesssim \eta_1( x ). \phantom{ \int_{ B_1 } }
\end{align*}
Inserting the last equation in \eqref{eqn:iteration-06}, we obtain
\begin{align}\label{eqn:iteration-05}
\begin{aligned}
| \bar \nu_{ 1 }(p) - \nu_{  | x - y | }(p)  |
&\leq | \bar \nu_{ 1 }(p) - \nu_{ \theta^N | x - y | }(p)  |
+ | \nu_{ \theta^N | x - y | }(p) -  \nu_{ | x - y | }(p)  | \\
&\lesssim \left( \varepsilon ( \log ( | x |_{+} + | y |_{ + } ) )^{ \frac{1}{2} } ( \log | x - y | )^{ \frac{11}{4} } \right)^{ \frac{1}{2} }.
\end{aligned}
\end{align}
Note that by construction $  \nu_{ | x - y | }(x) =  \nu_{ | x - y | }(y) $, so that \eqref{eqn:iteration-05} implies the claim.

\medskip

\textit{Step 3.2 (The case $ | x - y | \lesssim 1 $)}. We now handle the degenerate case that $ x $, $ y $ are close in the sense that $ | x - y | \lesssim 1 $ uniformly and Step 3.1 does not apply.

Again, we may pick a point $ z $ in $ J_m $ such that
\begin{align*}
{ \rm dist }_{ J_m }( x, z ) = { \rm dist }_{ J_m }( y, z ) = \frac{1}{2} { \rm dist }_{ J_m }(x,y)  .
\end{align*}
Since by Step 2, $ { \rm dist }_{ J_m }(x,y) $ is comparable to $ | x - y | $, which in this step of the proof is assumed to be of order one, we have $ B_{ 1 }(x), B_{ 1 }(y) \subseteq B_{ M } ( z ) $ for some universal constant $ M > 0 $. Again, we can appeal to Proposition \ref{prop:approx-halfspace-l1}, this time on $ B_{ 16 M } ( z ) \subseteq B_{ R } $ (which is again guaranteed by choosing $ \delta $ sufficiently small), followed by Proposition \ref{prop:iteration-step} to conclude
\begin{align*}
\frac{ 1 }{ M^2 } \int_{ B_{ M }(z) } | m - m_{ \rm line } | \lesssim ( \eta_{ M } ( z ) )^{ \frac{1}{2} }
\end{align*}
for some configuration $ m_{ \rm line } $ with jump set being a line. In addtion, Remark \ref{rmk:strong-exc-bound} yields the estimate
\begin{align*}
\frac{ 1 }{ M } \int_{ B_{ M }(z) } | \nu - \nu_{ M } ( z)  |^2 | \nabla m | \lesssim \eta_{ M  } ( z )
\end{align*}
for the normal $ \nu_{ M } ( z) $ to the jump set of $ m_{ \rm line } $. Using $ B_1(x), B_1(y) \subseteq B_{ M } (z) $, and the minimality of $ \bar \nu_1(x) $, $ \bar \nu_1 (y ) $, we obtain the estimate
\begin{align*}
\int_{ B_1(x) } | \nu - \bar \nu_1 (x) |^2 | \nabla m | + \int_{ B_1(y) } | \nu - \bar \nu_1 (y) |^2 | \nabla m | 
\leq \int_{ B_{ M } (z) } | \nu - \nu_{ M } (z) |^2 | \nabla m |
\lesssim M \, \eta_{ M } ( z ).
\end{align*}
Therefore, another application of the density bounds, \eqref{eqn:def-eta-r-p}, $ M \lesssim 1 $ and the fact that $ | x - y | \lesssim 1 $ yield
\begin{align*}
| \bar \nu_1 ( x ) - \bar \nu_1 ( y ) |^2
&\lesssim \int_{ B_1(x) } | \nu - \bar \nu_1 (x) |^2 | \nabla m | + \int_{ B_1(x) } | \nu - \bar \nu_1 (x) |^2 | \nabla m |  \\
&\lesssim \eta_{ M } ( z )
\lesssim \varepsilon ( \log | z |_{+} )^{ \frac{1}{2} } ( \log | x-y |_{+} )^{ \frac{3}{4} },
\end{align*}
as desired.
\end{proof}

Finally, we address how Theorem~\ref{thm:mr} depends on the probability $ p_0 $. Our method of proof is not optimized in that regard, but we may still get some information on this dependence:

\begin{remark}\label{rmk:dependence-on-p0}
	Reviewing the above proof, the only probabilistic input \eqref{eqn:def-eta-r-p-repeat} comes from Lemma \ref{proposition:eta-minimality}, more specifically estimate \eqref{eqn:def-eta-r-p} therein. In view of \eqref{eqn:concentration-in-space-tail} in Proposition \ref{prop:concentration-in-space}, we know that the event $\{$ \eqref{eqn:def-eta-r-p-repeat} does not hold $\}$ has a Gaussian tail (w.r.t.~$ C_\eta $). This constant $ C_{ \eta } $  enters the deterministic argument, which requires the r.h.s.~of \eqref{eqn:def-eta-r-p-repeat} to be small, thus posing a (random) constraint $\varepsilon \leq \mathcal{E}$ that also enters the estimate \eqref{eqn:iteration-06}. While our work manages to give a qualitative description of the random variable $\mathcal{E}$, our argument is not optimized to obtain quantitative information on $\mathcal{E}$ through Proposition~\ref{prop:concentration-in-space}. 
	\end{remark}

\section*{Acknowledgments}
We thank Felix Otto and Jonas Hirsch for many helpful discussions. We would also like to thank the referee for many helpful suggestions on improving the presentation of the article.

A large part of the work was done while the authors were affiliated with the Max Planck Institute for Mathematics in the Sciences; we thank the MPI for the support and warm hospitality.

TR gratefully acknowledges partial support by the U.S.\ National Science Foundation under award DMS-2453121.

\appendix

\section{Stability of minimizers}\label{app:minimality}

Here we carry out the adaptation of the argument from \cite[Theorem 21.14 \& Remark 21.17]{Maggi} to the setting of Lemma \ref{lem:implicit-argument}, to show that
\begin{align*}
	m = \lim_{ k \to \infty } m_k
	\quad \text{is a minimizer of the perimeter in } B_1 .
\end{align*}

To this end, we pick some $ m' : B_1 \rightarrow \{ -1, +1 \} $ of finite variation such that $ \{ m \neq m' \} $ is contained in some compact subset $ K \subset B_1 $. By standard measure theoretic arguments, see \cite[Remark 3.8 \& Example 1.63]{AmbrosioFuscoPallara}, we can find some open ball $ K \subset U \subsetneq B_1 $ such that $ \int_{ \partial U } | \nabla m ' | = 0 $.

Next, we consider the sequence of competitors $ m_k' = m ' \boldsymbol{1}_{ K } + m_k \boldsymbol{1}_{ \R^2 \setminus K } $. By construction we have $ m_k ' \to m $ in $ L^1 ( B_1 ) $. Additionally, we may assume by another compactness argument that $ \nabla ( m ' \boldsymbol{1}_{ K } + m_k \boldsymbol{1}_{ \R^2 \setminus K } ) \to \nabla m' $ weakly as measures. Using that the limiting measure satisfies $  \int_{ \partial U } | \nabla m ' | = 0 $, we thus learn from Portmanteau's theorem that
\begin{align*}
	\lim_{ k \to \infty } \int_{ U } | \nabla ( m ' \boldsymbol{1}_{ K } + m_k \boldsymbol{1}_{ \R^2 \setminus K } ) |
	= \int_{ U } | \nabla m' | .
\end{align*}
On the other hand, from the $ \eta_k $-minimality of $ m_k $ and the lower semicontinuity of the perimeter functional under $ L^1 $-convergence we obtain

\begin{align*}
	\int_{ U } | \nabla m | 
	\leq \liminf_{ k \to \infty } \int_{ U } | \nabla m_k | 
	&\leq \liminf_{ k \to \infty } \Big( ( 1 + \eta_k ) \int_{ \overline B_1 } | \nabla m_k ' | - \int_{ \overline B_1 \setminus U } | \nabla m_k |  \Big)  \\
	&= \liminf_{ k \to \infty } \Big( \int_{ U } | \nabla m_k ' | + \eta_k  \int_{ \overline B_1 } | \nabla m_k ' | \Big),
\end{align*}
where in the last step we used that $m_k = m_k'$ in $U^c$. 
Since $ \eta_k \to 0 $, the last two identities imply
\begin{align*}
	\int_{ U } | \nabla m | 
	\leq \int_{ U } | \nabla m ' | .
\end{align*}
Since $ \nabla m = \nabla m' $ outside $ U $, this also holds for $ U $ replaced by $ B_1 $. This establishes that $ m $ is a minimizer of the perimeter, as claimed in the proof of Lemma \ref{lem:implicit-argument}.

\section*{Data availability}
No datasets were generated or analyzed during the current study.

\section*{Competing interests}
The authors have no competing interests to declare that are relevant to the content of this article.

\bigskip

\end{document}